\documentclass[11pt,a4paper]{article}
\pdfoutput=1

\usepackage[nosort]{cite}
\usepackage[bulletsep]{collref} 

\usepackage[british]{babel} 

\usepackage{epsfig,rotating}
\usepackage{amsmath,amssymb,amsthm}
\usepackage{amsfonts}
\usepackage{mathrsfs}
\usepackage{bbm}
\usepackage{bm}
\usepackage[normalem]{ulem}
\usepackage[all,knot]{xy}
\xyoption{arc}

\usepackage{color}
\usepackage[textwidth = 430 pt, textheight = 630 pt]{geometry}
\definecolor{MyDarkBlue}{rgb}{0.15,0.25,0.45}
 
\usepackage{enumerate}
 
\usepackage[linktocpage=true,hypertexnames=false]{hyperref}
\hypersetup{
colorlinks=true,
citecolor=MyDarkBlue,
linkcolor=MyDarkBlue,
urlcolor=MyDarkBlue,
pdfauthor={Branislav Jurčo, Christian S\"amann, Martin Wolf},
pdftitle={Semistrict Higher Gauge Theory},
pdfsubject={hep-th}
breaklinks=true
}

\let\fn\footnote
\renewcommand{\footnote}[1]{\linespread{1.1}\fn{#1}\linespread{1.29}}

\makeatletter
\newcommand{\xRightarrow}[2][]{\ext@arrow 0359\Rightarrowfill@{#1}{#2}}
\renewcommand{\section}{\@startsection
{section}{1}{\z@}{-3.5ex plus -1ex minus
    -.2ex}{2.3ex plus .2ex}{\bf\mathversion{bold} }}
\renewcommand{\subsection}{\@startsection{subsection}{2}{\z@}{-3.25ex
plus -1ex minus
   -.2ex}{1.5ex plus .2ex}{\bf\mathversion{bold} }}
\renewcommand{\subsubsection}{\@startsection{subsubsection}{3}{-2.45ex}{-3.25ex
\makeatother
plus -1ex minus -.2ex}{1.5ex plus .2ex}{\it }}
\renewcommand{\thesection}{\arabic{section}}
\renewcommand{\thesubsection}{\arabic{section}.\arabic{subsection}}
\renewcommand{\@seccntformat}[1]{\@nameuse{the#1}.~~}

\renewcommand{\theequation}{\thesection.\arabic{equation}}
\makeatletter \@addtoreset{equation}{section}

\renewcommand*\l@section{\@dottedtocline{1}{0em}{2em}}
\renewcommand*\l@subsection{\@dottedtocline{2}{2em}{2.4em}}
\renewcommand*\l@subsubsection{\@dottedtocline{4}{3.8em}{3.7em}}

\renewcommand\tableofcontents{%
    \section*{\large\contentsname
        \@mkboth{%
          \MakeUppercase\contentsname}{\MakeUppercase\contentsname}}%
       {\baselineskip=15pt plus 2pt minus 1pt
    \@starttoc{toc}}%
}

\renewenvironment{thebibliography}[1]
     {\baselineskip=16pt plus 2pt minus 1pt
      \section*{\large\refname
        \@mkboth{\MakeUppercase\refname}{\MakeUppercase\refname}}%
     \list{\@biblabel{\@arabic\c@enumiv}}%
           {\settowidth\labelwidth{\@biblabel{#1}}%
            \leftmargin\labelwidth
            \advance\leftmargin\labelsep
            \@openbib@code
            \usecounter{enumiv}%
            \let\p@enumiv\@empty
            \renewcommand\theenumiv{\@arabic\c@enumiv}}%
      \sloppy
      \clubpenalty4000
      \@clubpenalty \clubpenalty
      \widowpenalty4000%
      \sfcode`\.\@m
 \catcode`\^^M=10%
}

\newcommand{\appendices}{
\section*{Appendix}\label{appendices}\setcounter{subsection}{0}
\addcontentsline{toc}{section}{Appendix}
\setcounter{equation}{0}
\setcounter{thm}{0}
\makeatletter
\renewcommand{\theequation}{\Alph{subsection}.\arabic{equation}}
\renewcommand{\thesubsection}{\Alph{subsection}}
\renewcommand{\thethm}{\Alph{subsection}.\arabic{thm}}
\@addtoreset{equation}{subsection}
\@addtoreset{thm}{subsection}
\makeatother
}

\newtheorem{thm}{Theorem}[section]
\renewcommand{\thethm}{\thesection.\arabic{thm}}
\newtheorem{lemma}[thm]{Lemma}
\newtheorem{definition}[thm]{Definition}
\newtheorem{theorem}[thm]{Theorem}
\newtheorem{prop}[thm]{Proposition}
\newtheorem{cor}[thm]{Corollary}
\newtheorem{rem}[thm]{Remark}
\newtheorem{exam}[thm]{Example}

\def\periodb#1{\setbox0=\hbox{$#1$}#1\hskip-\wd0\hbox to\wd0{-}}

\newcommand{\unit}{\mathbbm{1}}   			
\newcommand{\id}{\mathrm{id}}   			

\newcommand{\CCB}{\mathscr{B}}

\newcommand{\CCC}{\mathscr{C}}

\newcommand{\CCL}{\mathscr{L}}

\newcommand{\CF}{\mathcal{F}}

\newcommand{\CG}{\mathcal{G}}
\newcommand{\CCG}{\mathscr{G}}

\newcommand{\CI}{\mathcal{I}}

\newcommand{\CCK}{\mathscr{K}}

\newcommand{\CN}{\mathcal{N}}
\newcommand{\CO}{\mathcal{O}}

\newcommand{\CS}{\mathcal{S}}

\newcommand{\frg}{\mathfrak{g}}				
\newcommand{\frh}{\mathfrak{h}}				
\newcommand{\frv}{\mathfrak{v}}				
\newcommand{\frw}{\mathfrak{w}}				
\newcommand{\fra}{\mathfrak{a}}				
\renewcommand{\frm}{\mathfrak{m}}			
\newcommand{\frn}{\mathfrak{n}}

\newcommand{\frU}{\mathfrak{U}}

\newcommand{\FR}{\mathbbm{R}}     			
\newcommand{\FC}{\mathbbm{C}}     			

\newcommand{\NN}{\mathbbm{N}}     			
\newcommand{\RZ}{\mathbbm{Z}}     			
\newcommand{\PP}{{\mathbbm{P}}}    			

\newcommand{\dd}{\mathrm{d}}     			
\newcommand{\dpar}{\partial}     			
\newcommand{\embd}{{\hookrightarrow}}     		
\newcommand{\diag}{{\mathrm{diag}}}     		
\newcommand{\eps}{{\varepsilon}}			

\newcommand{\eand}{{~~~\mbox{and}~~~}}     		
\newcommand{\eon}{{~~\mbox{on}~~}}     		
\newcommand{\ewith}{{~~~\mbox{with}~~~}}
\newcommand{\efor}{{~~~\mbox{for}~~~}}

\newcommand{\der}[1]{\frac{\dpar}{\dpar #1}}   		
\newcommand{\dder}[1]{\frac{\dd}{\dd #1}}   		

\newcommand{\au}{\mathfrak{u}}

\newcommand{\sU}{\mathsf{U}}     			

\newcommand{\sSU}{\mathsf{SU}}

\newcommand{\sSO}{\mathsf{SO}}
\newcommand{\sLie}{\mathsf{Lie}}

\newcommand{\sCE}{\mathsf{CE}}

\newcommand{\sG}{\mathsf{G}}

\newcommand{\sH}{\mathsf{H}}

\newcommand{\sL}{\sfL}

\newcommand{\sSp}{\mathsf{Sp}}
\newcommand{\sOSp}{\mathsf{OSp}}

\newcommand{\acton}{\vartriangleright}     			
\newcommand{\remark}[1]{}     				
\newcommand{\myxymatrix}[1]{\vcenter{\vbox{\xymatrix{#1}}}}
\def\tyng(#1){\hbox{\tiny$\yng(#1)$}}			
\def\tyoung(#1){\hbox{\tiny$\young(#1)$}}			

\newcommand{\sft}{\mathsf{t}}

\newcommand{\sfa}{\mathsf{a}}

\newcommand{\sfr}{\mathsf{r}}
\newcommand{\sfl}{\mathsf{l}}

\newcommand{\sfs}{\mathsf{s}}
\newcommand{\sfe}{\mathsf{e}}
\newcommand{\sfi}{\mathsf{i}}

\newcommand{\sfL}{\mathsf{L}}

\newcommand{\sfS}{\mathsf{S}}
\newcommand{\sfJ}{\mathsf{J}}

\newcommand{\sB}{\mathsf{B}}

\newcommand{\CatSet}{\mathsf{Set}}
\newcommand{\CatDiff}{\mathsf{Diff}}
\newcommand{\CatDiffCat}{\mathsf{DiffCat}}

\begin{document}
\begin{titlepage}

\setcounter{page}{0}
\renewcommand{\thefootnote}{\fnsymbol{footnote}}

\begin{flushright}
 EMPG--14--06\\ DMUS--MP--14/02
\end{flushright}

\begin{center}

{\LARGE\textbf{\mathversion{bold}Semistrict Higher Gauge Theory}\par}

\vspace{1cm}

{\large
Branislav Jur\v co$^{a}$, Christian S\"amann$^{b}$, and Martin Wolf$^{\,c}$
\footnote{{\it E-mail addresses:\/}
\href{mailto:branislav.jurco@gmail.com}{\ttfamily branislav.jurco@gmail.com},
\href{mailto:c.saemann@hw.ac.uk}{\ttfamily c.saemann@hw.ac.uk}, 
\href{mailto:m.wolf@surrey.ac.uk}{\ttfamily m.wolf@surrey.ac.uk}
}}

\vspace{.5cm}

{\it
$^a$ 
Charles University in Prague\\
Faculty of Mathematics and Physics, Mathematical Institute\\
Prague 186 75, Czech Republic\\[.5cm]

$^b$ Maxwell Institute for Mathematical Sciences\\
Department of Mathematics,
Heriot--Watt University\\
Edinburgh EH14 4AS, United Kingdom\\[.5cm]

$^c$
Department of Mathematics,
University of Surrey\\
Guildford GU2 7XH, United Kingdom\\[.5cm]

}

\vspace{.5cm}

{\bf Abstract}
\end{center}
\vspace{-.3cm}
\begin{quote}
We develop semistrict higher gauge theory from first principles. In particular, we describe the differential Deligne cohomology underlying semistrict principal 2-bun\-dles with connective structures. Principal 2-bundles are obtained in terms of weak 2-functors from the \v Cech groupoid to weak Lie 2-groups. As is demonstrated, some of these Lie 2-groups can be differentiated to semistrict Lie 2-algebras by a method due to \v Severa. We further derive the full description of connective structures on semistrict principal 2-bundles including the non-linear gauge transformations. As an application, we use a twistor construction to derive superconformal constraint equations in six dimensions for a non-Abelian $\CN=(2,0)$ tensor multiplet taking values in a semistrict Lie 2-algebra.

\vfill\noindent 29th April 2015

\end{quote}

\setcounter{footnote}{0}\renewcommand{\thefootnote}{\arabic{thefootnote}}

\end{titlepage}

\tableofcontents

\bigskip
\bigskip
\hrule
\bigskip
\bigskip

\section{Introduction, summary, and outlook}

\subsection{Motivation}

Gauge theory is one of the most far-reaching concepts in modern theoretical physics as is exemplified by the impressive success of the standard model of elementary particles as well as many of the more recent developments in string theory such as the gauge/gravity correspondence. From a mathematical point of view, the kinematic data of classical gauge theory is described in terms of principal bundles with connection. Equivalence relations on this data, known as gauge transformations, are captured by a non-Abelian generalisation of the so-called first Abelian Deligne cohomology group. 

By now, there is a well-established way of categorifying gauge theory to what is known as higher gauge theory. Here, the kinematic data lives on non-Abelian gerbes \cite{Breen:math0106083,Aschieri:2003mw} or the more general principal 2-bundles of Bartels \cite{Bartels:2004aa}; higher categorifications leading to $p$-gerbes or principal $(p+1)$-bundles are also known, though explicit details are somewhat limited. The notion of a connection on principal bundles is generalised to so-called connective structures on principal 2-bundles. This is a well established approach albeit with one limitation: instead of featuring the most general, weak Lie 2-group as structure 2-group, the standard formulations employ so-called crossed modules of Lie groups, which are equivalent to strict Lie 2-groups. 

The main aims of this paper are to lift this limitation and to discuss in full detail principal 2-bundles with connective structures that have semistrict Lie 2-groups as structure 2-groups. This involves considerably more technical effort than the strict case, and we would therefore like to give ample motivation for our goal.

The most general notion of a categorified group or 2-group which we shall consider here is what is usually called a weak 2-group. Just as a group is a groupoid with a single object, a weak 2-group is a weak 2-groupoid with a single object. As shown by Baez \& Lauda \cite{Baez:0307200}, every weak 2-group can be enhanced to a coherent 2-group, and, furthermore, coherent 2-groups are categorically equivalent to strict 2-groups. Categorical equivalence, however, seems to be too coarse in many cases. Perhaps a prime example in this regard is the categorified operation of integrating a Lie 2-algebra: it is known that the string Lie 2-algebra cannot be integrated to a topological 2-group \cite{Baez:0307200}. This semistrict Lie 2-algebra, however, is categorically equivalent to an infinite dimensional strict Lie 2-algebra, which can be integrated to a strict Lie 2-group \cite{Baez:2005sn}. Similarly, it is natural to expect that dynamical models of connective structures on principal 2-bundles will not necessarily agree even if the underlying structure 2-groups are categorically equivalent. 

Our motivation for considering categorified differential geometry stems mostly from M-theory. Within M-theory, principal 2-bundles with connective structures arise quite naturally in a non-Abelian generalisation of the effective description of M5-branes. In particular, they capture the kinematic structure of the mysterious $\CN=(2,0)$ superconformal field theory in six dimensions, or $(2,0)$-theory for short. 

The existence of the $(2,0)$-theory has been shown by Witten \cite{Witten:1995zh} a long time ago. However, it remains unclear if this theory should have a classical description in terms of equations of motion or even a Lagrangian. Quite recently, there has been impressive success in the effective description of multiple M2-branes. Contrary to popular belief, it turned out that there are M2-brane models with a Lagrangian formulation, which pass many non-trivial consistency checks, see \cite{Bagger:2012jb} for a review. Spurred by this success, various directions of research have been pursued to try to arrive at an analogous classical description for multiple M5-branes. In fact, much of the current research activities in string theory is devoted to a more detailed understanding of the $(2,0)$-theory.

Since the Abelian tensor multiplet in six dimensions contains a 2-form gauge potential described by a $\sU(1)$-gerbe, it is only natural to expect that the non-Abelian case is described by the connective structure of a principal 2-bundle. Principal 2-bundles with connective structures allow for the parallel transport of one-dimensional objects, which is certainly relevant in the description of the self-dual strings that form the boundaries of the M2-branes mediating M5-brane interactions. A detailed explanation of the higher gauge theory approach to M5-branes can be found in Fiorenza, Sati \& Schreiber \cite{Fiorenza:2012tb}. 

Besides its mathematical appeal, an important argument for the use of higher gauge theory is that principal 2-bundles can indeed yield superspace constraint equations for the $\CN=(2,0)$ tensor multiplet in six dimensions. This was shown recently in Saemann \& Wolf \cite{Saemann:2012uq,Saemann:2013pca}, and the derivation of these equations involved a description of the tensor multiplet in terms of certain holomorphic principal 2- and 3-bundles over a twistor space. Interestingly, such a twistorial description might also yield a Lagrangian formulation of the theory, as was already demonstrated for the Abelian case in \cite{Saemann:2011nb,Mason:2011nw}. 

The constraint equations arising from a twistorial description starting from principal 2-bundles with strict structure 2-groups turned out to be somewhat restrictive. A first reason for considering semistrict principal 2-bundles is therefore to generalise the superconformal constraint equations arising from a twistor description of the (2,0)-theory and we shall present the outcome in Section \ref{sec:PenroseWard}. In particular, we shall see that semistrict principal 2-bundles will allow for incorporating cubic terms in the connection 1-form in the definition of the 3-form curvature.

Another popular approach to deriving a classical description of the $(2,0)$-theory is based on a non-Abelian generalisation of the tensor hierarchy \cite{Samtleben:2011fj,Samtleben:2012mi,Samtleben:2012fb,Bandos:2013jva,Ko:2013dka} with the closely related proposals of \cite{Chu:2011fd,Chu:2012um}. Here, one obtains $\CN=(1,0)$ superconformal equations of motion as well as a Lagrangian description. These (1,0)-models have an underlying gauge algebraic structure which is strongly reminiscent of a semistrict Lie 3-algebra. The detailed analysis of this algebraic structures in \cite{Palmer:2013pka} showed that there is indeed a large overlap. Moreover, it was shown that certain classes of (1,0)-models are reformulations of higher gauge theories with strict Lie 3-groups. To fully compare the (1,0)-models with higher gauge theory, however, it is indispensable  to develop a detailed description of gauge theory based on semistrict principal $n$-bundles. This is a second motivation for studying semistrict principal 2-bundles.

Further motivation for our study stems from the problem of differentiating semistrict Lie 2-groups to semistrict Lie 2-algebras. While there has been some effort to understand the integration of Lie 2-algebras to Lie 2-groups, see for example Getzler \cite{Getzler:0404003} and Henriques \cite{Henriques:2006aa}, the inverse operation does not seem to have attracted the same amount of attention. In the present work, we shall follow a general approach to this problem that was proposed by \v Severa \cite{Severa:2006aa}. In this construction, one considers a simplicial manifold and extracts a corresponding $L_\infty$-algebra as its first jet. A Lie 2-group can be encoded in terms of a simplicial manifold as the so-called Duskin nerve of its delooping. The first jet of this simplicial manifold is then constructed as a functor acting on descent data of a trivial principal 2-bundle. 

Finally, we would like to mention that a different proposal for semistrict higher gauge theory was given previously by Zucchini \cite{Zucchini:2011aa}. In this approach, the higher Maurer--Cartan forms are incorporated abstractly as constrained parameters into the gauge transformation. This is not the case in our approach; our detailed understanding of the differential cohomology underlying semistrict principal 2-bundles with connective structures makes the parameters of gauge transformations explicit.

\subsection{Summary of results}

For the reader's convenience, let us summarise our key results in an easily accessible way. In the following, we let $X$ be a smooth manifold with covering $\frU:=\{U_a\}$. Moreover, we let $\CCG=(M,N)$ be a weak Lie 2-group, which can be equivalently regarded as a smooth weak 2-groupoid with a single 0-cell $e$, $\sB\CCG=(\{e\},M,N)$. We denote the source and target maps by $\sfs$ and $\sft$. Vertical and horizontal composition in this weak 2-groupoid are denoted by $\circ$ and $\otimes$, respectively, $\sfa$ stands for the associator and $\sfl$ and $\sfr$ label the left- and right-unitors.

A weak principal 2-bundle is described by a $\CCG$-valued \v Cech 2-cocycle. Such a cocycle is given by an $M$-valued \v Cech 1-cochain $\{m_{ab}\}$ together with an $N$-valued \v Cech 0-cochain $\{n_{a}\}$ and an $N$-valued \v Cech 2-cochain $\{n_{abc}\}$ which satisfy the following cocycle conditions, cf.\ Definition \ref{def:principal-2-bundle}:
\begin{subequations}
\begin{equation}
\begin{aligned}
 n_{abc}\,:\, m_{ab}\otimes m_{bc}\ &\Rightarrow\ m_{ac}~,\\
   n_{acd}\circ(n_{abc}\otimes\id_{m_{cd}})\ &=\ 
   n_{abd}\circ (\id_{m_{ab}}\otimes n_{bcd})\circ \sfa_{m_{ab},m_{bc},m_{cd}}
\end{aligned}
\end{equation}
and
\begin{equation}
   n_{abb}\circ (\id_{m_{ab}}\otimes n_b)\ =\ \sfl_{m_{ab}}\eand
            n_{aab}\circ (n_a\otimes\id_{m_{ab}})\ =\   \sfr_{m_{ab}}~.
\end{equation}
\end{subequations}
Two weak principal 2-bundles are called equivalent whenever their degree-2 \v Cech cocycles are related by a $\CCG$-valued \v Cech 2-coboundary. This coboundary consists of an $M$-valued \v Cech 0-cochain $\{m_a\}$  and an $N$-valued \v Cech 1-cochain $\{n_{ab}\}$ such that for degree-2 \v Cech cocycles $(\{m_{ab}\},\{n_{abc}\},\{ n_a\})$ and $(\{\tilde m_{ab}\},\{\tilde n_{abc}\},\{\tilde n_a\})$ the following holds, cf.\ Definition \ref{eq:isomorphism-principal-2-bundles}:
\begin{subequations}\label{eq:sum_coboundaries}
\begin{equation}
 \begin{aligned}
    n_{ab}\, :\, m_{ab}\otimes m_b\ &\Rightarrow\ m_a\otimes \tilde m_{ab}~,\\
  n_{ac}\circ (n_{abc}\otimes\id_{m_c})\ &=\ (\id_{m_a}\otimes\tilde n_{abc})\circ\sfa_{m_a,\tilde m_{ab},\tilde m_{bc}}\circ(n_{ab}\otimes\id_{\tilde m_{bc}})\,\circ \\
  &\hspace{2cm}\circ\, \sfa^{-1}_{m_{ab},m_b,\tilde m_{bc}}\circ(\id_{m_{ab}}\otimes n_{bc})\circ\sfa_{m_{ab},m_{bc},m_c}~.
 \end{aligned}
\end{equation}
and
\begin{equation}
 n_{aa}\circ (n_a\otimes \id_{m_a})\ =\ (\id_{m_a}\otimes \tilde n_a)\circ \sfl^{-1}_{m_a}\circ \sfr_{m_a}~.
\end{equation}
\end{subequations}
As demonstrated in Proposition \ref{prop:Normalisation}, every $\CCG$-valued \v Cech 2-cocycle is equivalent to a $\CCG$-valued \v Cech 2-cocycle with all $\{n_a\}$ being trivial.

Furthermore, we define semistrict Lie 2-groups $\CCG$ as weak Lie 2-groups in which left- and right-unitors as well as the unit and counit are all trivial. Following \cite{Severa:2006aa}, we then consider a functor from the category of smooth manifolds to the category of $\CCG$-valued descent data on surjective submersions $\FR^{0|1}\times X\rightarrow X$. This functor is parameterised by a 2-term $L_\infty$-algebra as shown in Theorem \ref{th:differentiate-lie-2}. This 2-term $L_\infty$-algebra is, in turn, equivalent to the semistrict Lie 2-algebra associated with the semistrict Lie 2-group $\CCG$. Deriving the parametrisation of this functor is the higher equivalent of computing the Lie algebra of a Lie group.

Moreover, we demonstrate that local connective structures on principal 2-bundles with semistrict structure 2-group (as well as principal $n$-bundles with semistrict structure $n$-group) are readily derived. To this end, we consider the tensor product of the aforementioned 2-term $L_\infty$-algebra  with the differential graded algebra of differential forms on $X$. This leads to another $L_\infty$-algebra as well as its homotopy Maurer--Cartan equation including infinitesimal gauge transformations as shown in Propositions \ref{prop:hMC-eqn-and-gauge} and \ref{prop:semistrict-gauge-trafos}.

The finite gauge transformations are derived from an equivalence relation among the functors considered in the above differentiation of a Lie 2-group $\CCG=(M,N)$ to a 2-term $L_\infty$-algebra  $\frv\xrightarrow{\,\mu_1\,} \frw$ with $\frw:=T_{\id_e}M$ and $\frv:=\ker(\sft)\subseteq T_{\id_{\id_e}}N$ and higher or homotopy products $\mu_{1,2,3}$. This relation is presented in Theorem \ref{th:SemStrGaugeTrafo}, from which Proposition \ref{prop:semistrict-gauge-trafos} can be gleaned: a connective structure over $U_a\subseteq X$ on a semistrict principal 2-bundle is given locally on a patch $U_a$ in terms of a $\frw$-valued differential 1-form $A_a$ and a $\frv$-valued differential 2-form $B_a$ such that the fake 2-form curvature
\begin{equation}
 \CF_a\ :=\ \dd A_{a}+A_{a}\otimes A_{a}+\sfs(B_{a})
\end{equation}
vanishes. In addition, the curvature 3-form $H_a$ is defined by
\begin{equation}
 H_{a}\  :=\ \dd B_{a}+\id_{A_{a}}\otimes B_{a}-B_a\otimes\id_{A_a}+\mu(A_{a},A_{a},A_{a})~,
\end{equation}
where $\mu(A_a,A_a,A_a):A_a\otimes (A_a\otimes A_a)-(A_a\otimes A_a)\otimes A_a\Rightarrow 0$. Finite gauge transformations $(A_a,B_a)\mapsto (\tilde A_a,\tilde B_a)$ are then parameterised by $M$-valued functions $p_a$ and $T_{p_a}N$-valued 1-forms $\Lambda_{p_a}$ and read explicitly as
\begin{subequations}
\begin{eqnarray}
 \kern-20pt\Lambda_{p_a}\,:\,\tilde A_a\otimes p_a\! &\Rightarrow&\! p_a\otimes A_a-\dd p_a~,\\
 \kern-20pt  \tilde B_a\otimes \id_{p_a}\! &=&\! \mu(\tilde A_a,\tilde A_a,p_a)+ \big[\id_{p_a}\otimes B_a+\mu(p_a,A_a,A_a)\big]\circ\notag\\
    &&\kern2.5cm \circ\, \big[-\dd\Lambda_{p_a}-\Lambda_{p_a}\otimes\id_{A_a}-\mu(\tilde A_a,p_a,A_a)\big]\circ\notag\\
      &&\kern2.5cm \circ\,\big[-\id_{\sfs(\dd \Lambda_{p_a})}-\id_{\tilde A_a}\otimes(\Lambda_{p_a}+\id_{\dd p_a})\big]~.
\end{eqnarray}
\end{subequations}

Eventually, we combine our findings on \v Cech cohomology with values in a semistrict Lie 2-group with those on finite gauge transformations of local connective structures and develop full  semistrict Deligne cohomology of degree 2. The corresponding Deligne cocycle and coboundary relations are concisely listed in Definitions \ref{def:Deligne-cocycles} and \ref{def:Deligne-coboundaries}.

As a first application of our results, we employ semistrict Deligne cohomology of degree 2 in a twistor description of $\CN=(2,0)$ tensor multiplet equations in six dimensions. This is a generalisation of the previous results obtained in \cite{Saemann:2012uq,Saemann:2013pca} from strict to semistrict gauge 2-groups. The main result here is Theorem \ref{th:Penrose-Ward} in which a bijection is established  between equivalence classes of certain holomorphic semistrict principal 2-bundles over a twistor space and equivalence classes of solutions to certain superconformal tensor multiplet equations in six dimensions. We hope that the latter equations may serve as an inspiration for a classical formulation of the $(2,0)$-theory.

\pagebreak

\subsection{Outlook}

There are a number of questions arising from this paper that we plan to address in future work. First of all, there should be an integration operation, inverse to our differentiation of a Lie 2-group to a semistrict Lie 2-algebra. An obvious question is how this integration is related to that of Getzler \cite{Getzler:0404003} and Henriques \cite{Henriques:2006aa}. The answer seems to be similar to that found in \cite{Sheng:1109.4002} for the strict case. Here, straightforward Lie integration of a strict Lie 2-algebra led to a Lie 2-group which is Morita equivalent to the 2-group obtained by the method of Getzler and Henriques. 

As mentioned above, we hope that the detailed description of semistrict principal 2-bundles with connective structure allows for a more detailed understanding of the framework of higher gauge theory. More general theories than those derived in this present work can be considered so that the relation to alternative approaches such as the above-mentioned non-Abelian tensor hierarchies should become clearer.

The most interesting dynamical theories involving connective structures on semistrict principal 2-bundles are certainly the $(2,0)$-theory and its dimensional reductions. As is common to supersymmetric theories, particular attention should be paid to the BPS subsectors of this theory. Higher analogues of instantons and monopoles, such as, for example, self-dual strings, should be studied in more detail from a mathematical perspective. Especially, the relevant topological invariants should be analysed. Some preliminary comments in this direction were already given in \cite{Palmer:2013haa}. General considerations concerning topological invariants in higher gauge theory can be found in \cite{Fiorenza:2010mh} as well as in \cite{Kotov:2007nr} from the perspective of so-called $Q$-manifolds.

An important issue is to couple matter fields satisfyingly to higher gauge theories. Mathematically speaking, we would like to consider 2-vector bundles associated to our semistrict principal 2-bundles. Zucchini has suggested such a coupling in his approach to semistrict gauge theory \cite{Zucchini:2011aa}. However, the existence of so-called gauge rectifiers necessary in his approach could not be proved so far. Our twistor construction gives illuminating insights into how such couplings should be achieved. In particular, our approach yields the explicit example of the matter fields contained in the tensor multiplet, discusses the properties they satisfy, how gauge transformations act on them, and how they couple to connective structures. 

The most important consistency test for a classical (2,0)-theory is to reproduce five-dimensional maximally supersymmetric Yang--Mills theory in a certain limit. This is a requirement from string theory and so far, this has neither been achieved for higher gauge theories nor for the models arising from tensor hierarchies.

\section*{Acknowledgements}
 We are grateful to James Grant for helpful discussions and comments. The work of BJ was supported by grant GA\v CR P201/12/G028. The work of CS was supported by a Career Acceleration Fellowship from the UK Engineering and Physical Sciences Research Council.

\section{Preliminaries}\label{sec:Prelims}

In this paper, we require basics of weak 2-category theory. We shall try to be as self-contained as possible and therefore we present the relevant definitions together with some useful examples in this section.

\subsection{Weak 2-categories}\label{ssec:Weak2Categories}

We assume that the reader is familiar with elementary category theory. In the following, let $\CCC=(C_0,C_1)$ be a category with $C_0$ the objects of $\CCC$ and $C_1$ the morphisms of $\CCC$, respectively.  In addition, the source and target maps in $\CCC$ are denoted by $\sfs$ and $\sft$, that is, $\sfs,\sft:C_1\to C_0$. 

In higher category theory, there is always an issue concerning the level of strictness of the categorification under consideration. For example, 2-categories usually refer to strict 2-categories while weak 2-categories are often called bicategories. We shall exclusively use the terms weak 2-category, weak 2-groupoid etc.~and avoid the notions of bicategory, bigroupoid etc. 

We start off with the definition of a weak 2-category. The original definition stems from Benabou \cite{Benabou:1967:1}, and a good introduction to the topic can be found, for instance, in \cite{Leinster:1998aa} and in particular in \cite{Leinster:2003aa}. The following discussion follows mostly these references. 

\begin{definition} (Benabou \cite{Benabou:1967:1})
A \uline{weak 2-category} $\CCB=(B_0,B_1,B_2)$ consists of a collection  $B_0$ of objects $a,b,\ldots\in B_0$ and, for any pair of objects $a,b\in B_0$, an assignment $(a,b)\to \CCC(a,b)$ where $\CCC(a,b)=(C_0(a,b),C_1(a,b))$ is a category. The objects $B_0$ are called \uline{0-cells}, the objects $C_0(a,b)$ are called \uline{1-cells} or \uline{1-morphisms}, and the morphisms $C_1(a,b)$ are called \uline{2-cells} or \uline{2-morphisms}. Composition of 2-morphisms in $C_1(a,b)$ will be called \uline{vertical composition} and denoted by $\circ$.
 
In addition, $\CCB$ comes equipped with a bifunctor $\otimes:\CCC(a,b)\times \CCC(b,c)\rightarrow \CCC(a,c)$  for all $a,b,c\in B_0$ describing \uline{horizontal composition} in $\CCB$, a functor\footnote{Here, the $1$ is the terminal object in the category $\mathsf{Cat}$, that is, the singleton category consisting of one object $e$ and the corresponding morphism $\id_e$.} $\id:1 \mapsto \id_a \in C_0(a,a)$ for all $a\in B_0$, and natural isomorphisms $\sfa$, $\sfl$, and $\sfr$ defined by the following diagrams:
\begin{subequations}
\begin{equation}
\xymatrixcolsep{5pc}
\myxymatrix{
\CCC(a,b)\times \CCC(b,c)\times \CCC(c,d) \ar@{->}[r]^{\otimes\times 1} \ar@{->}[d]_{1\times\otimes} & \CCC(a,c)\times \CCC(c,d) \ar@{->}[d]^{\otimes}\\
\CCC(a,b)\times \CCC(b,d) \ar@{->}[r]_{\otimes} \ar@{=>}[ur]_{\sfa}& \CCC(a,d)
}
\end{equation}
and
\begin{equation}
\xymatrixcolsep{2pc}
\myxymatrix{
\CCC(a,b)\times 1 \ar@{->}[rrdd]^{\cong} \ar@{->}[dd]_{1 \times \id} & &\\
&&\\
\CCC(a,b)\times \CCC(b,b) \ar@{->}[rr]_{\otimes} \ar@{=>}[ur(0.8)]^{\sfl}& & \CCC(a,b)
}~~~~
\myxymatrix{
1\times\CCC(a,b) \ar@{->}[rrdd]^{\cong} \ar@{->}[dd]_{\id\times 1} & &\\
&&\\
\CCC(a,a)\times \CCC(a,b) \ar@{->}[rr]_{\otimes} \ar@{=>}[ur(0.8)]^{\sfr}& & \CCC(a,b)
}
\end{equation}
\end{subequations}
Here, the $1$ attached to the arrows refers to the identity functor and $\cong$ denotes the natural isomorphisms $1\times \CCC(a,b)\cong\CCC(a,b)\cong\CCC(a,b)\times 1$. The natural isomorphisms $\sfa$, $\sfl$, and $\sfr$ are referred to as the \uline{associator}, \uline{left unitor}, and \uline{right unitor}, and they yield the 2-cells 
\begin{equation}
 \sfa\,:\,(x\otimes y)\otimes z\ \stackrel{\cong}{\Rightarrow}\ x\otimes(y\otimes z)~,~~~\sfl\,:\,x\otimes \id_b\ \stackrel{\cong}{\Rightarrow}\ x~,~~~ \sfr\,:\,\id_a\otimes x\ \stackrel{\cong}{\Rightarrow}\ x
\end{equation}
for $x\in C_0(a,b)$, $y\in C_0(b,c)$, and $z\in C_0(c,d)$. These isomorphisms are required to satisfy the \uline{pentagon} and \uline{triangle identities}, that is, the diagrams
\begin{subequations}
\begin{equation}
\myxymatrix{
((x\otimes y)\otimes u)\otimes v \ar@{=>}[rrrr]^{\sfa\otimes \id}  \ar@{=>}[d]_{\sfa} & &&& (x\otimes(y\otimes u))\otimes v \ar@{=>}[d]^{\sfa}\\
(x\otimes y)\otimes (u\otimes v) \ar@{=>}[rr]^{\sfa}& & x\otimes(y\otimes(u\otimes v))&  &\ar@{=>}[ll]_{\id\otimes \sfa}  x\otimes((y\otimes u)\otimes v) 
}
\end{equation}
and
\begin{equation}
\xymatrixcolsep{4pc}\myxymatrix{
(x\otimes \id_b)\otimes y \ar@{=>}[rr]^{\sfa}  \ar@{=>}[dr]_{\sfl\otimes \id} & & x\otimes(\id_b\otimes y) \ar@{=>}[ld]^{\id\otimes \sfr}\\
& x\otimes y
}
\end{equation}
\end{subequations}
are commutative.
\end{definition}

\begin{rem}
The fact that $\otimes$ is a bifunctor implies the so-called \uline{interchange law}, that is, the diagram
\begin{equation}
 \myxymatrix{
  a 
&& b
  \ar@/_4ex/[ll]_{x_1}="g1"
  \ar[ll]_(0.65){x_2}
  \ar@{}[ll]|{}="g2"
  \ar@/^4ex/[ll]^{x_3}="g3"
  \ar@{=>}^{f_1} "g1"+<0ex,-2ex>;"g2"+<0ex,0.5ex>
  \ar@{=>}^{f_2} "g2"+<0ex,-0.5ex>;"g3"+<0ex,2ex>
&& c
  \ar@/_4ex/[ll]_{y_1}="h1"
  \ar[ll]_(0.65){y_2}
  \ar@{}[ll]|{}="h2"
  \ar@/^4ex/[ll]^{y_3}="h3"
  \ar@{=>}^{g_1} "h1"+<0ex,-2ex>;"h2"+<0ex,0.5ex>
  \ar@{=>}^{g_2} "h2"+<0ex,-0.5ex>;"h3"+<0ex,2ex>
}
\end{equation}
for $x_{1,2,3}\in C_0(a,b)$, $y_{1,2,3}\in C_0(b,c)$ and $f_{1,2}\in C_1(a,b)$, $g_{1,2}\in C_1(b,c)$ and $a,b,c\in B_0$ translates into 
\begin{equation}\label{eq:interchange_law}
 (f_2\otimes g_2)\circ (f_1\otimes g_1)\ =\ (f_2\circ f_1)\otimes (g_2\circ g_1)~,
\end{equation}
where $\circ$ denotes again vertical composition.
\end{rem}

\begin{rem}
The naturalness of the associator $\sfa$ implies that diagrams of the form
\begin{equation}
\xymatrixcolsep{4pc}\myxymatrix{
 (x\otimes y)\otimes z  \ar@{->}[rr]^{(f\otimes g)\otimes h} \ar@{->}[d]_{\sfa} & & (f(x)\otimes g(y))\otimes h(z)  \ar@{->}[d]^{\sfa}\\
x\otimes (y\otimes z) \ar@{->}[rr]_{f\otimes (g\otimes h)} & & f(x)\otimes (g(y)\otimes h(z))
}
\end{equation}
are commutative. There are similar commutative diagrams involving the unitors or a combination of the unitors and the associator.
\end{rem}

\begin{definition}
A \uline{strict 2-category} is a weak 2-category for which the associator and the left- and right-unitors are all trivial.
\end{definition}
 
\begin{exam}
The standard example of a strict 2-category is $\mathsf{Cat}$, regarded as a 2-category, in which the 0-cells are given by small categories, the 1-cells are functors between those, and the 2-cells are natural transformations between the latter. Horizontal composition is then the obvious composition of functors and natural transformations.
\end{exam}

\begin{definition}
A weak 2-category with a single 0-cell can be identified with a \uline{weak monoidal category}. If, in addition, the natural isomorphisms $\sfa,\sfl$, and $\sfr$ are all trivial, then we shall speak of a \uline{strict monoidal category}.  
\end{definition}

The process of identifying $n$-categories with a single object or 0-cell with $(n-1)$-categories is called \uline{looping}. Below, we shall also encounter the inverse operation called delooping, see Example \ref{exa:Delooping}. 
 
\begin{exam}
An example of a strict monoidal category is the category of sets endowed with a monoidal product given either by the Cartesian product or the disjoint union of sets. Here, $B_0=\{e\}$ and $\CCC(e,e)$ is the category $\CatSet$ whose objects $C_0$ are sets and whose morphisms $C_1$ are functions between sets.
\end{exam}

In weak 2-categories with a single 0-cell $e$, that is, in weak monoidal categories, we have the following result.
\begin{prop} (Kelly \cite{Kelly:1964:397})
In a weak monoidal category $\CCB$, the diagrams
\begin{subequations}
\begin{equation}
\xymatrixcolsep{4pc}\myxymatrix{
(x\otimes y)\otimes \id_e \ar@{->}[rr]^{\sfa}  \ar@{->}[dr]_{\sfl} & & x\otimes(y\otimes \id_e) \ar@{->}[ld]^{\id\otimes \sfl}\\
& x\otimes y
}
\end{equation}

\begin{equation}
\xymatrixcolsep{4pc}\myxymatrix{
(\id_e \otimes x)\otimes y \ar@{->}[rr]^{\sfa}  \ar@{->}[dr]_{\sfr\otimes \id} & & \id_e\otimes(x\otimes y) \ar@{->}[ld]^{\sfr}\\
& x\otimes y
}
\end{equation}
\begin{equation}
\xymatrixcolsep{4pc}\myxymatrix{
\id_e\otimes\id_e\ar@/^/[r]^{\sfl} \ar@/_/[r]_{\sfr} & \id_e   }
\end{equation}
\end{subequations}
are commutative.
\end{prop}

\vspace{10pt}

Morphisms between categories are called functors. Similarly, morphisms between 2-categories are called 2-functors. These come in a number of variants, the most general of which are the lax 2-functors.

\begin{definition}
Let $\CCB$ and $\tilde \CCB$ be two weak 2-categories. A \uline{lax 2-functor} $\Phi:\CCB\to\tilde\CCB$ is a triple $\Phi=(\Phi_0,\Phi_1,\Phi_2)$ consisting of a function $\Phi_0:B_0\rightarrow \tilde B_0$, a collection $\Phi_1$ of functors 
\begin{subequations}
\begin{equation}
\Phi_1^{ab}\,:\,\CCC(a,b)\ \rightarrow\ \tilde \CCC(\Phi_0(a),\Phi_0(b))~,
\end{equation}
and a collection $\Phi_2$ of 2-cells,
\begin{equation}
\begin{aligned}
\Phi_2^{abc}\,:\,\Phi_1^{ab}(x)\,\tilde \otimes\, \Phi_1^{bc}(y)\ &\Rightarrow\ \Phi_1^{ac}(x\otimes y)~,\\
  \Phi_2^{a}\,:\,\id_{\Phi_0(a)}\ &\Rightarrow\ \Phi_1^{aa}(\id_a)~,
\end{aligned}
\end{equation}
\end{subequations}
where  $a,b,c\in B_0$ and $x\in C_0(a,b)$ and $y\in C_0(b,c)$ such that the following diagrams are commutative:
\begin{subequations}\label{eq:diagrams_2-functors}
\begin{equation}
\xymatrixcolsep{4pc}
\myxymatrix{
& \Phi_1^{ac}(x\otimes y)\,\tilde \otimes\, \Phi_1^{cd}(z) \ar@{=>}[dr]^{\Phi_2^{acd}} & \\
 (\Phi_1^{ab}(x)\,\tilde \otimes\,\Phi_1^{bc}(y))\,\tilde \otimes\, \Phi_1^{cd}(z) \ar@{=>}[ur]^{\Phi_2^{abc}\otimes \id}  \ar@{=>}[d]_{\tilde \sfa} & & \Phi_1^{ad}((x\otimes y)\otimes z) \ar@{=>}[d]^{\Phi_1^{ad}(\sfa)}\\
 \Phi_1^{ab}(x)\,\tilde \otimes\,(\Phi_1^{bc}(y)\,\tilde \otimes\, \Phi_1^{cd}(z))\ar@{=>}[dr]_{\id\otimes \Phi_2^{bcd}} &  & \Phi_1^{ad}(x\otimes(y\otimes z))\\
 & \Phi_1^{ab}(x)\,\tilde \otimes\,\Phi_1^{bd}(y\otimes z)\ar@{=>}[ur]_{\Phi_2^{abd}} & }
\end{equation}
and
\begin{equation}
\xymatrixcolsep{4pc}
\myxymatrix{
&  \Phi_1^{ab}(x)\,\tilde \otimes\, \Phi_1^{bb}(\id_b) \ar@{=>}[dr]^{\Phi_2^{abb}}  & \\
 \Phi_1^{ab}(x)\,\tilde \otimes\, \id_{\Phi_0(b)} \ar@{=>}[ur]^{\id\otimes \Phi_2^{b}}  \ar@{=>}[dr]_{\tilde\sfr} & 
 &  \Phi_1^{ab}(x\otimes \id_b) \ar@{=>}[ld]^{\Phi_1^{ab}(\sfr)}\\
&  \Phi_1^{ab}(x)  &\\
 \id_{\Phi_0(a)}\,\tilde \otimes\, \Phi_1^{ab}(x) \ar@{=>}[dr]_{\Phi_2^{a}\otimes \id}  \ar@{=>}[ur]^{\tilde\sfl} & 
&  \Phi_1^{ab}(\id_a\otimes x) \ar@{=>}[lu]_{\Phi_1^{ab}(\sfl)}\\
&\Phi_1^{aa}(\id_a)\,\tilde \otimes\, \Phi_1^{ab}(x) \ar@{=>}[ur]_{\Phi_2^{aab}}  &
}
\end{equation}
\end{subequations}
\end{definition}

\begin{definition}\label{def:weak-2-functor}
A lax 2-functor $\Phi=(\Phi_0,\Phi_1,\Phi_2)$ for which the 2-cells $\Phi_2$ are natural isomorphisms is called a \uline{weak 2-functor}.\footnote{Weak 2-functors are also known as pseudo-functors.} A lax 2-functor $\Phi=(\Phi_0,\Phi_1,\Phi_2)$ for which the 2-cells $\Phi_2$ are trivial is called a \uline{strict 2-functor}. 
\end{definition}

\begin{rem}
Given two lax 2-functors $\Phi=(\Phi_0,\Phi_1,\Phi_2):\CCB\rightarrow \tilde \CCB$ and $\Psi=(\Psi_0,\Psi_1,\Psi_2):\tilde \CCB\rightarrow \hat \CCB$, their composition $\Phi\circ \Psi$ yields another lax 2-functor $\Xi=(\Xi_0,\Xi_1,\Xi_2)$ with
\begin{equation}
\begin{aligned}
 \Xi_0&\ =\ \Psi_0\circ \Phi_0\,:\,B_0\rightarrow \hat B_0~,\\
 \Xi_1&\ =\ \Psi^{\tilde{a}\tilde{b}}_1\circ \Phi^{ab}_1\,:\,\CCC(a,b)\rightarrow \hat \CCC(\Xi_0(a),\Xi_0(b))~,\\
 \Xi_2^{abc}&\ =\ \Psi^{\tilde a\tilde b}_1(\Phi^{abc}_2)\circ \Psi^{\tilde a\tilde b\tilde c}_2\,:\,\Xi_1^{ab}(x)\,\tilde \otimes\, \Xi_1^{bc}(y)\ \Rightarrow\ \Xi_1^{ac}(x\otimes y)~,\\
  \Xi_2^{a}&\ = \ \Psi^{\tilde a\tilde a}_1(\Phi^{a}_2)\circ \Psi^{\tilde a}_2\,:\,\id_{\Xi_0(a)}\ \Rightarrow\ \Xi_1^{aa}(\id_a)~,
\end{aligned}
\end{equation}
where $a,b,c\in B_0$ and $\tilde a=\Phi_0(a)$ etc.
\end{rem}

As expected, there are also generalisations of the notion of natural transformation to the case of weak 2-categories. Because we shall need these natural 2-transformation when defining coboundary conditions, we shall introduce them now in full detail.

\begin{definition} 
Let $\Phi,\Psi:\CCB\rightarrow \tilde\CCB$ be two lax 2-functors between two weak 2-categories $\CCB$ and $\tilde \CCB$. A \uline{lax natural 2-transformation} $\alpha:\Phi\Rightarrow\Psi$ with $\alpha=(\alpha_1,\alpha_2)$ consists of a family of 1-cells $\alpha_1^a:\Phi_0(a)\rightarrow \Psi_0(a)$ for each $a\in B_0$ together with a family of 2-cells $\alpha_2^{ab}$ defined by
\begin{equation}
 \xymatrixcolsep{5pc}
\myxymatrix{
 \Phi_0(b) \ar@{->}[r]^{\Phi_1^{ab}(x)}  \ar@{->}[d]_{\alpha_1^b} & \Phi_0(a) \ar@{->}[d]^{\alpha_1^a}\\
 \Psi_0(b) \ar@{->}[r]_{\Psi_1^{ab}(x)} \ar@{=>}[ru]^{\alpha_2^{ab}(x)}  & \Psi_0(a)
}
\end{equation}
for each 1-cell $x\in C_0(a,b)$ in $\CCB$, such that for all $x\in C_0(a,b)$, $y\in C_0(b,c)$ and $a,b,c\in B_0$ the diagrams
\begin{subequations}
\begin{equation}\label{eq:L2NT-A}
\xymatrixcolsep{3.5pc}\myxymatrix{
  \Psi_1^{ab}(x)\,\tilde \otimes\,  (\alpha_1^b\,\tilde \otimes\,  \Phi_1^{bc}(y))  \ar@{=>}[r]^{\tilde \sfa^{-1}}  &  (\Psi_1^{ab}(x)\,\tilde \otimes\,  \alpha_1^b)\,\tilde \otimes\,  \Phi_1^{bc}(y)  \ar@{=>}[r]^{\alpha_2^{ab}\,\tilde \otimes\, \id}  & (\alpha_1^a\,\tilde \otimes\, \Phi_1^{ab}(x))\,\tilde \otimes\, \Phi_1^{bc}(y)  \ar@{=>}[d]^{\tilde\sfa}\\
 \Psi_1^{ab}(x)\,\tilde \otimes\, (\Psi_1^{bc}(y)\,\tilde \otimes\,  \alpha_1^c)  \ar@{=>}[u]^{\id\,\tilde \otimes\,  \alpha_2^{bc}} & &   \alpha_1^a\,\tilde \otimes\, (\Phi_1^{ab}(x)\,\tilde \otimes\, \Phi_1^{bc}(y))  \ar@{=>}[d]^{\id\,\tilde \otimes\,  \Phi_2^{abc}}\\
(\Psi_1^{ab}(x)\,\tilde \otimes\, \Psi_1^{bc}(y))\,\tilde \otimes\,  \alpha_1^c   \ar@{=>}[r]_{~~~\Psi_2^{abc}\,\tilde \otimes\, \id}  
 \ar@{=>}[u]^{\tilde \sfa}  & \Psi_1^{ac}(x\otimes y)\,\tilde \otimes\,  \alpha_1^c  \ar@{=>}[r]_{\alpha_2^{ac}} & \alpha_1^a\,\tilde \otimes\,  \Phi_1^{ac}(x\otimes y)
 }
\end{equation}
and
\begin{equation}\label{eq:L2NT-B}
\myxymatrix{
  \id_{\Psi_0(a)}\otimes \alpha_1^a \ar@{=>}[r]^{~~~~~\sfr} \ar@{=>}[d]_{\Psi_2^a\otimes\id} & \alpha_1^a \ar@{=>}[r]^{\sfl^{-1}~~~~~} & \alpha_1^a\otimes \id_{\Phi_0(a)}  \ar@{=>}[d]^{\id\otimes\Phi_2^{a}}  \\
  \Psi_1^{aa}(\id_a)\otimes \alpha_1^a  \ar@{=>}[rr]_{\alpha_2^{aa}}  & & \alpha_1^a\otimes\Phi_1^{aa}(\id_a)
  }
\end{equation}
\end{subequations}
are commutative.
\end{definition}

\begin{definition}
A lax natural 2-transformation $\alpha=(\alpha_1,\alpha_2)$ for which the 2-cells $\alpha_2$ are natural isomorphisms is called a \uline{weak natural 2-transformation}.\footnote{Weak natural 2-transformations are also known as pseudo-natural transformations.} A lax natural 2-transformation $\alpha=(\alpha_1,\alpha_2)$ for which the 2-cells $\alpha_2$ are trivial is called a \uline{strict natural} \uline{2-transformation}.
\end{definition}

The composition of natural 2-transformations is governed by the following proposition. 
 
 \begin{prop}\label{prop:Combine2Lax2NatTrans}
 Given three lax 2-functors $\Phi,\Psi,\Xi:\CCB\to\tilde\CCB$ between two weak 2-categories $\CCB$ and $\tilde\CCB$ and two lax natural 2-transformations $\alpha:\Phi\Rightarrow\Psi$ and $\beta:\Psi\Rightarrow\Xi$, then there is a lax natural 2-transformation $\gamma:\Phi\Rightarrow\Xi$ such that
 \begin{subequations}
 \begin{equation}\label{eq:CompL2NT-A}
 \xymatrixcolsep{5pc}
\myxymatrix{
 \Phi_0(b) \ar@{->}[r]^{\Phi_1^{ab}(x)}  \ar@{->}[d]_{\alpha_1^b} & \Phi_0(a) \ar@{->}[d]^{\alpha_1^a}\\
 \Psi_0(b) \ar@{->}[r]^{\Psi_1^{ab}(x)} \ar@{=>}[ru]^{\alpha_2^{ab}(x)} \ar@{->}[d]_{\beta_1^b}  & \Psi_0(a)  \ar@{->}[d]^{\beta_1^a}\\ 
  \Xi_0(b) \ar@{->}[r]^{\Xi_1^{ab}(x)} \ar@{=>}[ru]^{\beta_2^{ab}(x)} & \Xi_0(a)  
}
\ =\ \xymatrixcolsep{5pc}
\myxymatrix{
 \Phi_0(b) \ar@{->}[r]^{\Phi_1^{ab}(x)}  \ar@{->}[d]_{\gamma_1^b} & \Phi_0(a) \ar@{->}[d]^{\gamma_1^a}\\
 \Xi_0(b) \ar@{->}[r]^{\Xi_1^{ab}(x)} \ar@{=>}[ru]^{\gamma_2^{ab}(x)}  & \Xi_0(a)
}
\end{equation}
with $\gamma_1^a:\Phi_0(a)\to\Xi_0(a)$ and $\gamma_2^{ab}:  \Xi_1^{ab}(x)\,\tilde \otimes\, \gamma_1^b \Rightarrow \gamma_1^a\,\tilde\otimes\, \Phi_1^{ab}(x)$ and 
\begin{equation}\label{eq:CompL2NT-B}
\begin{aligned}
 \gamma_1^a\ &=\ \beta_1^a\,\tilde\otimes\,\alpha_1^a~,\\
 \gamma_2^{ab}\ &=\ \tilde\sfa_{\beta_1^a,\alpha_1^a,\Phi^{ab}(x)}^{-1}\,\tilde\circ\,(\id_{\beta_1^a}\,\tilde\otimes\,\alpha_2^{ab}(x))\,\tilde\circ\,\tilde\sfa_{\beta_1^a,\Psi^{ab}(x),\alpha_1^b}\,\tilde\circ\,(\beta_2^{ab}(x)\,\tilde\otimes\,\id_{\alpha_1^b})\,\tilde\circ\,\tilde\sfa^{-1}_{\Xi^{ab}(x),\beta_1^b,\alpha_1^b}~
 \end{aligned}
\end{equation}
\end{subequations}
for all $a,b\in B_0$ and $x\in C_0(a,b)$.
\end{prop}
 
\vspace{5pt}
\noindent
{\it Proof:} It is straightforward to see that $\gamma=(\gamma_1,\gamma_2)$ given in \eqref{eq:CompL2NT-B} is a  map $\gamma_1^a:\Phi_0(a)\to\Xi_0(a)$ and $\gamma_2^{ab}:  \Xi_1^{ab}(x)\,\tilde \otimes\, \gamma_1^b \Rightarrow \gamma_1^a\,\tilde\otimes\, \Phi_1^{ab}(x)$ between the lax 2-functors $\Phi$ and $\Xi$. That this is indeed a lax natural 2-transformation is a consequence of the pasting theorem for weak 2-categories, see Verity \cite{Verity:2011aa}. \hfill $\Box$ 

\vspace{10pt}
Finally, for 2-categories, it is useful to continue the sequence of 2-categories, 2-functors, 2-transformations to 2-modifications.
 
 \begin{definition}
 Let $\Phi,\Psi:\CCB\to\tilde\CCB$ be two lax 2-functors between two weak 2-categories $\CCB$ and $\tilde\CCB$.  A \uline{2-modification} between two lax natural 2-transformations $\alpha,\beta:\Phi\to\Psi$ is a collection of morphisms ${\varphi_a:\alpha^a_1\Rightarrow \beta^a_1}$ for each $a\in B_0$ such that
\begin{equation}
 \xymatrixcolsep{5pc}
\myxymatrix{
 \Psi_1^{ab}(x)\,\tilde\otimes\, \alpha_1^b \ar@{=>}[r]^{\id \,\tilde\otimes\, \varphi_b}  \ar@{=>}[d]_{\alpha_2^{ab}} & \Psi_1^{ab}(x)\,\tilde\otimes\, \beta_1^b \ar@{=>}[d]^{\beta_2^{ab}}\\
 \alpha_1^a\,\tilde\otimes\, \Phi^{ab}_1(x) \ar@{=>}[r]_{\varphi_a \,\tilde\otimes\, \id} & \beta_1^a\,\tilde\otimes\, \Phi_1^{ab}(x)
}
\end{equation}
is commutative. If the morphisms ${\varphi_a}$ are invertible, we call the 2-modification invertible.
\end{definition}

\noindent
Note that composition of 2-modifications is trivially obtained by concatenation.

\subsection{Weak 2-groupoids}

In this section, we would like to introduce the notion of 2-groupoids as they play key roles in the definition of principal 2-bundles. We begin by recalling the definition of a groupoid first.

\begin{definition}
 A \uline{groupoid} is a small category in which every morphism is invertible. 
 \end{definition}
 
\noindent
Two important examples of groupoids that we shall frequently encounter throughout this work are those of the \v Cech groupoid and the delooping of a group.  

\begin{exam}\label{exa:CechGroupoid} The \uline{\v{C}ech groupoid} relative to a covering $\frU:=\{U_a\}$ of a topological manifold $X$, denoted by $\check\CCC(\frU)$ in the following, is defined to be the groupoid that has the covering sets as objects and the intersection of covering sets as morphisms. Concretely, the set of objects of $\check\CCC(\frU)$ is defined to be the disjoint union $\dot\bigcup_{a} U_a:=\bigcup_{a} \{(x,a)\,|\,x\in U_a\}$ and the set of morphisms of $\check\CCC(\frU)$ is defined to be the disjoint union $\dot\bigcup_{a,b} U_a\cap U_b:=\bigcup_{a,b}\{(x,a,b)\,|\,x\in U_a\cap U_b\}$, together with the structure maps
\begin{equation}
\begin{aligned}
 \sfs(x,a,b)\ :=\ (x,b)~,\quad \sft(x,a,b)\ :=\ (x,a)~,\quad \id_{(x,a)}\ :=\ (x,a,a)~,\\
  (x,a,b)\circ (x,b,c)\ :=\ (x,a,c)~.\kern2.6cm
\end{aligned}
\end{equation}
\end{exam}

\begin{exam} \label{exa:Delooping} Let $\sG$ be a group. The \uline{delooping} of  $\sG$, denoted by $\sB\sG$,  is defined to be the groupoid that has only a single object, denoted by $e$, and the elements of the group $\sG$ as its morphisms, $g:e\to e$ with $g\in\sG$. In  $\sB \sG$, the composition of morphisms is then simply given by the group multiplication on $\sG$, that is, $g_2\circ g_1:=g_2g_1$ for any $g_{1,2}\in \sG$.
\end{exam}

We are interested in the categorification of the notion of a groupoid, which is defined as follows.
 
 \begin{definition} 
A \uline{weak 2-groupoid} is a weak 2-category such that all morphisms are equivalences. A weak 2-groupoid with an underlying strict 2-category is a called a \uline{strict 2-groupoid}.
\end{definition}

\noindent
All morphisms being equivalences implies that the 2-cells are strictly invertible and the 1-cells are invertible up to isomorphisms. Unpacking this definition further\footnote{cf.\ Hardie et al.\ \cite{Hardie:2001aa}}, a weak 2-groupoid is a weak 2-category $\CCB$ such that for every pair of objects $a,b\in B_0$, the category $\CCC(a,b)$ is a groupoid. Moreover, for every pair $a,b\in B_0$ there is a functor $\bar{\cdot}:\CCC(a,b)\rightarrow \CCC(b,a)$ and for every 1-cell $x\in C_0(a,b)$ there are natural isomorphisms $\sfi_x:\id_a\Rightarrow x\otimes \bar x$ and $\sfe_x:\bar x \otimes x \Rightarrow \id_b$ called the \uline{unit} and \uline{counit}. These have to satisfy coherence axioms, which state that for any 1-cell $x\in C_0(a,b)$ and $a,b\in B_0$, the diagrams
\begin{subequations}\label{eq:ZigZag}
\begin{equation}
\myxymatrix{
(x\otimes \bar{x})\otimes x \ar@{=>}[rrrr]^{\sfa}  \ar@{=>}[d]_{\sfi^{-1}\otimes \id} & &&& x\otimes (\bar{x}\otimes x) \ar@{=>}[d]^{\id\otimes\sfe}\\
\id_a\otimes x \ar@{=>}[rr]^{\sfl}& & x&  &\ar@{=>}[ll]_{\sfr}  x\otimes\id_b
}
\end{equation}
and
\begin{equation}
\myxymatrix{
(\bar x\otimes x)\otimes \bar x \ar@{=>}[rrrr]^{\sfa}  \ar@{=>}[d]_{\sfe\otimes \id} & &&& \bar x\otimes ( x\otimes \bar x) \ar@{=>}[d]^{\id\otimes\sfi^{-1}}\\
\id_b\otimes \bar x\ar@{=>}[rr]^{\sfr}& & x&  &\ar@{=>}[ll]_{\sfl}  \bar x\otimes \id_a 
}
\end{equation}
\end{subequations}
are commutative.

\begin{exam}
An example of a strict 2-groupoid important in our subsequent discussion is the so-called \v Cech 2-groupoid. The 0- and 1-cells are given by the objects and morphisms of the \v Cech groupoid (see Example \ref{exa:CechGroupoid}), and all 2-cells defined to be trivial. 
\end{exam}

In Example \ref{exa:Delooping}, we have seen that any group can be viewed as a groupoid with a single object. Analogously, we give the following definition.

\begin{definition}\label{def:weak-2-group}
 A \uline{weak 2-group} $\CCG=(M,N)$ is the looping of a weak 2-groupoid with a single 0-cell $\CCB=(\{e\},M,N)$. 
\end{definition}

\begin{rem}
This definition is equivalent to that given by Baez \& Lauda \cite{Baez:0307200}. In particular, they define weak 2-groups as weak monoidal categories in which all morphisms are invertible and all objects are weakly invertible. They also introduce so-called \uline{coherent 2-groups} as weak monoidal categories in which all morphisms are invertible and all objects come with an adjoint equivalence. Both notions are shown to be equivalent. Our definition \ref{def:weak-2-group} uses the looping, as the weak 2-groups we are interested in will mostly appear as deloopings of coherent 2-groups in the sense of Baez \& Lauda. We shall therefore write $\sB\CCG=(\{e\},M,N)$: the single 0-cell is denoted by $e$ in the following while the 1- and 2-cells are denoted by $M$ and $N$, respectively. The (monoidal) category $\CCC(e,e)$ contained in $\sB\CCG$ is then the actual weak 2-group.
\end{rem}

\begin{definition}
 A \uline{strict 2-group} is the looping of a strict 2-groupoid with a single 0-cell.
\end{definition}

\noindent
Put differently, a strict 2-group is a weak 2-group in which the unitors, the unit and counit, and the associator are all trivial. Furthermore, we will need the notion of a skeletal 2-group which is as follows.

\begin{definition}
 A \uline{skeletal 2-group} is a weak 2-group, in which the underlying category is skeletal.
\end{definition}

\noindent
Recall that a category is skeletal whenever all isomorphic objects are equal: for all morphisms $f$ in the category, $\sfs(f)=\sft(f)$.

One version of Mac Lane's coherence theorem \cite{0387984038} states that every weak monoidal category is equivalent to a strict monoidal category. In the case of weak 2-groups, we have the following proposition from \cite[Sec.\ 8.3]{Baez:0307200}, which can be used to classify weak Lie 2-groups. 

\begin{prop}\label{prop:classify-Lie-2-groups}(Baez \& Lauda \cite{Baez:0307200})
Every weak 2-group is categorically equivalent to a `special' weak 2-group which is skeletal and in which all unitors, units, and counits are identity natural transformations.  In particular, a special weak 2-group can be given in terms of a group $\sG$, an Abelian group $\sH$, a representation $\alpha$ of $\sG$ as automorphisms of $\sH$ and an element $[\sfa]\in H^3(\sG,\sH)$.
\end{prop}

\noindent
In addition, we have the following result.

\begin{prop}\label{prop:strictify-Lie-2-groups}(Baez \& Lauda \cite{Baez:0307200})
Every weak 2-group is categorically equivalent to a strict 2-group.
\end{prop}

\noindent
The notion of 2-groups relevant for our subsequent discussion will be the following.

\begin{definition}
A \uline{semistrict 2-group} is a weak 2-group in which the unitors and the unit and counit are all trivial.
\end{definition}

\noindent
We would like to emphasise that this notion is weaker than that of a strict 2-group, because the associator remains unrestricted. For semistrict 2-groups, we have the following results.

\vspace{10pt}
\begin{prop}\label{prop:TrivAssocId}
 In the delooping of any semistrict 2-group $\sB\CCG=(\{e\},M,N)$, the associators $\sfa_{\id_e,m,m'}$, $\sfa_{m,m',\id_e}$, $\sfa_{m,\id_e,m'}$, $\sfa_{m,\overline{m},m}$, and $\sfa_{\overline{m}, m,\overline{m}}$ are trivial for all $m,m'\in M$.
\end{prop}

\vspace{10pt}
\noindent
{\it Proof:}
 This follows trivially by combining the pentagon and triangle diagrams together with the diagrams displayed in \eqref{eq:ZigZag}. \hfill $\Box$

\vspace{10pt}
\begin{prop}\label{prop:inverse-concatenation}
 In any semistrict 2-group $\CCG=(M,N)$ and for any 2-cell $n\in N$,
 \begin{equation}
   n^{-1}\ =\  \sfa_{\sfs(n),\overline{\sft(n)},\sft(n)} \circ ((\id_{\sfs(n)}\otimes \bar n)\otimes \id_{\sft(n)})\,:\, \sft(n)\ \Rightarrow\ \sfs(n)
 \end{equation}
 such that $n\circ n^{-1}=\id_{\sft(n)}$ and $n^{-1}\circ n=\id_{\sfs(n)}$.
\end{prop}

\vspace{10pt}
\noindent
{\it Proof:}
 This follows from the proof of Proposition 20 in \cite{Baez:0307200}. \hfill $\Box$

\subsection{Lie 2-groups}\label{ssec:Lie_2_groups}

To restrict the rather general notion of a groupoid, we can regard Lie groupoids as groupoids internal to a certain category $\CCC$. In general, a category internal to $\CCC=(C_0,C_1)$ consists of an object of objects and an object of morphisms, which are both elements in $C_0$. The structure maps $\sfs$, $\sft$, $\id$, and $\circ$ are given in terms of elements of $C_1$ and all commutative diagrams which hold in a category also hold in the internalised category. Internal functors and modifications are defined in an analogous manner. A groupoid internal to a category $\CCC$ is simply a category internal to $\CCC$, in which all the morphisms are strictly invertible.

In this manner, we can define, for instance, topological groupoids as groupoids in $\mathsf{Top}$, the category of topological spaces and continuous functions between them. Similarly, Lie groupoids are defined as follows.

\begin{definition}
 A \uline{Lie groupoid} is a groupoid internal to $\CatDiff$, the category of smooth manifolds and smooth functions between them.
\end{definition}

\noindent
Thus, Lie groupoids are groupoids in which the sets of objects and morphisms are smooth manifolds and all the structure maps are smooth.

\begin{rem}
 Recall that for any category $\CCK$ there exists a strict 2-category $\CCK\mathsf{Cat}$ with objects being categories internal to $\CCK$, morphisms being functors in $\CCK$ and 2-morphisms being natural transformations in $\CCK$. In particular, $\CatDiffCat$ is the strict 2-category with categories in $\CatDiff$ as 0-cells, functors between these as 1-cells and natural transformations between the latter as 2-cells.
\end{rem}

We can now define weak Lie 2-groupoids and weak Lie 2-groups by internalising weak 2-groupoids and weak 2-groups, respectively.

\begin{definition}\label{def:Lie-2-group}
 A \uline{weak Lie 2-groupoid} is a weak 2-groupoid internal to $\CatDiffCat$. A \uline{weak Lie 2-group} is a weak 2-group internal to $\CatDiffCat$.
\end{definition}

\noindent
Equivalently, a weak Lie 2-group is a weak Lie 2-groupoid with a single object. Specifically, such a weak Lie 2-group consists of an object $C$ in $\CatDiffCat$, a multiplication morphism $\otimes:C\times C\rightarrow C$, an identity object $\unit$, and an inverse map $\bar{\cdot}:C\rightarrow C$ with respect to $\otimes$. Furthermore, we have for all objects $x$, $y$, and $z$ in the category $C$ the following natural isomorphisms: an associator $\sfa:(x\otimes y)\otimes z\Rightarrow x\otimes (y\otimes z)$, left- and right-unitors $\sfl_x:\unit\otimes x\rightarrow x$ and $\sfr_x:x\otimes\unit\Rightarrow x$ as well as a unit and counit $\sfi_x:\unit\Rightarrow x\otimes\bar x$ and $\sfe_x:\bar x\otimes x\rightarrow \unit$, such that the pentagon and triangle identities as well as the first and second zig-zag identities are satisfied, cf.\ \cite{Baez:0307200}. 

For our purposes, we wish to restrict the notion of a weak Lie 2-group as given in Definition \ref{def:Lie-2-group} somewhat further.

\begin{definition}
 A \uline{semistrict Lie 2-group} is a weak 2-group internal to $\CatDiffCat$ such that the unitors, the unit, and the counit are all trivial.
\end{definition}
Note that by Proposition \ref{prop:classify-Lie-2-groups}, semistrict Lie 2-groups are still categorically equivalent to weak Lie 2-groups.

\begin{definition}
 A \uline{strict Lie 2-group} is a weak 2-group  in $\CatDiffCat$ such that the associator, the unitors, the unit, and the counit are all trivial.

\end{definition}

\noindent
We recall that there is an equivalent formulation of strict Lie 2-groups in terms of crossed modules of Lie groups.
\begin{definition}
A \uline{crossed module of Lie groups} is a pair of Lie groups $(\sH,\sG)$ together with a Lie group homomorphism\footnote{This homomorphism is often denoted by $\sft$. Here, however, to avoid confusion with the source and target maps of the weak 2-group, we use the symbol $\dpar$.} $\dpar:\sH\rightarrow \sG$ and an action $\acton$ of $\sG$ on $\sH$ by automorphisms. The map $\dpar$ is $\sG$-equivariant and satisfies the Peiffer identity,
\begin{equation}
 \dpar(g\acton h)\ =\ g\dpar(h)g^{-1}\eand \dpar(h_1)\acton h_2\ =\ h_1h_2h_1^{-1}
\end{equation}
for all $g\in \sG$ and $h,h_1,h_2\in \sH$.
\end{definition}

Then we have the following result.

\begin{prop}\label{prop:equiv-strict-Lie-2-CM} 
 A strict Lie 2-group is equivalent to a crossed module of Lie groups.
\end{prop}

See Baez \& Lauda \cite{Baez:0307200} for a detailed proof. We shall use an identification between strict Lie 2-groups and crossed modules of Lie groups that slightly differs from that of \cite{Baez:0307200}. Given a crossed module of Lie groups $(\sH\xrightarrow{\,\dpar\,} \sG,\acton)$, we obtain a strict Lie 2-group $\CCG=(M,N)$ by identifying $M:=\sG$ and $N:=\sG\ltimes \sH$ and setting $\sfs(g,h):=\dpar(h^{-1})g$, $\sft(g,h):=g$, and $\id_{g}=(g,\unit_\sH)$ for $h,h_{1,2}\in\sH$ and $g,g_{1,2}\in\sG$ together with
\begin{equation}\label{eq:CompositionsStrictLie2Group}
\begin{aligned}
 g_2\otimes  g_1\ &:=\ g_2g_1~,\\
  (g_2,h_2)\otimes (g_1,h_1)\ &:=\ (g_2g_1,(g_2\acton h_1)h_2)~,\\
  (g,h_2)\circ  (\dpar(h_2^{-1})g,h_1)\ &:=\ (g,h_2h_1)~.
\end{aligned}  
\end{equation}

On the other hand, given a strict Lie 2-group $\CCG=(M,N)$, we define a crossed module $(\sH\xrightarrow{\,\dpar\,}\sG,\acton)$ by putting $\sG:=M$ and $\sH:=\ker(\sft)$ and
\begin{equation}\label{eq:definition-CM-ops}
 \begin{aligned}
  g_2g_1\ & :=\ g_2\otimes g_1~,~~~&h_2 h_1\ & :=\ h_2\circ (h_1\otimes \id_{\sfs(h_2)})~,\\
   \dpar(h)\ & :=\ \sfs(h^{-1})~,~~~
  &g\acton h\ & :=\ \id_g\otimes h\otimes\id_{\overline g}~.
 \end{aligned}
\end{equation}

\subsection{Lie 2-algebras}

Apart from Lie 2-groups, we shall also be dealing with Lie 2-algebras. The most general kind of Lie 2-algebra currently in use has been defined by Roytenberg \cite{Roytenberg:0712.3461} as follows.

\begin{definition}
A \uline{weak Lie 2-algebra} is a linear category $\CCL=(L_0,L_1)$ equipped with 
  \begin{enumerate}[(i)]\setlength{\itemsep}{-1mm}
   \item a bilinear functor $[\cdot,\cdot]: \CCL\times \CCL \rightarrow \CCL$ called the \uline{bracket}, 
 \item a bilinear natural transformation $\sfS:[X,Y]\Rightarrow -[Y,X]$ called the  \uline{alternator}, and
 \item a trilinear natural transformation $\sfJ:[X,[Y,Z]]\Rightarrow [[X,Y],Z]+[Y,[X,Z]]$ called the \uline{Jacobiator}
 \end{enumerate}
for all $X,Y,Z\in L_0$. These structure maps are subject to a number of coherence axioms, cf.\ \cite{Roytenberg:0712.3461}.
\end{definition}

In this paper, we are merely interested in so-called semistrict Lie 2-algebras.
\begin{definition}
 A \uline{semistrict Lie 2-algebra} is a weak Lie 2-algebra in which the alternator is trivial.
\end{definition}

\noindent
Instead of working directly with semistrict Lie 2-algebras and their rather involved coherence axioms, we can switch to a categorically equivalent formulation in terms of 2-term $L_\infty$-algebras, as was shown in \cite{Baez:2003aa}. The general definition of a strong homotopy Lie algebra is given in appendix \ref{app:A}. Here, we just recall the following definition.

\begin{definition}\label{def:2-term-SH-algebra}
 A \uline{2-term $L_\infty$-algebra} consists of a 2-term complex of vector spaces $\frv$ and $\frw$,
\begin{equation}
  \frv~\overset{\mu_1}{\longrightarrow}~\frw~\overset{\mu_2}{\longrightarrow}~0~,
\end{equation}
where we associate gradings $-1$ and $0$ to elements of $\frv$ and $\frw$, respectively. This complex is equipped with higher products $\mu_1$, $\mu_2$, $\mu_3$, which vanish except for
\begin{equation}
\begin{aligned}
 \mu_1\,:\,\frv\ \rightarrow\ \frw~,~~~\mu_2\,:\,\frw\wedge \frw\ \rightarrow\ \frw~,~~~\mu_2\,:\, \frv\wedge \frw\ \rightarrow\ \frv~,\\ \mu_3\,:\,\frw\wedge \frw\wedge \frw\ \rightarrow\ \frv~.\kern3.2cm
 \end{aligned}
\end{equation}
Moreover, these products are required to satisfy the following \uline{higher homotopy Jacobi identities}:
\begin{equation}\label{eq:homotopy_relations}
\begin{aligned}
 \mu_1(\mu_2(w,v))\ &=\ \mu_2(w,\mu_1(v))~,\\
 \mu_2(\mu_1(v_1),v_2)\ &=\ \mu_2(v_1,\mu_1(v_2))~,\\
 \mu_1(\mu_3(w_1,w_2,w_3))\ &=\ -\mu_2(\mu_2(w_1,w_2),w_3)-\mu_2(\mu_2(w_3,w_1),w_2)-\mu_2(\mu_2(w_2,w_3),w_1)~,\\
 \mu_3(\mu_1(v),w_1,w_2)\ &=\ -\mu_2(\mu_2(w_1,w_2),v)-\mu_2(\mu_2(v,w_1),w_2)-\mu_2(\mu_2(w_2,v),w_1)~,\\
 \mu_2(\mu_3(w_1,w_2,w_3),w_4)&-\mu_2(\mu_3(w_4,w_1,w_2),w_3)+\mu_2(\mu_3(w_3,w_4,w_1),w_2)\,-\\
  &-\mu_2(\mu_3(w_2,w_3,w_4),w_1)\ =\\
 &\kern-2.5cm=\ \mu_3(\mu_2(w_1,w_2),w_3,w_4)-\mu_3(\mu_2(w_2,w_3),w_4,w_1)+\mu_3(\mu_2(w_3,w_4),w_1,w_2)\,-\\
 &\kern-2.5cm -\mu_3(\mu_2(w_4,w_1),w_2,w_3)
 -\mu_3(\mu_2(w_1,w_3),w_2,w_4)-\mu_3(\mu_2(w_2,w_4),w_1,w_3)~,
\end{aligned}
\end{equation}
where $v,v_i\in \frv$ and $w,w_i\in \frw$.
\end{definition}

\begin{rem}\label{rem:invert-L-infty}
 Note that for every 2-term $L_\infty$-algebra $\frv\xrightarrow{\,\mu_1\,} \frw$ with products $(\mu_1,\mu_2,\mu_3)$, there is another 2-term $L_\infty$-algebra with the same underlying vector spaces $\tilde \frv:=\frv$ and $\tilde \frw:=\frw$ but with higher products $\tilde \mu_1:=-\mu_1$, $\tilde \mu_2:=\mu_2$, and $\tilde \mu_3:=-\mu_3$.
\end{rem}

\begin{exam}
 A typical example of a semistrict Lie 2-algebra is the string Lie 2-algebra of a Lie algebra $\frg$. Here, $\frw=\frg$, $\frv=\FR$ and the only non-trivial higher products are $\mu_2(w_1,w_2)=[w_1,w_2]$ and $\mu_3(w_1,w_2,w_3)=\langle w_1,[w_2,w_3]\rangle$, where $w_1,w_2,w_3\in \frw$ and $\langle \cdot,\cdot\rangle$ is the Killing form on $\frg$.
\end{exam}

Let us briefly recall the details of the equivalence between semistrict Lie 2-algebras and 2-term $L_\infty$-algebras.\footnote{A similar equivalence exists for weak Lie 2-algebras \cite{Roytenberg:0712.3461}, but the resulting normalised chain complex is less convenient to work with.} We start from a Lie 2-algebra $\CCL=(L_0,L_1)$ and put 
\begin{equation}
 \frv\ :=\ \ker(\sft)\ \subseteq\ L_1~,~~~\frw\ :=\ L_0~,\eand\mu_1\ :=\ -\sfs|_{\frv}~.
\end{equation}
The higher products are defined as follows:
\begin{equation}
\begin{aligned}
\mu_2(w_1,w_2)\ :=\ [w_1,w_2]~,~~~\mu_2(w,v)\ =\ -\mu_2(v,w)\ :=\ [\id_w,v]~,\\
\mu_3(w_1,w_2,w_3)\ :=\ \sfJ(w_1,w_2,w_3)-\id_{[[w_1,w_2],w_3]+[w_2,[w_1,w_3]]}~,\kern1.2cm
\end{aligned}
\end{equation}
where $w_1,w_2,w_3,w\in \frw$ and $v\in \frv$. This map from a semistrict Lie 2-algebra to a 2-term $L_\infty$-algebra can be extended to a functor $\Phi$ between the corresponding categories.

Conversely, given a 2-term $L_\infty$-algebra $\frv\xrightarrow{\,\mu_1\,} \frw$, we obtain a semistrict Lie 2-algebra $\CCL=(L_0,L_1)$ by putting
\begin{equation}
\begin{aligned}
 L_0&\ :=\ \frw~,~~~L_1\ :=\ \frv\oplus \frw~,~~~\sfs(w,v)\ :=\ w-\mu_1(v)~,~~~\sft(w,v)\ :=\ w~,\\
 &\kern.5cm\id_w\ :=\ (w,0)~,~~~(w,v_2)\circ (w-\mu_1(v_2),v_1)\ :=\ (w,v_1+v_2)
\end{aligned}
\end{equation}
for  all $v,v_1,v_2\in \frv$ and $w\in \frw$. In addition, we set
\begin{equation}
\begin{aligned}
 {}[w_1,w_2]\ &:=\ \mu_2(w_1,w_2)~,\\
 [(w_1,v_1),(w_2,v_2)]\ &:=\ \big(\mu_2(w_1,w_2),\mu_2(v_1,w_2)+\mu_2(w_1-\mu_1(v_1),v_2)\big)~,\\
 \sfJ(w_1,w_2,w_3)\ &:=\ \big(-\mu_2(\mu_2(w_1,w_2),w_3)-\mu_2(\mu_2(w_3,w_1),w_2),\mu_3(w_1,w_2,w_3)\big)~.
\end{aligned}
\end{equation}
Again, this map from a 2-term $L_\infty$-algebra  to a semistrict Lie 2-algebra can be extended to a functor $\Psi$ between the corresponding categories. 

We have the following results.

\begin{prop} (Baez \& Crans \cite{Baez:2003aa})
Together, the functors $\Phi$ and $\Psi$ defined above can be shown to form an equivalence, which can even be extended to an equivalence of 2-categories.
\end{prop}

\begin{prop}\label{prop:classify-Lie-2-algebras} (Baez \& Crans \cite{Baez:2003aa})
There is a one-to-one correspondence between equivalence classes of semistrict Lie 2-algebras and `special' 2-term $L_\infty$-algebras  given in terms of a Lie algebra $\frg$, a representation of $\frg$ on a vector space $\frv$, and an element $J$ of $H^3(\frg,\frv)$. Here, $\mu_1=0$, $\mu_2$ is the Lie bracket in $\frg$ or the action on $\frv$, and $\mu_3=J$.
\end{prop}

Semistrict Lie 2-algebras can be restricted further to obtain strict Lie 2-algebras.

\begin{definition}
 A \uline{strict Lie 2-algebra} is a weak Lie 2-algebra with trivial alternator and trivial Jacobiator.
\end{definition}

Our above discussion immediately implies that strict Lie 2-algebras are equivalent to 2-term $L_\infty$-algebras with trivial product $\mu_3$, which in turn, can be encoded in a differential crossed module. 

\begin{definition}
 The \uline{differential crossed module} of a crossed module of Lie groups is obtained by applying the tangent functor to the crossed module.
\end{definition}

\noindent
In particular, given a crossed module of Lie groups $(\sH\xrightarrow{\,\dpar\,}\sG,\acton)$, the tangent functor yields a differential crossed module\footnote{Our notation does not distinguish between the maps $\dpar$, $\acton$ and their differentials.} $(\frh\xrightarrow{\,\dpar\,}\frg,\acton)$, where $\frh:=\sLie(\sH)$ and $\frg:=\sLie(\sG)$. The maps $\dpar$ and $\acton$ satisfy
\begin{equation}\label{eq:peiffer}
 \dpar(X\acton Y)\ =\ [X,\dpar(Y)]\eand \dpar(Y_1)\acton Y_2\ =\ [Y_1,Y_2]~,
\end{equation}
where $X\in\frg$ and $Y,Y_{1,2}\in \frh$.

The differential crossed module corresponding to a 2-term $L_\infty$-algebra $\frv\xrightarrow{\,\mu_1\,} \frw$ with trivial $\mu_3$ is obtained by identifying $\frh$, $\frg$, and $\dpar$ with $\frv$, $\frw$, and $\mu_1$ as well as 
\begin{equation}
 {}[w_1,w_2]\ :=\ \mu_2(w_1,w_2)~,~~~v\acton w\ :=\ \mu_2(v,w)\eand [v_1,v_2]\ :=\ \mu_2(\mu_1(v_1),v_2)
\end{equation}
for $v_1,v_2,v\in \frv=\frh$ and $w_1,w_2,w\in \frw=\frg$. This identification is readily inverted.

\section{Principal 2-bundles with Lie 2-groups}

We come now to the discussion of principal 2-bundles with weak structure 2-groups over smooth manifolds. An earlier description of general 2-bundles from a slightly different point of view can be found in Bartels \cite{Bartels:2004aa}. In the following, let $X$ be a smooth manifold and let $\frU=\{U_a\}$ be a covering of $X$.  

\subsection{Principal bundles as functors}

Recall that a \v Cech $p$-cochain with values in a group $\sG$ on $X$ relative to the covering $\frU$ is a set of smooth $\sG$-valued functions on all non-empty intersections $U_{a_0}\cap \cdots \cap U_{a_{p}}$.\footnote{If not stated otherwise, we shall always assume that intersections of patches are non-empty from now on.} We then give the following definition.

\begin{definition}
 A  \uline{\v Cech 1-cocycle} is a \v Cech 1-cochain $\{g_{ab}\}$ consisting of smooth maps  $g_{ab}:U_a\cap U_b\to\sG$  such that
\begin{equation}\label{eq:CocycConPB}
g_{ab}g_{bc}\ =\ g_{ac}\eon U_a\cap U_b\cap U_c~.
\end{equation}
 Two \v Cech 1-cocycles $\{g_{ab}\}$ and $\{\tilde g_{ab}\}$ are \uline{cohomologous} or \uline{equivalent} if and only if there is a \v Cech 0-cochain $\{g_a\}$ consisting of smooth maps $g_a:U_a\to\sG$  such that
\begin{equation}\label{eq:EquivCocycPB}
 g_{ab}\ =\ g_a \tilde g_{ab} g_b^{-1}~.
\end{equation}
 The \uline{first \v Cech cohomology set}, denoted by $H^1(\frU,\sG)$, is defined as the set of \v Cech 1-cocycles modulo this equivalence.
\end{definition}

\noindent
\v Cech cohomology sets can be rendered independent of the covering by taking the direct limit over all coverings $\frU$ of $X$. We then write $H^1(X,\sG)$ instead of $H^1(\frU,\sG)$, that is,
\begin{equation}
  H^1(X,\sG)\ =\ \varinjlim_\frU H^1(\frU,\sG)~.
\end{equation}

Elements of $H^1(X,\sG)$ are also known as (sets of) transition functions of principal bundles with structure group $\sG$ (or principal $\sG$-bundles for short), and it is well-known that principal $\sG$-bundles over $X$ can be identified with an elements in $H^1(X,\sG)$. To allow for a categorification of this picture, we switch to a functorial description of principal bundles.

\begin{definition}\label{def:principal-bundle}
 A smooth \uline{principal bundle} $\Phi$ with structure group $\sG$ is a smooth functor $\Phi$ from the \v Cech groupoid to the Lie groupoid $\sB \sG$.\footnote{See Examples \ref{exa:CechGroupoid} and \ref{exa:Delooping} for the relevant definitions.} Any two principal bundles are called \uline{equivalent} if and only if there is a natural isomorphism between their defining functors.
\end{definition}

Definition \ref{def:principal-bundle} is well-known from the description of principal bundles in terms of classifying spaces \cite{Segal1968}. Explicitly, we have a functor
\begin{equation}
  \Phi\,:\, \check\CCC(\frU)\ \to\ \sB\sG
\end{equation}
and we set $e_a:=\Phi(x,a)$ and $g_{ab}:=\Phi(x,a,b)$. Because $\Phi$ is a functor, we immediately arrive at the cocycle conditions \eqref{eq:CocycConPB} as well as $\Phi(x,a,a)=\id_{\Phi(x,a)}=\unit_\sG\in \sG$. In addition, two functors $\Phi$ and $\Psi$ corresponding to principal bundles are equivalent if and only if there is a natural isomorphism $\alpha:\tilde \Phi\to \Phi$. Defining $e_a:=\Phi(x,a)$, $g_{ab}=\Phi(x,a,b)$, and $g_a:=\alpha_{(x,a)}:\tilde\Phi(x,a)\to \Phi(x,a)$, the following diagram is commutative:
\begin{equation}\label{eq:NatTrafoPB}
\xymatrixcolsep{4pc}
\myxymatrix{
\tilde e_b
\ar[r]^{\tilde g_{ab}}
\ar[d]_{g_b}
& \tilde e_a 
\ar[d]^{g_a}
\\
e_b
\ar[r]^{g_{ab}}
& e_a
} 
\end{equation}
In formul{\ae}, this is
\begin{equation}\label{eq:comm_rel_nt}
g_a \tilde g_{ab}\ =\ g_{ab} g_b~,
\end{equation}
which amounts to \eqref{eq:EquivCocycPB}. We thus arrive at the following statement, which motivates our Definition \ref{def:principal-bundle}.

\begin{prop}
 Denoting the set of equivalence classes of smooth functors between $\check\CCC(\frU)$ and $\sB \sG$ by $[\check\CCC(\frU)\to\sB \sG]$, we have
 \begin{equation}
  H^1(\frU,\sG)\ \cong\ [\check\CCC(\frU)\to\sB\sG]~.
 \end{equation}
\end{prop}

Other conventional definitions are now also straightforwardly rephrased.
\begin{definition}
 A principal bundle is called \uline{trivial} if and only if its defining functor is equivalent to the functor
 \begin{equation}
  \Phi(x,a)\ =\ e_a\eand \Phi(x,a,b)\ =\ \unit_\sG~.
 \end{equation}
\end{definition}

\noindent
Concretely, a principal bundle is trivial whenever there is a natural isomorphism $\alpha=\{g_a\}$ such that
\begin{equation}
 g_a\ =\ g_{ab} g_b~.
\end{equation}

Finally, let  $\phi:X\rightarrow Y$ be a smooth map between two smooth manifolds $X$ and $Y$. Let $\frU_Y$  be a covering of $Y$. Then we can construct a covering $\frU_X$ of $X$ from the pre-images of the  patches in $\frU_Y$ under $\phi$. This yields a morphism of groupoids $\check\CCC(\frU_X)\rightarrow \check\CCC(\frU_Y)$.

\begin{definition}\label{def:pullback-1}
 The \uline{pullback} of a principal bundle $\Phi$ over $Y$ with respect to an open covering $\frU_Y$ along a smooth map $\phi:X\rightarrow Y$ is the composition $\Phi\circ \phi_{\frU}$, where $\phi_{\frU}$ is the groupoid morphism induced by $\phi$.
\end{definition}

\begin{definition}\label{def:restriction-1}
 The \uline{restriction} of a principal bundle $\Phi$ over a manifold $X$ to a submanifold $Y$ of $X$ is the pullback of $\Phi$ along the embedding map $Y\embd X$.
\end{definition}

\subsection{Principal 2-bundles as 2-functors}

The reformulation of principal bundles with structure group $\sG$ in terms of functors between the \v Cech groupoid and the Lie groupoid $\sB \sG$ is a good starting point for categorifying the notion of principal bundles. We can simply regard the \v Cech groupoid as an $n$-groupoid and take an $n$-functor to a Lie $n$-groupoid with a single 0-cell. In the following, we shall develop the case $n=2$ in detail. Note that our discussion will first centre around weak principal 2-bundles which we shall define in terms of weak 2-functors. In the following, we shall consider the delooping $\sB\CCG=(\{e\},M,N)$ of a weak Lie 2-group $\CCG=(M,N)$, which is a weak Lie 2-groupoid with a single object $e$. As in Section \ref{sec:Prelims}, we shall denote horizontal and vertical composition in $\sB\CCG$ by $\otimes$ and $\circ$, respectively. 

Principal 2-bundles will be described by \v Cech cocycles with values in $\CCG$. We therefore start by giving the following definition. 
\begin{definition}
 A \uline{degree-$p$ \v Cech cochain} with values in a weak Lie 2-group $\CCG=(M,N)$ consists of a smooth $M$-valued degree-$(p-1)$ \v Cech cochain $\{m_{a_0\cdots a_{p-1}}\}$, a smooth $N$-valued degree-$p$ \v Cech cochain $\{n_{a_0\cdots a_{p}}\}$, and a smooth $N$-valued degree-$(p-2)$ \v Cech cochain $\{n_{a_0\cdots a_{p-2}}\}$.
\end{definition}
In the following, we are interested in the case $p=2$, for which we have a triple
\begin{equation}
 (\{m_{ab}\},\{n_{abc}\},\{n_a\})~.
\end{equation}
These cochains generalise the usual \v Cech cochains appearing in the definition of an ordinary principal bundle in the way that is familiar from strict principal 2-bundles: the $\{m_{ab}\}$ are generalised transition functions on overlaps, the $\{n_{abc}\}$ are the gluing isomorphisms on triple overlaps and the $\{n_a\}$ are the isomorphisms between the unit in $M$ and the transition functions $\{m_{aa}\}$.

To derive the explicit cocycle and coboundary conditions appropriate for weak Lie 2-groups, we again employ the functorial approach.

\begin{definition}\label{def:principal-2-bundle}
A smooth \uline{weak principal 2-bundle} $\Phi$ with structure 2-group $\CCG$ relative to the covering $\frU$ is a smooth weak 2-functor $\Phi$ from the \v Cech 2-groupoid $\check\CCC(\frU)$ to $\sB\CCG$.
\end{definition}

\noindent
Let us be more specific. We have a weak 2-functor\footnote{cf.\ Definition \ref{def:weak-2-functor}}
\begin{equation}
  \Phi\,:\, \check\CCC(\frU)\ \to\ \sB\CCG
\end{equation}
consisting of a function $\Phi_0(x,a)$, functors $\Phi_1(x,a,b)$ and 2-cells $\Phi_2$. Note that the 0-cells of $\sB\CCG=(\{e\},M,N)$ and the 2-cells of $\check\CCC(\frU)$ are trivial and we shall denote them by $e$. We can therefore specify $\Phi$ in terms of constant functions $e_a:=\Phi_0(x,a):U_a\rightarrow e$, functions $m_{ab}:=\Phi_1(x,a,b)|_M: U_a\cap U_b\rightarrow M$, and constant functions $e_{ab}:=\Phi_1(x,a,b)|_N:e \rightarrow \id_{\id_e}$ together with invertible functions $n_{abc}:U_a\cap U_b\cap U_c\rightarrow N$ and $n_a:U_a\rightarrow N$ describing the 2-cell $\Phi_2$. Because $\id_{(x,a)}=(x,a,a)$, we have by definition $\Phi_1(\id_{(x,a)})=\Phi_1(x,a,a)=m_{aa}$. The fact that $\Phi_1$ is a functor implies $\id_{m_{ab}}=\id_{\Phi_1(x,a,b)}=\Phi_1(\id_{(x,a,b)})$. Finally, with $\Phi_1\big((x,a,b)\circ (x,b,c)\big)\ =\ \Phi_1(x,a,c)=m_{ac}$, we have the natural isomorphisms
\begin{equation}\label{eq:CocycConWeakPB-I}
\begin{aligned}
n_{abc}\,:\, m_{ab}\otimes m_{bc}\ &\Rightarrow\ m_{ac}~,\\
n_a\,:\, \id_{e_a} \ &\Rightarrow\  m_{aa}~,
\end{aligned}
\end{equation}
with $\id_{e_a} \in M$. 

The following diagrams, which arise from \eqref{eq:diagrams_2-functors} with $\sfa$, $\sfr$, and $\sfl$ being trivial in $\check \CCC(\frU)$, are commutative:
\begin{subequations}
\begin{equation}
\myxymatrix{
m_{ab}\otimes (m_{bc} \otimes m_{cd}) \ar@{=>}[rr]^{\sfa^{-1}_{m_{ab},m_{bc},m_{cd}}}  \ar@{=>}[d]_{\id_{m_{ab}}\otimes n_{bcd}} & & (m_{ab}\otimes m_{bc}) \otimes m_{cd}  \ar@{=>}[d]^{n_{abc}\otimes\id_{m_{cd}}}\\
m_{ab}\otimes m_{bd} \ar@{=>}[r]_{n_{abd}} & m_{ad} & m_{ac}\otimes m_{cd}  \ar@{=>}[l]_{n_{acd}} 
}
\end{equation}
and
\begin{equation}\label{diag:CocycConWeakPB-b}
\xymatrixcolsep{4pc}\myxymatrix{
m_{ab}\otimes\id_{e_b}  \ar@{=>}[r]^{\id_{m_{ab}}\otimes n_b}  \ar@{=>}[dr]_{\sfl_{m_{ab}}}  & m_{ab}\otimes m_{bb}  \ar@{=>}[d]^{n_{abb}}\\
& m_{ab}
} \eand
\myxymatrix{
\id_{e_a} \otimes m_{ab} \ar@{=>}[r]^{n_a\otimes\id_{m_{ab}}}  \ar@{=>}[dr]_{\sfr_{m_{ab}}}  & m_{aa}\otimes m_{ab}  \ar@{=>}[d]^{n_{aab}}\\
& m_{ab}
}
\end{equation}
\end{subequations}
In formul{\ae}, this reads as
\begin{subequations}\label{eq:CocycConWeakPB-II}
\begin{equation}\label{eq:CocycConWeakPB-a}
   n_{acd}\circ(n_{abc}\otimes\id_{m_{cd}})\circ \sfa^{-1}_{m_{ab},m_{bc},m_{cd}}\ =\ 
   n_{abd}\circ (\id_{m_{ab}}\otimes n_{bcd})
  \end{equation} 
and
\begin{equation}\label{eq:CocycConWeakPB-b}
           n_{abb}\circ (\id_{m_{ab}}\otimes n_b)\ =\ \sfl_{m_{ab}}\eand
            n_{aab}\circ (n_a\otimes\id_{m_{ab}})\ =\   \sfr_{m_{ab}}~.
\end{equation}
\end{subequations}

\begin{definition}
 A $\CCG$-valued degree-2 \v Cech cochain $(\{m_{ab}\},\{n_{abc}\},\{n_a\})$ that satisfies the equations \eqref{eq:CocycConWeakPB-I} and \eqref{eq:CocycConWeakPB-II} is called a $\CCG$-valued \uline{degree-2 \v Cech cocycle}. The equations \eqref{eq:CocycConWeakPB-I} and \eqref{eq:CocycConWeakPB-II}  are called the \uline{cocycle conditions} of a weak principal 2-bundle $\Phi$ defined by $(\{m_{ab}\},\{n_{abc}\},\{n_a\})$ and the degree-2 \v Cech cocycle $(\{m_{ab}\},\{n_{abc}\},\{n_a\})$ is called its \uline{transition functions}.
\end{definition}

Pushing the analogy with the case of principle bundles further, we derive equivalence relations between weak principal 2-bundles from natural 2-transformations.

\begin{definition}\label{eq:isomorphism-principal-2-bundles}
 Any two weak principal 2-bundles are called \uline{equivalent} if and only if there is a smooth weak natural 2-transformation between their defining weak 2-functors.
\end{definition}

\noindent
Explicitly, for weak principal 2-bundles $\Phi$ and $\tilde \Phi$, such a natural 2-transformation $\alpha:\tilde\Phi\to\Phi$ is given by the following data: we have 1-cells $\{m_a\}$ and 2-cells $\{n_{ab}\}$,
\begin{equation}\label{eq:EquivCocycWeakPB-I}
\begin{aligned}
 m_a\, :\, \tilde e_a\ &\to\ e_a~,\\
 n_{ab}\, :\, m_{ab}\otimes m_b\ &\Rightarrow\ m_a\otimes \tilde m_{ab}~,
\end{aligned}
\end{equation}
defined by the diagram 
\begin{equation}
 \xymatrixcolsep{5pc}
\myxymatrix{
 e_b \ar@{->}[r]^{m_{ab}} \ar@{=>}[rd]^{n_{ab}~} & e_a \\
 \tilde e_b \ar@{->}[u]^{m_b} \ar@{->}[r]_{\tilde m_{ab}}  & \tilde e_a\ar@{->}[u]_{m_a} 
}
\end{equation}
 The coherence conditions for natural 2-transformations require also the diagrams
\begin{subequations}
\begin{equation}\label{diag:EquivCocycWeakPB-a}
\xymatrixcolsep{3.5pc}\myxymatrix{
  m_{ab}\otimes  (m_{b}\otimes  \tilde m_{bc})  \ar@{=>}[r]^{\sfa^{-1}_{m_{ab},m_b,\tilde m_{bc}}}  & (m_{ab}\otimes  m_{b})\otimes  \tilde m_{bc}   \ar@{=>}[r]^{n_{ab}\otimes \id_{\tilde m_{bc}}}  & (m_a\otimes \tilde m_{ab})\otimes \tilde m_{bc}  \ar@{=>}[d]^{\sfa_{m_{a},\tilde m_{ab},\tilde m_{bc}}}\\
 m_{ab}\otimes (m_{bc}\otimes  m_c)  \ar@{=>}[u]^{\id_{m_{ab}}\otimes    n_{bc}} & &  m_a\otimes (\tilde m_{ab}\otimes \tilde m_{bc})  \ar@{=>}[d]^{\id_{m_a}\otimes \tilde  n_{abc}}\\
 (m_{ab}\otimes m_{bc})\otimes  m_c  \ar@{=>}[r]_{~~~~~~n_{abc}\otimes \id_{m_c}}  
 \ar@{=>}[u]^{\sfa_{m_{ab},m_{bc},m_c}}  & m_{ac}\otimes  m_c  \ar@{=>}[r]_{n_{ac}} & m_a\otimes  \tilde m_{ac}
 }
\end{equation}
and
\begin{equation}\label{diag:EquivCocycWeakPB-b}
\myxymatrix{
  \id_{e_a}\otimes  m_a  \ar@{=>}[r]^{~~~~~\sfr_{m_a}}  \ar@{=>}[d]_{n_a\otimes \id_{m_a}} & m_a  \ar@{=>}[r]^{\sfl^{-1}_{m_a}~~~~~} & m_a\otimes  \id_{e_a}  \ar@{=>}[d]^{\id_{m_a}\otimes \tilde n_a}  \\
  m_{aa}\otimes  m_a  \ar@{=>}[rr]_{n_{aa}}  & & m_a\otimes \tilde m_{aa}
  }
\end{equation}
\end{subequations}
to be commutative. In formul{\ae}, this amounts to
\begin{subequations}\label{eq:EquivCocycWeakPB-II}
\begin{equation}\label{eq:EquivCocycWeakPB-a}
\begin{aligned}
  n_{ac}\circ  (n_{abc}\otimes \id_{m_c})\ =\ (\id_{m_a}\otimes \tilde n_{abc})\circ \sfa_{m_a,\tilde m_{ab},\tilde m_{bc}}\circ (n_{ab}\otimes \id_{\tilde m_{bc}})\,\circ  \\
  \circ \, \sfa^{-1}_{m_{ab},m_b,\tilde m_{bc}}\circ (\id_{m_{ab}}\otimes  n_{bc})\circ \sfa_{m_{ab},m_{bc},m_c}
\end{aligned}
\end{equation}
and
\begin{equation}\label{eq:EquivCocycWeakPB-b}
 n_{aa}\circ (n_a\otimes \id_{m_a})\ =\ (\id_{m_a}\otimes \tilde n_a)\circ \sfl^{-1}_{m_a}\circ \sfr_{m_a}~.
\end{equation}
\end{subequations}

\begin{definition}
 Any two $\CCG$-valued degree-2 \v Cech cocycles $(\{m_{ab}\},\{n_{abc}\},\{n_a\})$ and\linebreak $(\{\tilde m_{ab}\},\{\tilde n_{abc}\},\{\tilde n_a\})$ are called \uline{equivalent} or \uline{cohomologous} if and only if there is a $\CCG$-valued degree-1 \v Cech cochain $(\{m_{a}\},\{n_{ab}\})$ such that the equations \eqref{eq:EquivCocycWeakPB-I} and \eqref{eq:EquivCocycWeakPB-II} are satisfied. These equations are called the \uline{coboundary conditions} for a weak principal 2-bundle $\Phi$ defined by $(\{m_{ab}\},\{n_{abc}\},\{n_a\})$, and, slightly deviating from the usual nomenclature, the degree-1 \v Cech cochain $(\{m_{a}\},\{n_{ab}\})$ is called a \uline{degree-2 \v Cech coboundary}.
\end{definition}

\begin{definition}
 A weak principal 2-bundle that is  equivalent to the weak principal 2-bundle specified by the functor
  \begin{equation}
  \{m_{ab}=\id_{e_a}\}~,~~~\{n_{abc}=\sfl_{\id_{e_a}}=\sfr_{\id_{e_a}}\}~,\eand \{n_a=\id_{\id_{e_a}}\}
 \end{equation}
  is called \uline{trivial}.
\end{definition}

\noindent
We shall give explicit formul{\ae} for the transition functions of trivial bundles in the case of {\it semistrict} principal 2-bundles later on. 

Note that for strict 2-bundles, the 2-cells $\{n_a\}$ can always be chosen to be trivial, as was done, for instance, in \cite{Saemann:2012uq,Saemann:2013pca}. The same is true here, as we verify now.

\begin{lemma}\label{lem:construct_equiv_2-bundle}
 Consider transition functions $(\{m_{ab}\},\{n_{abc}\},\{n_a\})$ of a weak principal 2-bundle $\Phi$. The triple $(\{\tilde m_{ab}\},$ $\{\tilde n_{abc}\},\{\tilde n_a\})$ which agrees with that of $\Phi$ except for 
 \begin{equation}
  \{\tilde m_{aa}=\id_{e_a}\}~,~~~\{\tilde n_{aab}=\sfr_{\tilde m_{ab}}\}~,~~~
  \{\tilde n_{abb}=\sfl_{\tilde m_{ab}}\}~,\eand \{\tilde n_a=\id_{\id_{e_a}}\}~,
 \end{equation}
 defines another weak principal 2-bundle $\tilde \Phi$. In addition, these equations imply
 \begin{equation}
  \{\tilde n_{aaa}=\sfr_{\id_{e_a}}=\sfl_{\id_{e_a}}\}  
 \end{equation}
\end{lemma}
\begin{proof}
 One readily checks that the cocycle conditions \eqref{eq:CocycConWeakPB-I} and \eqref{eq:CocycConWeakPB-II} are satisfied for any possible doubling of indices.
\end{proof}

\begin{definition}\label{def:normalisation}
 For every weak principal 2-bundle $\Phi$, the weak principal 2-bundle $\tilde \Phi$ obtained from the construction of Lemma \ref{lem:construct_equiv_2-bundle} is called the \uline{normalisation of $\Phi$}.
\end{definition}

\begin{prop}\label{prop:Normalisation}
 Every weak principal 2-bundle is equivalent to its normalisation.
\end{prop}
\begin{proof}
 A natural 2-transformation that yields the equivalence is given by
 \begin{equation}
  \{m_a=\id_{e_a}\}\eand \left\{n_{ab}=\left\{\begin{array}{ll}
                                              \sfr^{-1}_{\tilde m_{ab}}\circ \sfl_{m_{ab}} &\efor a\neq b\\
                                              n_a^{-1}\otimes \id_{\id_{e_a}} &\efor a=b
                                             \end{array}\right.\right\}.
 \end{equation}
 Note that $\tilde m_{ab}=m_{ab}$ for $a\neq b$. As one may check, the coboundary conditions \eqref{eq:EquivCocycWeakPB-I} and \eqref{eq:EquivCocycWeakPB-II} are indeed satisfied. Note that the choice of $n_{ab}$ in this transformation for $a\neq b$ is not unique. 
\end{proof}

\begin{cor}
 Every weak principal 2-bundle is locally trivialisable.
\end{cor}
\begin{proof}
 By Proposition \ref{prop:Normalisation}, a weak principal 2-bundle $\Phi$ is equivalent to its normalisation, for which we have
 \begin{equation}
  n_{aaa}\ =\ \sfl_{\id_{e_a}}=\sfr_{\id_{e_a}}~,~~~n_a\ =\ \id_{\id_{e_a}}~,\eand m_{aa}\ =\ \id_{e_a}~.
 \end{equation}
 on any $U_a\in\frU$. Thus, the weak principal 2-bundle is locally equivalent to a trivial one.
\end{proof}

Recall that trivial principal bundles with structure group $\sG$ are described by transition functions $\{g_{ab}\}$ of the form $g_{ab}=g_ag_b^{-1}$, where $\{g_a\}$ is a $\sG$-valued \v Cech 0-cochain. Note that the $g_a$ can be multiplied by a (global) $\sG$-valued function from the right, leaving $g_{ab}=g_ag_b^{-1}$ invariant. This is an equivalence relation, which is described by modifications in functorial language.

The corresponding equivalence relations are more comprehensive in the case of principal 2-bundles, as we shall see in the following. Consider two equivalent weak principal 2-bundles $\Phi$ and $\tilde \Phi$ with natural 2-transformations $\alpha:\tilde \Phi\to \Phi$ and $\tilde\alpha:\tilde \Phi\to \Phi$ between them. A weak 2-modification $\varphi:\alpha \Rightarrow \tilde \alpha$ is given by a smooth map $\varphi:\alpha\to\tilde\alpha$ that assigns to every object $(x,a)\in \check \CCC(\frU)$ a 2-morphism $\varphi_{(x,a)}: \alpha_{(x,a)}\Rightarrow\tilde\alpha_{(x,a)}$. We set $o_a:=\varphi_{(x,a)}$ so that $o_a:m_a\Rightarrow\tilde m_a$. Moreover, the following diagram is required to be commutative:
\begin{equation}
\xymatrixcolsep{4pc}\myxymatrix{
m_{ab}\otimes m_b  \ar@{=>}[r]^{\id_{m_{ab}}\otimes o_{b}} \ar@{=>}[d]_{n_{ab}} & m_{ab}\otimes \tilde m_b  \ar@{=>}[d]^{\tilde n_{ab}}\\
m_a\otimes\tilde m_{ab} \ar@{=>}[r]_{o_{a}\otimes\id_{\tilde m_{ab}}} & \tilde m_a\otimes\tilde m_{ab}
}
\end{equation}
that is,
\begin{equation}\label{eq:EquivCocycWeakPB-modi}
 \tilde n_{ab}\circ (\id_{m_{ab}}\otimes o_{b})\ =\ (o_{a}\otimes\id_{\tilde m_{ab}})\circ n_{ab}~.
\end{equation}

\begin{definition}
  Any two $\CCG$-valued degree-2 \v Cech coboundaries $(\{m_a\},\{n_{ab}\})$ and $(\{\tilde m_a\},$ $\{\tilde n_{ab}\})$ between any two $\CCG$-valued degree-2 \v Cech cocycles $(\{m_{ab}\},$ $\{n_{abc}\},$ $\{n_a\})$ and $(\{\tilde m_{ab}\},$ $\{\tilde n_{abc}\},\{\tilde n_a\})$ are said to be \uline{equivalent} if and only if there is a $\CCG$-valued degree-0 \v Cech cochain $\{o_a\}$ such that equations \eqref{eq:EquivCocycWeakPB-modi} are satisfied. Such a degree-0 \v Cech cochain $\{o_a\}$ is called a \uline{degree-2 \v Cech modification}.
\end{definition}

To define pullbacks and restrictions of weak principal 2-bundles, we proceed just as in the case of the functorial description of principal bundles; see Definitions \ref{def:pullback-1} and \ref{def:restriction-1}. Recall that given a smooth map $\phi:X\rightarrow Y$ and a covering $\frU_Y$ of $Y$, the pre-images of the patches in $\frU_Y$ form a covering of $\frU_X$. The resulting groupoid morphisms $\check\CCC(\frU_X)\rightarrow \check\CCC(\frU_Y)$ can be extended to a strict 2-functor $\phi_{\frU}$. Therefore, we give the following definitions.
\begin{definition}\label{def:pullback-2}
 The \uline{pullback} of a weak principal 2-bundle $\Phi$ over $Y$ with respect to an open covering $\frU_Y$ along a map $\phi:X\rightarrow Y$ is the composition of 2-functors $\Phi\circ \phi_{\frU}$.
\end{definition}

\begin{definition}\label{def:restriction-2}
 The \uline{restriction} of a weak principal 2-bundle $\Phi$ over a manifold $X$ to a submanifold $Y$ inside $X$ is the pullback of $\Phi$ along the embedding map $Y\embd X$.
\end{definition}

\subsection{Semistrict and strict principal 2-bundles}

We shall be specifically interested in weak principal 2-bundles with semistrict structure 2-groups. This implies a number of simplifications, which we shall discuss in the following.

\begin{definition}
 A \uline{semistrict principal 2-bundle} is a normalised weak principal 2-bundle with semistrict structure 2-group $\CCG$.
\end{definition}

\noindent
Explicitly, we have a weak 2-functor $\Phi$ described by a \v Cech 2-cocycle (transition functions) $(\{m_{ab}\},$ $\{n_{abc}\},\{n_a\})$ with values in $\CCG$ such that
\begin{subequations}
 \begin{equation}
  \{ m_{aa}=\id_{e_a}\}~,~~~\{ n_{aab}=\id_{ m_{ab}}\}~,~~~
  \{ n_{abb}=\id_{ m_{ab}}\}~,\eand \{ n_a=\id_{\id_{e_a}}\}~.
 \end{equation}
The cocycle conditions for this type of principal 2-bundle then read as 
\begin{equation}\label{eq:CocycConSemiPB}
\begin{aligned}
  n_{abc}\,:\, m_{ab}\otimes  m_{bc}\ &\Rightarrow\ m_{ac}~,\\
   n_{acd}\circ (n_{abc}\otimes \id_{m_{cd}})\circ  \sfa^{-1}_{m_{ab},m_{bc},m_{cd}}\ &=\ 
   n_{abd}\circ  (\id_{m_{ab}}\otimes  n_{bcd})~,
   \end{aligned}
  \end{equation} 
  while the coboundary conditions and modifications are given by
\begin{equation}\label{eq:EquivCocycSemiPB}
\begin{aligned}
 m_a\, :\, \tilde e_a\ &\to\  e_a~,\\
 n_{ab}\, :\, m_{ab}\otimes  m_b\ &\Rightarrow\ m_a\otimes  \tilde m_{ab}~,\\
  n_{ac}\circ  (n_{abc}\otimes \id_{m_c})\ &=\ (\id_{m_a}\otimes \tilde n_{abc})\circ \sfa_{m_a,\tilde m_{ab},\tilde m_{bc}}\circ (n_{ab}\otimes \id_{\tilde m_{bc}})\,\circ  \\
  &\kern1cm \circ \, \sfa^{-1}_{m_{ab},m_b,\tilde m_{bc}}\circ (\id_{m_{ab}}\otimes  n_{bc})\circ \sfa_{m_{ab},m_{bc},m_c}
\end{aligned}
\end{equation}
and
\begin{equation}
\begin{aligned}
o_a:m_a\ &\Rightarrow\ \tilde m_a~,\\
 \tilde n_{ab}\circ  (\id_{m_{ab}}\otimes  o_{b})\ &=\ (o_{a}\otimes \id_{\tilde m_{ab}})\circ  n_{ab}~,
 \end{aligned}
\end{equation}
\end{subequations}
respectively.

\begin{rem}\label{rem:trivial-2-bundle}
A trivial semistrict principal 2-bundle is described by transition functions $(\{m_{ab}\}$, $\{n_{abc}\})$ given in terms of coboundary data $(\{m_a\},\{n_{ab}\})$ according to
\begin{equation}\label{eq:trivial-2-cocycles}
\begin{aligned}
 m_a\, :\, \tilde e_a\ \to\  e_a &\eand n_{ab}\, :\, m_{ab}\otimes  m_b\ \Rightarrow\ m_a~,\\
  n_{ac}\circ  (n_{abc}\otimes \id_{m_c})\ &=\ n_{ab}\,\circ(\id_{m_{ab}}\otimes  n_{bc})\circ \sfa_{m_{ab},m_{bc},m_c}~,
\end{aligned}
\end{equation}
where $n_{aa}=\id_{m_a}$.
\end{rem}

To recover principal 2-bundles based on crossed modules as discussed in most of the current literature, we define the following.
\begin{definition}
 A \uline{strict principal 2-bundle} is a weak principal 2-bundle with strict structure 2-group.
\end{definition}

A well-known result is then the following.
\begin{prop}
 A strict principal 2-bundle $\Phi$ with strict structure 2-group $\CCG$ can be equivalently described in terms of \v Cech cochains taking values in the equivalent crossed module of Lie groups $(\sH\xrightarrow{\,\dpar\,}\sG,\acton)$. In particular, we have a $\sG$-valued \v Cech 1-cochain $\{g_{ab}\}$ and an $\sH$-valued \v Cech 2-cochain $\{h_{abc}\}$  such that
 \begin{subequations}
 \begin{equation}\label{eq:strict-cocycle}
  \begin{aligned}
   \dpar(h_{abc})g_{ab}g_{bc}\ =\ g_{ac}\eand h_{acd}h_{abc}\ =\ h_{abd}(g_{ab}\acton h_{bcd})~.
  \end{aligned}
 \end{equation}

Coboundaries are then described in terms of $\sG$-valued \v Cech 0-cochains $\{g_a\}$ and $\sH$-valued \v Cech 1-cochains $\{h_{ab}\}$. In particular, any two strict principal 2-bundles $(\{g_{ab}\},$ $\{h_{abc}\})$ and $(\{\tilde g_{ab}\},\{\tilde h_{abc}\})$ are said to equivalent if and only if
 \begin{equation}\label{eq:strict-coboundaries}
 g_a\tilde g_{ab}\ =\ \dpar(h_{ab})g_{ab}g_b\eand
 h_{ac}h_{abc}\ =\ (g_a\acton\tilde h_{abc})h_{ab}(g_{ab}\acton h_{bc})~.
 \end{equation}
 
 In addition, any two coboundaries  $(\{g_{a}\},\{h_{ab}\})$ and $(\{\tilde g_{a}\},\{\tilde h_{ab}\})$ are equivalent if and only if there is an $\sH$-valued \v Cech 0-cochain $\{h_a\}$  such that 
 \begin{equation}\label{eq:strict-modifications}
 \tilde g_a\ =\ g_a\dpar(h_a)\eand
 \tilde h_{ab}\ =\ (g_a\acton h_ah_b^{-1})h_{ab}~.
 \end{equation}
 \end{subequations}
\end{prop}

\vspace{10pt}
\noindent
{\it Proof:}
Let us again sketch the identification. For a strict principal 2-bundle, the cocycle and coboundary conditions, as well as the coherence equation for modifications, reduce to
\begin{subequations}
\begin{equation}\label{eq:CocycConStrictPB}
\begin{aligned}
  n_{abc}\,:\, m_{ab}\otimes m_{bc}\ &\Rightarrow\ m_{ac}~,\\
   n_{acd}\circ(n_{abc}\otimes\id_{m_{cd}})\ &=\ 
   n_{abd}\circ(\id_{m_{ab}}\otimes n_{bcd})~,
   \end{aligned}
  \end{equation} 
and
\begin{equation}\label{eq:EquivCocycStrictPB}
\begin{aligned}
m_a\, :\, \tilde e_a\ &\to\ e_a~,\\
 n_{ab}\, :\, m_{ab}\otimes m_b\ &\Rightarrow\ m_a\otimes \tilde m_{ab}~,\\
  n_{ac}\circ (n_{abc}\otimes\id_{m_c})\ &=\ (\id_{m_a}\otimes\tilde n_{abc})\circ(n_{ab}\otimes\id_{\tilde m_{bc}})
 \circ(\id_{m_{ab}}\otimes n_{bc})~,
\end{aligned}
\end{equation}
and
\begin{equation}\label{eq:ModiStrictPB}
\begin{aligned}
o_a:m_a\ &\Rightarrow\ \tilde m_a~,\\
 \tilde n_{ab}\circ (\id_{m_{ab}}\otimes o_{b})\ &=\ (o_{a}\otimes\id_{\tilde m_{ab}})\circ n_{ab}~.
 \end{aligned}
\end{equation}
\end{subequations}
Next, recall the identification of strict Lie 2-groups with crossed modules of Lie groups of Proposition \ref{prop:equiv-strict-Lie-2-CM}. To go from a crossed module of Lie groups $\sH\xrightarrow{\,\dpar\,}\sG$ to a strict Lie 2-group $\CCG$, we identify $\CCG$ with $(\sG,\sG\ltimes \sH)$ in terms of the Lie groups $\sG$ and $\sH$ contained in the equivalent crossed module, we can identify $m_{ab}=g_{ab}$ and $n_{abc}=(g_{abc},h_{abc})$. From
\begin{equation}
\begin{aligned}
 g_{abc}\ &=\ \sft(n_{abc})\ =\ g_{ac}\ =\ m_{ac}~,\\
 \sfs(n_{abc})\ &=\ m_{ab}\otimes m_{bc}\ =\ g_{ab}g_{bc}\ =\ \dpar(h_{abc}^{-1})g_{abc}~,
 \end{aligned}
\end{equation}
we immediately obtain the first equation in \eqref{eq:strict-cocycle}. Likewise, using $\id_{m_{ab}}=(g_{ab},\unit_\sH)$ and \eqref{eq:CompositionsStrictLie2Group}, it is a straightforward exercise to show that \eqref{eq:CocycConStrictPB} simplifies to the second equation in \eqref{eq:strict-cocycle}. We skip the inverse transition from Lie 2-groups to crossed modules here; details on this point can be found in the proof of Proposition \ref{eq:inverse-identification}.

Following the same line of arguments, the coboundary conditions \eqref{eq:EquivCocycStrictPB} and modifications \eqref{eq:ModiStrictPB} are rewritten as \eqref{eq:strict-coboundaries} and \eqref{eq:strict-modifications}. \hfill $\Box$
 
\vspace{10pt}

\begin{rem}\label{rem:strict2bundis}
In the strict setting, we may define
\begin{equation}
\ell_{ab}\ :=\ n_{ab}\otimes\id_{\overline{m}_b}~,
\end{equation}
where $m\otimes\overline{m}=\id_e$. It is easy to see that $\ell_{ab}: m_{ab}\Rightarrow m_a\otimes \tilde m_{ab}\otimes \overline{m}_b$, and, in particular, if the bundle is trivial, then $\ell_{ab}: m_{ab}\Rightarrow m_a \otimes \overline{m}_b$. In this case, one may show that $n_{abc}$ can be rewritten in terms of $\ell_{ab}$ as
\begin{equation}
 n_{abc}\ =\ell_{ac}^{-1}\circ (\ell_{ab}\otimes\ell_{bc})~.
\end{equation}
It is amusing to note the resemblance with a trivial Abelian gerbe: the only difference is that ordinary products are replaced by $\circ$ and $\otimes$. 
\end{rem}

\section{Differentiating semistrict Lie 2-groups}\label{sec:Lie-functor}

In order to define connective structures on semistrict principal 2-bundles, we first need to develop a way of differentiating a semistrict Lie 2-group to a semistrict Lie 2-algebra. The approach we shall develop is based on an idea of \v Severa's \cite{Severa:2006aa} (see also Jur\v co \cite{Jurco20122389}). 

As before, we let $X$ be a smooth manifold. The sheaf of smooth differential $p$-forms on $X$ is denoted by $\Omega^p_X$, and we set $\Omega^\bullet_X:=\bigoplus_{p\geq 0} \Omega^p_X$. In general, given a module $\frv=\bigoplus_{k\in\RZ} \frv_k$  with a $\RZ$-grading, one may always introduce a $\RZ_2$-grading referred to as the Gra{\ss}mann parity in terms of the parity of degrees: $\frv=\bigoplus_{k\in\RZ} \frv_{2k}\oplus \bigoplus_{k\in\RZ} \frv_{2k-1}$. Elements of $\bigoplus_{k\in\RZ} \frv_{2k}$ are said to be Gra{\ss}mann-even while elements of $\bigoplus_{k\in\RZ} \frv_{2k-1}$  are said to be Gra{\ss}mann-odd, respectively. We shall also make use of the Gra{\ss}mann-parity changing functor $\Pi$. For instance, $\FR^{m|n}:=\FR^m\oplus \Pi\FR^n$. Moreover,  $\frv[k]$ will denote the module $\frv$ with grading shifted by $k$. Similarly, $T[k]X$ denotes the tangent bundle of $X$ with the grading of the fibres shifted by $k$.

\subsection{Basic ideas}

\begin{definition}
Let $\sigma:Y\rightarrow X$ be a surjective submersion and $\sG$ be a Lie group. A $\sG$-valued \uline{descent datum} on  $\sigma:Y\rightarrow X$  is a map $g:Y\times_X Y\to \sG$ such that\footnote{Recall that $Y\times_X\cdots\times_X Y:=\{(x_1,\ldots,x_k))\in Y\times \cdots\times Y\,|\, \sigma(x_1)=\cdots=\sigma(x_k)\}$.}
 \begin{equation}
   g(x_1,x_1)\ =\ \unit_\sG\eand g(x_1,x_2)g(x_2,x_3)\ =\ g(x_1,x_3) 
 \end{equation}
 for all $(x_1,x_2,x_3)\in Y\times_X Y\times_X Y$. 
 \end{definition}

\noindent
Specifically, a given descent datum describes the descent of a trivial principal $\sG$-bundle over $Y$ to a non-trivial principal $\sG$-bundle over $X$. The following example makes this more transparent.

\begin{exam}
Let $X$ be a smooth manifold with covering $\frU=\{U_a\}_{a\in\CI}$ indexed by the index set $\CI$. Consider the trivial projection $\sigma : \CI\times X\to X$. A $\sG$-valued descent datum is then given by a map $g:\CI\times\CI \times X\to X$ such that $g(a,a,x)=\unit_\sG$ and $g(a,b,x)g(b,c,x)=g(a,c,x)$ for all $a,b\in\CI$ and $x\in X$. Setting $g_{ab}(x):=g(a,b,x)$, we have obtained a $\sG$-valued \v Cech 1-cocycle $\{g_{ab}\}$ on $X$ relative to the covering $\frU$. This, in turn, describes a principal $\sG$-bundle over $X$.
\end{exam}

Below, we shall be interested in the trivial projection $\sigma:\FR^{0|1}\times X\rightarrow X$, so a $\sG$-valued decent datum is in this case given by a map $g:\FR^{0|1}\times \FR^{0|1}\times X\rightarrow \sG$ such that 
\begin{equation}\label{eq:cocycle-1-theta}
 g(\theta_0,\theta_1,x)g(\theta_1,\theta_2,x)\ =\ g(\theta_0,\theta_2,x)\efor x\ \in\ X~.
\end{equation}
We can regard the maps from the surjective submersion $\FR^{0|1}\times X\rightarrow X$ to a descent datum as a contravariant functor from the category of smooth manifolds to the category of sets. As we shall see below, this functor is representable by $\frg[-1]$, where $\frg$ is the Lie algebra of $\sG$.  In particular, calculating the moduli of this functor yields the Lie algebra $\frg$ as a vector space. To describe its Lie bracket, one needs to compute the action of its Chevalley--Eilenberg differential\footnote{see Appendix \ref{app:A} for the relevant definitions} $\dd_{\rm CE}$. This differential is governed by a generator of the natural action of $C^\infty(\FR^{0|1},\FR^{0,1})$ on the descent data, as was first discussed by Kontsevich \cite{Kontsevich:1997vb} (see also \cite{Severa:2006aa}). Let us now review this in some more detail.

\begin{prop}
 There is a natural isomorphism $H^0(X,\Omega^\bullet_X)\cong C^\infty (C^\infty(\FR^{0|1},X),\FR)$.
\end{prop}

\begin{proof}
Consider first the case  $X=\FR^n$ equipped with standard coordinates $(x^1,\ldots, x^n)$. An element of $C^\infty(\FR^{0|1},\FR^n)$ is parameterised as $(x^1,\ldots, x^n)=(a^1+\alpha^1\theta,\ldots,a^n+\alpha^n\theta)$, where $\theta,\alpha^i\in \FR^{0|1}$ are Gra{\ss}mann-odd and $a^i\in \FR$ are Gra{\ss}mann-even for $i=1,\ldots,n$. We thus have established $C^\infty(\FR^{0|1},\FR^n)\cong\FR^{n|n}$. Furthermore, functions on $\FR^{n|n}$ are polynomials in the Gra{\ss}mann-odd coordinates. Thus, upon identifying the $a^i$ with the coordinates on $\FR^n$ and the $\alpha^i$ with the corresponding differential 1-forms, we have  obtained $ H^0(\FR^n,\Omega^\bullet_{\FR^n})\cong C^\infty(C^\infty(\FR^{0|1},\FR^n),\FR)$. For a general smooth manifold $X$, we have thus a local isomorphism between  $C^\infty(\FR^{0|1},X)$ and $T[-1]X$. However, this isomorphism is independent of the choice of coordinates, and, hence, $C^\infty(\FR^{0|1},X)$ and $T[-1]X$
can be naturally identified. This, in turn, leads to the isomorphism $H^0(X,\Omega^\bullet_X)\cong C^\infty (C^\infty(\FR^{0|1},X),\FR)$.
\end{proof}

Furthermore, the de Rham differential $\dd$ on $H^0(X,\Omega^\bullet_X)$ follows from the action of $C^\infty(\FR^{0|1},\FR^{0|1})$ on $C^\infty(\FR^{0|1},X)$. Concretely, transformations of the form $\theta\mapsto \tilde{\theta}= b\theta+\beta$ for $b\in \FR$, $\beta\in \FR^{0|1}$ induce an action on elements of $C^\infty(\FR^{0|1},X)$ which in local coordinates $(x^1,\ldots,x^n)$ of $X$ is given by 
\begin{equation}
 x^i(\theta)\ =\ a^i+\alpha^i\theta~~~\mapsto~~~x^i(\tilde{\theta})\ =\ a^i+(b\theta+\beta)\alpha^i\ =\ (a^i+\beta\alpha^i)+b\alpha^i\theta
\end{equation}
for $i=1,\ldots,n$. Translated into differential forms, this means that $x^i\mapsto x^i+\beta\dd x^i$ and $\dd x^i\mapsto b \dd x^i$. We thus arrive at the following result.

\begin{prop}\label{prop:action-dK}
 The action of the de Rham differential $\dd$ on $H^0(X,\Omega^\bullet_X)$ translates to an action of the generator $\dd_{\rm K}$ of $C^\infty(\FR^{0|1},\FR^{0|1})$ given by
 \begin{equation}\label{eq:action-dK}
  \dd_{\rm K} f(x(\theta))\ =\ \dder{\eps}f(x(\theta+\eps))
 \end{equation}
 for any $f\in C^\infty (C^\infty(\FR^{0|1},X),\FR)$.
\end{prop}

\noindent
We would like to point out that the differential $\dd_{\rm K}$ extends to smooth functions $f\in C^\infty (C^\infty(\FR^{0|k},X),\FR)$, since there is a natural action of $C^\infty(\FR^{0|1},\FR^{0|1})$ on $C^\infty(\FR^{0|k},X)$. Specifically, its action on a function of several Gra{\ss}mann-odd coordinates $(\theta_0,\ldots,\theta_{k-1})$ is diagonal,
 \begin{equation}\label{eq:action-dK-1}
  \dd_{\rm K} f(x(\theta_0,\ldots,\theta_{k-1}))\ =\ \dder{\eps}f(x(\theta_0+\eps,\ldots,\theta_{k-1}+\eps))~,
 \end{equation}
 where $x(\theta_0,\ldots,\theta_{k-1}):=(x^1(\theta_0,\ldots,\theta_{k-1}),\ldots,x^n(\theta_0,\ldots,\theta_{k-1}))$.

 \begin{exam}
Consider a Gra{\ss}mann-even function $f(x(\theta_0,\theta_1))=g+\phi\theta_0 +\psi\theta_1+F\theta_0\theta_1$. We obtain
\begin{equation}\label{eq:induced_differential}
\dd_{\rm K} f(x(\theta_0,\theta_1))\ =\ \dder{\eps} f(x(\theta_0+\eps,\theta_1+\eps))\ =\ -\phi-\psi+(\theta_0-\theta_1)F
\end{equation}
from \eqref{eq:action-dK-1}. Comparing coefficients in the Gra{\ss}mann-odd coordinates, we can read off the action of an induced operator, again denoted by $\dd_{\rm K}$ on the components
\begin{equation}
 \dd_{\rm K} g\ =\ -\phi-\psi~,\quad \dd_{\rm K} \phi\ =\ F~,\quad \dd_{\rm K} \psi\ =\ -F~,\eand \dd_{\rm K}F\ =\ 0~.
\end{equation}
\end{exam}

\pagebreak

\begin{prop}
 The operator $\dd_{\rm K}$ is a differential. That is, it has the following properties:
 \begin{enumerate}[(i)]\setlength{\itemsep}{-1mm}
  \item $\dd_{\rm K}\circ\dd_{\rm K}=0$,
  \item for any $f\in C^\infty (C^\infty(\FR^{0|k},X),\FR)$ and $g\in C^\infty (C^\infty(\FR^{0|l},X),\FR)$,  the operator $\dd_{\rm K}$ obeys a graded Leibniz rule,
 \begin{equation}\label{eq:Leibniz}
  \dd_{\rm K}(fg)\ =\ (\dd_{\rm K}f)g+(-1)^{|f|}f\dd_{\rm K}g~,
 \end{equation}
 where $|f|$ denotes the Gra{\ss}mann parity of $f$.
 \end{enumerate}
\end{prop}

\begin{proof}
These properties are an immediate consequence of the definition of $\dd_{\rm K}$. \end{proof}

\subsection{Lie algebra of a Lie group}\label{ssec:diff-Lie-groups}

Having collected all relevant ideas, let us put them to use and start by computing the Lie algebra of a Lie group as a guiding example for the case of Lie 2-groups. This has been done in \cite{Severa:2006aa,Jurco20122389}, and our discussion below is an expanded version of the one found in these references.

Consider a Lie group $\sG$ with Lie algebra $\frg=T_{\unit_\sG}\sG$. To prepare our discussion for semistrict Lie 2-groups, we shall not assume that $\sG$ is a matrix group, rather we only make use of the fact that there is a local diffeomorphism $\varphi$ between a neighbourhood $U_\frg$ of $0\in \frg$ and a  neighbourhood $U_{\sG}$ of ${\unit_\sG}\in \sG$ with $\varphi(a)=g$ for $a\in U_\frg$ and $g\in U_\sG$, $\varphi(0)=\unit_\sG$, and $\varphi_*|_0$ is the identity. In addition, we wish to restrict ourselves to infinitesimal neighbourhoods by considering elements of $\frg[-1]$ multiplied by a Gra\ss mann-odd coordinate. 

\begin{prop}\label{prop:mult-rule-1}
 Let $\varphi:U_\frg\to U_\sG$ be the above-described local diffeomorphism. For $a,a_{1,2}\in\frg[-1]$, we have the following relations:
 \begin{equation}\label{eq:algebra_tensor_product}
 \begin{aligned}
   \varphi(a\theta)^{-1}\ &=\ \varphi(-a\theta)~,\\
 \varphi^{-1}(\varphi(a_1\theta_1)\varphi(a_2\theta_2))\ &=\ a_1\theta_1+a_2\theta_2 - a_1\cdot a_2\,\theta_1\theta_2~,
  \end{aligned}
 \end{equation}
 where the operation $\cdot:\frg[-1]\times\frg[-1]\to\frg[-2]$ is defined by the second equation. This operation is bilinear and $a_1\cdot a_2+a_2\cdot a_1= [a_1,a_2]$ is the Lie bracket shifted by one degree. 
\end{prop}
\begin{proof}
First of all, it is clear that $\varphi^{-1}(\varphi(a_1\theta_1)\varphi(a_2\theta_2))$ is a polynomial in the Gra\ss mann-odd coordinates. The terms of this expression linear in $\theta_1$ and $\theta_2$ then follow from putting $\theta_1$ or $\theta_2$ to zero, respectively. In the special case when $\theta_1=\theta_2$ and $a_1=-a_2$, we recover the first equation of \eqref{eq:algebra_tensor_product}. It remains to understand the operation `$\cdot$'. For this, consider the expression
\begin{equation}
\begin{aligned}
&\varphi^{-1}(\varphi(a_1\theta_1+a_2\theta_2)\varphi(a_3\theta_3+a_4\theta_4))\ =\ a_1\theta_1+\cdots +a_4\theta_4\,-\\
 &\hspace{2cm}-\,a_1\cdot a_3\,\theta_1\theta_3-a_1\cdot a_4\,\theta_1\theta_4-
 a_2\cdot a_3\, \theta_2\theta_3-a_2\cdot a_4\,\theta_2\theta_4+\CO(\theta^3)~.
\end{aligned}
\end{equation}
This expansion follows from the second equation of \eqref{eq:algebra_tensor_product} and the special cases $a_1=a_3=0$, $a_1=a_4=0$, $a_2=a_3=0$, and $a_2=a_4=0$. Bilinearity of  `$\cdot$' then follows directly from this expression for $\theta_1=\theta_2$ and $\theta_3=\theta_4$ together with the second equation of \eqref{eq:algebra_tensor_product}. Furthermore, considering the algebra element corresponding to the group commutator
\begin{equation}
\begin{aligned}
  \varphi^{-1}(\varphi(-a_1\theta_1)\varphi(-a_2\theta_2)\varphi(a_1\theta_1)\varphi(a_2\theta_2))\ =\ (a_1\cdot a_2+a_2\cdot a_1)\theta_1\theta_2~,
\end{aligned}
\end{equation}
where the expansion follows from considering the cases either $a_1$ and/or $a_2$ vanish,
we find the shifted Lie bracket $[a_1,a_2]=a_1\cdot a_2+a_2\cdot a_1$. This concludes the proof.
\end{proof}

\begin{rem}
For matrix Lie groups, we may suggestively write
  \begin{equation}
  \begin{aligned}
   (\unit_\sG+a\theta)^{-1}\ &=\ \unit_\sG-a\theta~,\\
  (\unit_\sG+a_1\theta_1)(\unit_\sG+a_2\theta_2)\ &=\ \unit_\sG+\ a_1\theta_1+a_2\theta_2 - a_1\cdot a_2\,\theta_1\theta_2~.
  \end{aligned}
 \end{equation}
In addition, one may also define products between elements $g$ and $a$ of $\sG$ and $\frg[-1]$, respectively. For matrix Lie groups, we simply write $g a$. For general Lie groups, one replaces such expressions by the pullback $\sL^*_g a$ of $a$, where $\sL_g$ denotes left multiplication on $\sG$.
\end{rem}

We are now ready to discuss the computation of the Lie algebra of a Lie group by \v Severa's construction \cite{Severa:2006aa}. Consider a $\sG$-valued descent datum on the trivial projection $\FR^{0|1}\times X\rightarrow X$. That is, we a have smooth map $g:\FR^{0|1}\times \FR^{0|1}\times X\rightarrow \sG$ satisfying the  cocycle condition  \eqref{eq:cocycle-1-theta}. Since we are interested in the functor from the category of smooth manifolds to the category of descent data in the following, we shall suppress the explicit dependence on $x\in X$ and simply write $\{g_{01}:=g(\theta_0,\theta_1)\}$ with $g_{01}g_{12}=g_{02}$ and $g_{10}=g^{-1}_{01}$. Then, we have the following result.

\begin{lemma}\label{ref:lemma-trivialising-1}
Letting $g(\theta):=g(\theta,0)$, we have 
 \begin{equation}
  g(\theta_0,\theta_1)\ =\ g(\theta_0)g(\theta_1)^{-1}~.
 \end{equation}
\end{lemma}

\begin{proof}
This is an immediate consequence of \eqref{eq:cocycle-1-theta}. 
\end{proof}

Next, we may expand\footnote{To simplify notation, we use suggestive notation for matrix groups, which is readily translated to general expressions involving the diffeomorphism $\varphi$.} $g(\theta_0)=\unit_\sG+a\theta_0$ for some $a\in\frg[-1]$ since $g(0)=g(0,0)=\unit_\sG$. Together with the Propositions \ref{prop:action-dK} and \ref{prop:mult-rule-1}, we get the following result.

\begin{prop}\label{lemma:induced_differential}
 A $\sG$-valued descent datum on $\FR^{0|1}\times X\to X$ is parametrised by an element $a\in\frg[-1]$ according to 
 \begin{equation}
  g_{01}\ =\ (\unit_\sG+a\theta_0)(\unit_\sG-a\theta_1)\ =\ \unit_\sG+a(\theta_0-\theta_1)+\tfrac{1}{2}[a,a]\theta_0\theta_1~.
 \end{equation}
 The induced differential is given by
 \begin{equation}\label{eq:CE-MC-1}
  \dd_{\rm K} a+\tfrac{1}{2}[a,a]\ =\ 0~.
 \end{equation}
\end{prop}

As stated previously, we wish to identify the induced action of the differential $\dd_{\rm K}$ with the Chevalley--Eilenberg differential $\dd_{\rm CE}$ on $\frg$. Recall that the Chevalley--Eilenberg differential of a Lie algebra $\frg$ acts as 
\begin{equation}\label{eq:CE-Lie-algebra}
 \dd_{\rm CE} \check{\tau}^i\ =\ -\tfrac{1}{2}f^i_{jk}\check{\tau}^j\wedge \check{\tau}^k~,
\end{equation}
on basis elements $\check \tau^i$ of the dual Lie algebra $\frg^\vee$ of $\frg$. Here, the $f^i_{jk}$ are the structure constants of $\frg$ with respect to the basis elements $\tau_i$ of $\frg$ with $\check \tau^i(\tau_j)=\delta^i_j$. The equation \eqref{eq:CE-MC-1} amounts to the Maurer--Cartan equation $\dd_{\rm CE}a+\tfrac{1}{2}[a,a]=0$ which should be regarded as the equation \eqref{eq:CE-Lie-algebra} evaluated for a polynomial in $a^i$ with $a=a^i\tau_i$.

Altogether, we have proved the following theorem.
\begin{theorem}
 The functor from the category of smooth manifolds $X$ to the category of $\sG$-valued descent data on surjective submersions $\FR^{0|1}\times X\rightarrow X$ is parameterised by elements of $\frg[-1]$ with $\frg=\sLie(\sG)$. The action of the differential $\dd_{\rm K}$ on descent data yields the action of the Chevalley--Eilenberg differential corresponding to $\frg$.
\end{theorem}

Finally, let us consider \v Cech coboundary transformations on  $\{g_{01}=g(\theta_0,\theta_1)\}$. Such transformations are parameterised by smooth maps $p:\FR^{0|1}\rightarrow \sG$ with $\{p_0:=p(\theta_0)\}$ and $p(\theta)=p+ \pi\theta$ for some $p\in\sG$ and $\pi\in T_p[-1]\sG$ according to
\begin{subequations}
\begin{equation}
 \tilde g_{01}\ =\ p_0g_{01}p_1^{-1}\ =\ \unit_\sG+\tilde a(\theta_0-\theta_1)+\tfrac12 [\tilde a,\tilde a]\theta_0\theta_1~,
\end{equation}
where
\begin{equation}
 \tilde a\ :=\ pap^{-1}+\pi p^{-1}~.
\end{equation} 
\end{subequations}
Together with the induced differential $\dd_{\rm K} p=-\pi$, we obtain the following.
\begin{prop}\label{prop:pre-gauge}
 Consider two equivalent $\sG$-valued descent data that are parametrised by $a\in\frg[-1]$ and $\tilde a\in\frg[-1]$, respectively. Then there is a \v Cech coboundary transformations between these, which is parametrised by $p:\FR^{0|1}\rightarrow \sG$ with $p(\theta)=p+ \pi\theta$ for some $p\in\sG$ and $\pi\in T_p[-1]\sG$, such that
 \begin{equation}
  \tilde a \ =\ p a p^{-1}+\pi p^{-1}\ =\ p a p^{-1}+p\dd_{\rm K} p^{-1}~.
 \end{equation}
 The equation $\dd_{\rm K} a+\tfrac{1}{2}[a,a]=0$ is invariant under coboundary transformations. That is, whenever $\dd_{\rm K} a+\tfrac{1}{2}[a,a]=0$ we have $\dd_{\rm K} \tilde a+\tfrac{1}{2}[\tilde a,\tilde a]=0$ and vice versa.
\end{prop}

\begin{rem}
Note that by replacing $\dd_{\rm K}$ by the de Rham differential in all of the above, we recover the definition of the curvature of a connection 1-form on a principal bundle with structure group $\sG$ as well as its gauge transformation. We will make use of this observation later on.
\end{rem}

\subsection{Semistrict Lie 2-algebra of a semistrict Lie 2-group}\label{ssec:Severa-2-groups}

Now we generalise the previous discussion to the case of semistrict Lie 2-groups $\CCG=(M,N)$, which we shall regard as a weak Lie 2-groupoid $\sB\CCG(\{e\},M,N)$ in the following. In this case, the local diffeomorphism $\varphi=(\varphi_M,\varphi_N)$ goes between neighbourhoods $U_\frm$ of $\frm:=T_{\id_{e}}M$ and $U_\frn$ of $\frn:=T_{\id_{\id_{e}}}N$ as well as neighbourhoods $U_M$ of $\id_e$ and $U_N$ of $\id_{\id_{e}}$. As before, $\varphi(0)=(\id_e,\id_{\id_e})$ and $\varphi_*|_0$ is the identity. Following our previous discussion, we shall again be interested in infinitesimal neighbourhoods and we shall always write suggestively $\id_e+a \theta$ and $\id_{\id_e}+b\theta$ for $\varphi_M(a\theta)$ and $\varphi_N(b\theta)$, where $a\in \frm[-1]$ and $b\in \frn[-1]$. 

\begin{prop}\label{prop:linearized-tensor-products}
The bifunctor $\otimes:\sB\CCG\times\sB\CCG\to \sB\CCG$ induces bilinear non-associative products $\otimes:\frm[-1]\times\frm[-1]\to\frm[-2]$ and $\otimes:\frn[-1]\times\frn[-1]\to\frn[-2]$ by means of
 \begin{equation}\label{eq:def:tensor-product}
 \begin{aligned}
  (\id_e+a_1\theta_1)\otimes(\id_e+a_2\theta_2)\ &=\ \id_e+a_1\theta_1+a_2\theta_2-a_1\otimes a_2\,\theta_1\theta_2~,\\
   (\id_{\id_e}+b_1\theta_1)\otimes(\id_{\id_e}+b_2\theta_2)\ &=\ \id_{\id_e}+b_1\theta_1+b_2\theta_2-b_1\otimes b_2\,\theta_1\theta_2~,
  \end{aligned}
  \end{equation}
  where $a_{1,2}\in \frm[-1]$ and $b_{1,2}\in \frn[-1]$, respectively.
\end{prop}

\begin{proof}
 The proof is essentially the same as the one given for Proposition \ref{prop:mult-rule-1}.
\end{proof}

We now turn to the maps induced by the structure maps $\sfs$, $\sft$, and $\id$ on $\frn[-1]$ and $\frm[-1]$. Note that for elements $a\in\frm[-1]$ and $b\in\frn[-1]$, we have
\begin{equation}
\begin{aligned}
 \id_{\id_e+a\theta}\ =\ &\id_{\id_e}+\id_*(a)\theta~,\\
 \sfs(\id_{\id_e}+b\theta)\ =\ \id_e+\sfs_*(b)\theta~&,~~~\sft(\id_{\id_e}+b\theta)\ =\ \id_e+\sft_*(b)\theta~,
\end{aligned}
\end{equation}
where the differentials are to be taken at $\id_{\id_e}$ and $\id_e$, respectively. More generally, the following result holds.

\begin{prop}
 Around $\id_e+a_1\theta_1+a_2\theta_2$ and $\id_{\id_e}+b_1\theta_1+b_2\theta_2$ for some $a_{1,2}\in\frm[-1]$ and $b_{1,2}\in \frn[-1]$, the structure maps expand as follows:
 \begin{equation}
 \begin{aligned}
  &\id_{\id_e+a_1\theta_1+a_2\theta_2}\ =\\ 
  &\hspace{1cm}=\ \id_{\id_e}+\id_*(a_1)\theta_1+\id_*(a_2)\theta_2+(\id_*(a_1\otimes a_2)-\id_*(a_1)\otimes \id_*(a_2))\theta_1\theta_2~, \\
  &\sfs(\id_{\id_e}+b_1\theta_1+b_2\theta_2)\ =\\
  &\hspace{1cm}=\ \id_e+\sfs_*(b_1)\theta_1+\sfs_*(b_2)\theta_2+(\sfs_*(b_1\otimes b_2)-\sfs_*(b_1)\otimes \sfs_*(b_2))\theta_1\theta_2~, \\
  &\sft(\id_{\id_e}+b_1\theta_1+b_2\theta_2)\ =\\
  &\hspace{1cm}=\ \id_e+\sft_*(b_1)\theta_1+\sft_*(b_2)\theta_2+(\sft_*(b_1\otimes b_2)-\sft_*(b_1)\otimes \sft_*(b_2))\theta_1\theta_2~. 
 \end{aligned}
 \end{equation}
\end{prop}

\begin{proof}
 The map $\id$ is compatible with $\otimes$ on $M$ in the following way:
  \begin{equation}
    \id_{(\id_e+a_1\theta_1)\otimes(\id_e+a_2\theta_2)}\ =\ \id_{\id_e+a_1\theta_1}\otimes \id_{\id_e+a_2\theta_2}~.
  \end{equation}
 Expanding both sides of this equation according to Proposition \ref{prop:linearized-tensor-products}
 yields the desired result. The argument for the maps $\sfs$ and $\sft$ is fully analogous.
\end{proof}

Finally, we have to discuss an induced concatenation map on $\frn[-1]$. Note that if $\sfs_*(b_1)=\sft_*(b_2)$ for some $b_{1,2}\in\frn[-1]$, then $\sfs(\id_{\id_e}+b_1\theta)=\sft(\id_{\id_e}+b_2\theta)$.

\begin{definition}\label{def:induced-concatenation}
 For elements $b_{1,2}\in \frn[-1]$ with $\sfs_*(b_1)=\sft_*(b_2)$, we define implicitly
\begin{equation}
 (\id_{\id_e}+b_1\theta)\circ(\id_{\id_e}+b_2\theta)\ =:\ \id_{\id_e}+b_1\circ b_2\,\theta~.
\end{equation}
\end{definition}

It trivially follows that $b_1\circ 0=b_1$ for $\sfs_*(b_1)=0$ and $0\circ b_2=b_2$ for $\sft_*(b_2)=0$. More generally, the induced concatenation map satisfies the following.
\begin{prop}\label{prop:induced-concatenation}
 For $b_{1,2,3,4}\in \frn[-1]$ with $\sfs_*(b_1)=\sft_*(b_3)$, $\sfs_*(b_2)=\sft_*(b_4)$, and $\sfs_*(b_1\otimes b_2)=\sft_*(b_3\otimes b_4)$,
we have
\begin{equation}
  \begin{aligned}
  (\id_{\id_e}+b_1\theta_1+b_2\theta_2)\circ(\id_{\id_e}+b_3\theta_1+b_4\theta_2)\ =\ \id_{\id_e}+b_1\circ b_3\,\theta_1+b_2\circ b_4\,\theta_2~.
  \end{aligned}
\end{equation}
\end{prop}

\begin{rem}\label{rem:linearising-at-p}
Note that above we have linearised all the structure maps $\sfs$, $\sft$, $\id$, $\otimes$, and $\circ$ at $\id_e$ or $\id_{\id_e}$ and obtained maps on $\frm[-1]$ or $\frn[-1]$. We can certainly consider linearisations also at other points $p$ of $M$ or $N$, leading to maps on $T_p[-1]M$ or $T_p[-1]N$. The formul{\ae} in these cases are obvious generalisations of the ones derived above.
\end{rem}

\begin{rem}
In the following, we shall simply write $\sfs$, $\sft$, and $\id$ for $\sfs_*,$ $\sft_*$, and $\id_*$, slightly abusing notation. We shall also write $\id_a$ instead of $\id_*(a)$. The distinction between these linear maps and the finite maps on $M$ and $N$ should always be clear from the context.
\end{rem}

This completes the preliminary discussion, and we can turn to the differentiation of a semistrict Lie 2-group $\CCG=(M,N)$ to a 2-term $L_\infty$-algebra.  Following our discussion for Lie groups, we consider the functor from the category of smooth manifolds $X$ to the category of $\CCG$-valued descent data on surjective submersions $\FR^{0|1}\times X\rightarrow X$ that  are represented by $M$-valued 1-cells  $\{m_{01}:=m(\theta_0,\theta_1)\}$ and $N$-valued 2-cells $\{n_{012}:=n(\theta_0,\theta_1,\theta_2)\}$ so that
\begin{subequations}
\begin{equation}\label{eq:2CellExpCocyc-A}
   n_{012}\,:\, m_{01}\otimes m_{12}\ \Rightarrow\ m_{02}~,
\end{equation}
and
\begin{equation}\label{eq:2CellExpCocyc-B}
n_{023}\circ(n_{012}\otimes\id_{m_{23}})\ =\ 
   n_{013}\circ (\id_{m_{01}}\otimes n_{123})\circ \sfa_{m_{01},m_{12},m_{23}}~.
\end{equation}
\end{subequations}

\noindent
Analogously to Lemma \ref{ref:lemma-trivialising-1}, we have the following statement; see also Remark \ref{rem:trivial-2-bundle}.

\begin{lemma}
 The functor $(\{m_{01}\},\{n_{012}\})$ is trivialised by the following $\CCG$-valued \v Cech 1-cochains $(\{m_0\},\{n_{01}\})$:
 \begin{equation}\label{eq:normalised-cochains-2}
  m_0\ :=\ m(\theta_0)\ :=\ m(\theta_0,0)\eand n_{01}\ :=\ n(\theta_0,\theta_1)\ :=\ n(\theta_0,\theta_1,0)~.
 \end{equation}
 That is, $n_{01}: m_{01}\otimes m_{1}\Rightarrow m_{0}$ with 
\begin{equation}\label{eq:CoBounTrivExp}
n_{02}\circ(n_{012}\otimes\id_{m_{2}})\ =\ 
   n_{01}\circ (\id_{m_{01}}\otimes n_{12})\circ \sfa_{m_{01},m_{12},m_{2}}~.
\end{equation}
Furthermore, 
\begin{equation}\label{eq:Norm2Modn}
  m(0)\ =\ \id_e\eand n(\theta_0,0)\ =\ \id_{m_0}~.
\end{equation}
\end{lemma}
\begin{proof}
 This statement is readily proved by computation and comparison with Remark \ref{rem:trivial-2-bundle}. To this end one needs to use the fact that $\sfa_{m,m',\id_e}$ is trivial for all $m,m'\in M$; see Proposition \ref{prop:TrivAssocId}. Equations \eqref{eq:Norm2Modn} follow from the normalisations of the cocycle conditions for semistrict principal 2-bundles, cf.\  Lemma \ref{lem:construct_equiv_2-bundle}.
\end{proof}

\begin{rem}
 Clearly, there is  a one-to-one correspondence between $\CCG$-valued descent data $(\{m_{01}\},\{n_{012}\})$ and trivialising $\CCG$-valued \v Cech 1-cochains $(\{m_0\},\{n_{01}\})$. Moreover, by a modification isomorphism, any trivialising $\CCG$-valued \v Cech 1-cochain $(\{m_0\},\{n_{01}\})$ is equivalent to one of the form \eqref{eq:normalised-cochains-2}.
\end{rem}

\begin{prop}\label{prop:Expansions}
A descent datum $(\{m_{01}\},\{n_{012}\})$ and the corresponding coboundary datum $(\{m_0\},\{n_{01}\})$ are parametrised by 1-cells $\alpha\in\frm[-1]$ and 2-cells $\beta\in\frn[-2]$ with
\begin{equation}
 \alpha:0\ \rightarrow\ 0 \eand \beta: \sfs(\beta)\ \Rightarrow\ 0
\end{equation}
according to the following expansions in the Gra{\ss}mann-odd coordinates:
\begin{subequations}\label{eq:Expansions}
\begin{eqnarray}
     m_0\! &=&\! \id_e+\alpha\theta_0~,\label{eq:Expansions-A}\\
     n_{01}\! &=&\! \id_{\id_e}+\id_\alpha \theta_0+\beta\theta_0\theta_1~,\label{eq:Expansions-B}\\
     m_{01}\! &=&\! \id_e+\alpha(\theta_0-\theta_1)+\big[\alpha\otimes\alpha+\sfs(\beta)\big]\theta_0\theta_1~,\label{eq:Expansions-C}\\
     n_{012}\! &=&\! \id_{\id_e}+\id_\alpha(\theta_0-\theta_2)+\beta(\theta_0\theta_1+\theta_1\theta_2-\theta_0\theta_2)\,+\notag\\
  &&\kern1cm+\, \id_{\alpha\otimes\alpha+\sfs(\beta)}\theta_0\theta_2+\big[\id_\alpha\otimes\beta-\beta\otimes\id_\alpha+\mu(\alpha,\alpha,\alpha)\big]\theta_0\theta_1\theta_2~, \label{eq:Expansions-D}
\end{eqnarray}
  \end{subequations}
  where $\mu(\alpha,\alpha,\alpha):\alpha\otimes(\alpha\otimes \alpha)-(\alpha\otimes\alpha)\otimes\alpha\Rightarrow 0$.
\end{prop}
\begin{proof}
The expansion of $m_0$ is a direct consequence of \eqref{eq:Norm2Modn} while the expansion \eqref{eq:Expansions-B} follows directly from the conditions $n_{00}=\id_{m_0}=\id_{\id_e}+\id_\alpha\theta_0$ and $n(\theta_0,0)=\id_{m_0}$; $\sft(n_{01})=m_0=\id_e+\alpha\theta_0$ implies $\sft(\beta)=0$. The expansion \eqref{eq:Expansions-C} follows from the normalisation $m_{00}=\id_e$ together with \eqref{eq:Expansions-B} by comparing coefficients in $\sfs(n_{01})=m_{01}\otimes m_1$, where we used the identity
 \begin{equation}\label{eq:induced-identity-1}
 \begin{aligned}
  (\id_e+\alpha(\theta_0-\theta_1)&+\alpha_2\theta_0\theta_1)\otimes(\id_e+\alpha\theta_1)\ =\\
  &=\ \big(\id_e+(\alpha-\tfrac{1}{2}\alpha_2(\theta_0+\theta_1))(\theta_0-\theta_1)\big)\otimes(\id_e+\alpha\theta_1)\\
  &=\ \id_e+\alpha\theta_0+(\alpha_2-\alpha\otimes \alpha)\theta_0\theta_1
 \end{aligned}
 \end{equation}
to evaluate the product. 

To derive the expansion \eqref{eq:Expansions-D}, we use $n(\theta_0,\theta_1,0)=n(\theta_0,\theta_1)$ together  with the normalisation $n_{001}=\id_{m_{01}}$ and $n_{011}=\id_{m_{01}}$. Hence, $n_{012}$ must be of the form
\begin{equation}\label{eq:expn012}
 n_{012}\ =\ \id_{\id_e}+\id_\alpha(\theta_0-\theta_2)+\beta(\theta_0\theta_1+\theta_1\theta_2-\theta_0\theta_2)+
 \id_{\alpha\otimes\alpha+\sfs(\beta)}\theta_0\theta_2+\gamma\,\theta_0\theta_1\theta_2~.
\end{equation}
for some 2-cell $\gamma\in\frn[-3]$. To find $\gamma$  from \eqref{eq:CoBounTrivExp} and \eqref{eq:Expansions-A}--\eqref{eq:Expansions-C}, we require an expansion of the associator $\sfa_{m_{01},m_{12},m_2}$. Since according to Proposition \ref{prop:TrivAssocId} $\sfa_{\id_e,m,m'}$, $\sfa_{m,\id_e,m'}$, and $\sfa_{m,m',\id_e}$ are trivial for all $m,m'\in M$, we can write 
\begin{equation}\label{eq:ExAssPf}
  \sfa_{m_{01},m_{12},m_2}\ =\ \id_{m_{01}\otimes(m_{12}\otimes m_2)}+\mu(\alpha,\alpha,\alpha)\,\theta_0\theta_1\theta_2~,
\end{equation}
defining a linearised 2-cell $\mu(\alpha,\alpha,\alpha):\alpha\otimes(\alpha\otimes \alpha)-(\alpha\otimes\alpha)\otimes\alpha\Rightarrow 0$. In order to evaluate \eqref{eq:CoBounTrivExp} for coboundaries given in \eqref{eq:Expansions}, we note that \eqref{eq:expn012} can be rewritten as
\begin{equation}
  n_{012}\ =\ \id_{\id_e}+\big[\id_\alpha+\tfrac12(\beta- \id_{\alpha\otimes\alpha+\sfs(\beta)}+\gamma\theta_1)(\theta_0+\theta_2)-\beta\theta_1\big](\theta_0-\theta_2)
\end{equation}
and likewise for $n_{01}=\id_{\id_e}+(\id_\alpha-\beta\theta_1)\theta_0$ and all the other terms appearing in \eqref{eq:CoBounTrivExp}. Thus, our definitions of the induced concatenation and products $\otimes$ to linear order are  sufficient to evaluate \eqref{eq:CoBounTrivExp}. For example, we compute
\begin{equation}
 n_{012}\otimes \id_{m_2}\ =\ \id_{\id_e}+\id_\alpha\theta_0+\beta(\theta_0\theta_1+\theta_1\theta_2-\theta_0\theta_2)+
 \id_{\sfs(\beta)}\theta_0\theta_2+(\gamma+\beta\otimes\id_\alpha)\theta_0\theta_1\theta_2~.
\end{equation}
Comparing the coefficient of $\theta_0\theta_1\theta_2$ of both sides of equation \eqref{eq:CoBounTrivExp}, we obtain
\begin{equation}
\gamma\ =\ \id_\alpha\otimes\beta-\beta\otimes\id_\alpha+\mu(\alpha,\alpha,\alpha)~.
\end{equation}
In deriving the latter, we have used $\beta\circ(\id_{\sfs(\beta)}-\beta)=0$, which follows immediately from Proposition \ref{prop:inverse-concatenation}. 
\end{proof}

\begin{cor} \label{cor:Konsdiffalbe}
The induced differentials $\dd_{\rm K}$ of $\alpha\in \frm[-1]$ and $\beta\in\frn[-2]$ with $\sft(\beta)=0$ are given by
\begin{equation}\label{eq:action-dK-2}
\begin{aligned}
 \dd_{\rm K} \alpha\ &=\ -\alpha\otimes\alpha- \sfs(\beta) ~,\\
 \dd_{\rm K} \beta\ &=\ -\id_\alpha\otimes\beta+\beta\otimes \id_\alpha-\mu(\alpha,\alpha,\alpha)~.
\end{aligned}
\end{equation}
\end{cor}
\begin{proof}
This is a direct consequence of the application of the differential $\dd_{\rm K}$ to $\{n_{012}\}$ as given in Proposition \ref{prop:Expansions}. Alternatively, the first of these equations can also be obtained from the application of $\dd_{\rm K}$ to $\{m_{01}\}$.
\end{proof}

 From equations \eqref{eq:action-dK-2}, we can now extract the Chevalley--Eilenberg algebra of a 2-term $L_\infty$-algebra. In particular, let $(\tau_i)$ and $(\sigma_m)$  be bases of $\frw:=\frm=T_{\id_e}M$ and $\frv:=\ker(\sft)\subseteq \frn=T_{\id_{\id_e}}N$, respectively, and  let $(\check\tau^i)$ and $(\check\sigma^m)$ be the corresponding dual bases of $\frw^\vee$ and $\frv^\vee$. The equations \eqref{eq:action-dK-2} should be regarded as the evaluation of 
 \begin{equation}\label{eq:CE-identification-1}
  \begin{aligned}
   \dd_{\rm CE} \check\tau^i\ &=\ -s^i_m\check \sigma^m-\tfrac{1}{2}f^i_{jk}\check\tau^j\wedge \check\tau^k~,\\
   \dd_{\rm CE} \check\sigma^m\ &=\ -\tfrac{1}{2}c^m_{in}(\check \tau^i\wedge \check \sigma^n-\check \sigma^n\wedge \check \tau^i)+\tfrac{1}{3!}d^m_{ijk}\check\tau^i\wedge \check\tau^j\wedge \check\tau^k~,
  \end{aligned}
 \end{equation}
 at $\check \tau^i=\alpha^i$ and $\check \sigma^m=\beta^m$ with $\alpha=\alpha^i\tau_j$ and $\beta=\beta^m\sigma_m$. The constants $s^i_m$, $f^i_{jk}$, $c^m_{in}$, and $d^m_{ijk}$ are the generalised structure constants of the 2-term $L_\infty$-algebra $\frv\xrightarrow{\,\mu_1\,} \frw$:
 \begin{equation}\label{eq:CE-identification-2}
 \begin{aligned}
  \mu_1(\sigma_m)\ &=\ -s^i_m\tau_i~,\\
  \mu_2(\tau_i,\tau_j)\ &=\ f^k_{ij}\tau_k\eand \mu_2(\tau_i,\sigma_m)\ =\ c^n_{im}\sigma_n~,\\
  \mu_3(\tau_i,\tau_j,\tau_k)\ &=\ -d^m_{ijk}\sigma_m~.
  \end{aligned}
 \end{equation}
 The additional signs are included to match our overall conventions, cf.\ Remark \ref{rem:invert-L-infty}. The higher homotopy Jacobi identities follow from the fact that $\dd_{\rm CE}^2=\dd_{\rm K}^2=0$ \cite{Lada:1992wc}.

We sum up our findings in the following theorem.

\begin{theorem}\label{th:differentiate-lie-2}
 For a semistrict Lie 2-group $\CCG=(M,N)$, the functor from the category of smooth manifolds $X$ to the category of $\CCG$-valued descent data on surjective submersions $\FR^{1|0}\times X\to X$ is parameterised by elements of $\frw[-1]\oplus \frv[-2]$, where $\frv\rightarrow \frw$ is  the 2-term $L_\infty$-algebra for which $\frw:=T_{\id_e}M$ and $\frv:=\ker(\sft)\subseteq T_{\id_{\id_e}}N$. The action of the differential $\dd_{\rm K}$ on the descent data yields the Chevalley--Eilenberg differential of the 2-term $L_\infty$-algebra $\frv\rightarrow \frw$.
\end{theorem}

Analogously to Lie groups, we would like to consider an equivalent descent datum and compare the change of the resulting Chevalley--Eilenberg algebra. This will eventually give us equivalent an parameterisation $(\tilde\alpha,\tilde\beta)\in \frw[-1]\oplus \frv[-2]$ obtained from $(\alpha,\beta)\in \frw[-1]\oplus \frv[-2]$. 

\begin{lemma}\label{lem:expansion-coboundary}
Equivalent descent data $(\{\tilde m_{01}\},\{\tilde n_{012}\})$ and $(\{m_{01}\},\{n_{012}\})$ are related by a degree-2 \v Cech coboundary $(\{p_{0}:=p(\theta_0)\},\{q_{01}:=q(\theta_0,\theta_1)\})$ according to 
\begin{equation}\label{eq:GaugeTrafoStart}
\begin{aligned}
 q_{01}\, :\, \tilde m_{01}\otimes p_1\ &\Rightarrow\ p_0\otimes m_{01}~,\\
  q_{02}\circ (\tilde n_{012}\otimes\id_{p_2})\ &=\ (\id_{p_0}\otimes  n_{012})\circ\sfa_{p_0,m_{01},m_{12}}\circ(q_{01}\otimes\id_{m_{12}})\,\circ \\
  &\kern1cm \circ\, \sfa^{-1}_{\tilde m_{01},p_1,m_{12}}\circ(\id_{\tilde m_{01}}\otimes q_{12})\circ\sfa_{\tilde m_{01},\tilde m_{12},p_2}
\end{aligned}
\end{equation}
with
\begin{equation}\label{eq:expansion-pq}
 p_0\ =\ p-\dd_{\rm K} p \theta_0\eand q_{01}\ =\ \id_p+\lambda_p(\theta_0-\theta_1)-\id_{\dd_{\rm K}p}\theta_1-\dd_{\rm K}\lambda_p\theta_0\theta_1
\end{equation}
for some $p\in N$ and $\lambda_p\in T_p[-1]N$.
\end{lemma}
\begin{proof}
 The expansion for $q_{01}$ in \eqref{eq:expansion-pq} follows from $q_{00}=\id_{p_0}$, cf.\ Remark \ref{eq:trivial-2-cocycles}, together with $\dd_{\rm K}\,\id_{\dd_{\rm K}p}=0$.
\end{proof}
Note that contrary to the previously considered coboundaries, $p_0$ and $q_{01}$ are points in $M$ near $p$ and in $N$ near $\id_p$, respectively. Our formul\ae{} for linearising the structure maps at $p$ and $\id_p$, however, remain essentially the same, cf.\ Remark \ref{rem:linearising-at-p}.

Following Proposition \ref{prop:Combine2Lax2NatTrans}, we may now combine the coboundaries $(\{m_0\},\{n_{01}\})$ and $(\{p_0\},\{q_{01}\})$ appearing in  \eqref{eq:CoBounTrivExp} to a new coboundary $(\{m'_0\},\{n'_{01}\})$. The diagram
\begin{equation}
 \xymatrixcolsep{8pc}
\myxymatrix{
 e\ar@{->}[r]^{\id_e}  \ar@{->}[d]_{m_1} & e \ar@{->}[d]^{m_0}\\
 e \ar@{->}[d]_{p_1} \ar@{->}[r]^{m_{01}} \ar@{=>}[ru]^{n_{01}~}  & e \ar@{->}[d]^{p_0} \\
 e \ar@{->}[r]_{\tilde m_{01}} \ar@{=>}[ru]^{q_{01}~}  & e
}\ =\ 
\myxymatrix{
 e\ar@{->}[r]^{\id_e}  \ar@{->}[d]_{m'_1} & e \ar@{->}[d]^{m'_0}\\
 e  \ar@{->}[r]_{m'_{01}} \ar@{=>}[ru]^{n'_{01}~}  & e 
}
\end{equation}
yields the formul\ae
\begin{equation}\label{eq:GaugeTrafo-nprime}
\begin{aligned}
m'_0\ &=\ p_0\otimes m_0~,\\
n'_{01}&~:~\tilde m_{01}\otimes m'_1\Rightarrow m'_0~,\\
 n'_{01}\ &=\ (\id_{p_0}\otimes n_{01})\circ\sfa_{p_0,m_{01},m_1}\circ(q_{01}\otimes\id_{m_1})\circ \sfa^{-1}_{\tilde m_{01},p_1,m_1}~.
 \end{aligned}
\end{equation}
Hence, $\tilde n_{012}$ obeys
\begin{equation}
 n'_{02}\circ (\tilde n_{012}\otimes \id_{m'_2})\ =\ n'_{01}\circ(\id_{\tilde m_{01}}\otimes n'_{12})\circ\sfa_{\tilde m_{01},\tilde m_{12},m'_2}~.
\end{equation}

Comparing the parameterisation of the coboundary $(\{m_0\},\{n_{01}\})$ with that of $(\{m'_0\},$ $\{n'_{01}\})$ is not straightforward as their expansions in the Gra\ss mann-odd coordinates are different. In particular $m'_0$ and $n'_{01}$ are  not the same as $\tilde m_0:=\tilde m(\theta_0,0)$ and $\tilde n_{01}:=\tilde n(\theta_0,\theta_1,0)$, in general. To remedy this, we apply a modification isomorphism $\{o_0:m'_0\Rightarrow \tilde m_0\otimes p\}$, taking us from the coboundary $(\{m'_0\},\{n'_{01}\})$ to the coboundary $(\{\tilde m_0\},\{\hat n_{01}\})$:
\begin{equation}\label{eq:ModTrafNorm}
  o_0\circ n'_{01}\ =\ \hat n_{01}\circ(\id_{\tilde m_{01}}\otimes o_1)\ewith
 \{o_0:=o(\theta_0):=q^{-1}(\theta_0,0)\}~,
\end{equation}
where $\hat n_{01}:\tilde m_{01}\otimes (\tilde m_1\otimes p)\Rightarrow \tilde m_0\otimes p$. It is then easy to see that
\begin{equation}\label{eq:GaugeTrafoResMod}
  \tilde m(0)\ =\ \id_e~,\quad \hat n_{00}\ =\ \id_{\tilde m_0\otimes p}~,\eand \hat n(\theta_0,0)\ =\ \id_{\tilde m_0\otimes p}
\end{equation}
and hence, 
\begin{equation}
  \hat n_{02}\circ (\tilde n_{012}\otimes \id_{\tilde m_2\otimes p})\ =\ \hat n_{01}\circ(\id_{\tilde m_{01}}\otimes \hat n_{12})\circ\sfa_{\tilde m_{01},\tilde m_{12},\tilde m_2\otimes p}~.
\end{equation}
For $\theta_2=0$, this equation implies that  
\begin{equation}\label{eq:GaugeTrafoResMod-2}
 \tilde n_{01}\otimes \id_p\ =\ \hat n_{01}\circ \sfa_{\tilde m_{01},\tilde m_1,p}~.
\end{equation}
  
Altogether, we have thus constructed a coboundary $(\{\tilde m_0\},\{\tilde n_{01}\})$ representing the equivalent descent data $(\{\tilde m_{01},\tilde n_{012}\})\sim (\{m_{01}, n_{012}\})$ according to
\begin{equation}\label{eq:GaugeTrafoEnd}
 \tilde n_{02}\circ (\tilde n_{012}\otimes \id_{\tilde m_2})\ =\ \tilde n_{01}\circ(\id_{\tilde m_{01}}\otimes \tilde n_{12})\circ\sfa_{\tilde m_{01},\tilde m_{12},\tilde m_2}~.
 \end{equation}

These considerations then lead to the following theorem. 
\begin{theorem}\label{th:SemStrGaugeTrafo}
Let $(\{m_{01}\},\{n_{012}\})$ be a descent datum parametrised by $(\alpha,\beta)\in \frm[-1]\oplus\frn[-2]$ with $\sft(\beta)=0$. Furthermore, let $(\{\tilde m_{01}\},\{\tilde n_{012}\})$ be an equivalent descent datum  that is parametrised by $(\tilde \alpha,\tilde \beta)\in \frm[-1]\oplus\frn[-2]$ with $\sft(\tilde\beta)=0$. Then $\tilde\alpha$ and $\tilde\beta$  are expressed in terms of $\alpha$ and $\beta$ according to
\begin{subequations}
\begin{eqnarray}\label{eq:SemStrGaugeTrafo}
 \kern-20pt\lambda_p\,:\,\tilde\alpha\otimes p\! &\Rightarrow&\! p\otimes\alpha-\dd_{\rm K} p~, \label{eq:SemStrGaugeTrafo-B}\\
 \kern-20pt  \tilde\beta\otimes \id_p\! &=&\! \mu(\tilde\alpha,\tilde\alpha,p)+ \big[\id_p\otimes\beta+\mu(p,\alpha,\alpha)\big]\circ\notag\\
    &&\kern2.5cm \circ\, \big[-\dd_{\rm K}\lambda_p-\lambda_p\otimes\id_{\alpha}-\mu(\tilde\alpha,p,\alpha)\big]\circ\notag\\
      &&\kern2.5cm \circ\,\big[-\id_{\sfs(\dd_{\rm K}\lambda_p)}-\id_{\tilde\alpha}\otimes(\lambda_p+\id_{\dd_{\rm K}p})\big]~,\label{eq:SemStrGaugeTrafo-C}
\end{eqnarray}
\end{subequations}
where $p\in M$ and $\lambda_p\in T_p[-1]N$. By construction, equations \eqref{eq:action-dK-2} are invariant under this equivalence relation.
\end{theorem} 

\vspace{5pt}
\noindent
{\it Proof:} We follow the arguments around \eqref{eq:GaugeTrafoStart}--\eqref{eq:GaugeTrafoEnd} so that the expansions of $\{m_0\}$, $\{n_{01}\}$, $\{m_{01}\}$, and $\{n_{012}\}$ and $\{\tilde m_0\}$,  $\{\tilde n_{01}\}$, $\{\tilde m_{01}\}$, and $\{\tilde n_{012}\}$, are those given in Proposition \ref{prop:Expansions}, with tilded coefficients for tilded quantities. The expansion of the coboundary $(\{p_0\},\{q_{01}\})$ are given in Lemma \ref{lem:expansion-coboundary}.

Since  $q_{01}:\tilde m_{01}\otimes p_1\Rightarrow p_0\otimes m_{01}$, we find by computing the source and target and using the expansions (see also Proposition \ref{prop:Expansions} and Corollary \ref{cor:Konsdiffalbe})
\begin{equation}
\begin{aligned}
  m_{01}\ &=\ \id_e+\alpha(\theta_0-\theta_1)+\big[\alpha\otimes\alpha+\sfs(\beta)\big]\theta_0\theta_1\ =\ \id_e+\alpha(\theta_0-\theta_1)-\dd_{\rm K}\alpha\, \theta_0\theta_1~,\\
  \tilde m_{01}\ &=\ \id_e+\tilde \alpha(\theta_0-\theta_1)+\big[\tilde \alpha\otimes\tilde \alpha+\sfs(\tilde \beta)\big]\theta_0\theta_1\  =\ \id_e+\tilde \alpha(\theta_0-\theta_1)-\dd_{\rm K}\tilde \alpha\,\theta_0\theta_1~,
  \end{aligned}
\end{equation}
that
\begin{equation}
\begin{aligned}
   \lambda_p\,:\,\tilde\alpha\otimes p\ & \Rightarrow\ p\otimes\alpha-\dd_{\rm K}p~,\\
   \dd_{\rm K}\lambda_p\,:\,-\dd_{\rm K}\tilde\alpha\otimes p+\tilde\alpha\otimes \dd_{\rm K}p\ & \Rightarrow\ -\dd_{\rm K}p\otimes\alpha-p\otimes\dd_{\rm K}\alpha~,
\end{aligned}
\end{equation}
thus verifying \eqref{eq:SemStrGaugeTrafo-B}.

To compute $n'_{01}$ from \eqref{eq:GaugeTrafo-nprime}, we need to establish the explicit form of the two associators $\sfa_{p_0,m_{01},m_1}$ and $\sfa^{-1}_{\tilde m_{01},p_1,m_1}$. Both of these become trivial for $\theta_0=\theta_1$ or $\theta_1=0$. We therefore have the following expansions, 
\begin{equation}
\begin{aligned}
 \sfa_{p_0,m_{01},m_1}\ &=:\ \id_{p_0\otimes(m_{01}\otimes m_1)}+\mu(p,\alpha,\alpha)\theta_0\theta_1~,\\
  \sfa^{-1}_{\tilde m_{01},p_1,m_1}\ &=:\ \id_{(\tilde m_{01}\otimes p_1)\otimes m_1}-\mu(\tilde\alpha,p,\alpha)\theta_0\theta_1~,
 \end{aligned}
\end{equation}
defining two maps, which we both denote by $\mu$:
\begin{equation}
 \begin{aligned}
  \mu(p,\alpha,\alpha)~:&~p\otimes(\alpha\otimes \alpha)-(p\otimes\alpha)\otimes\alpha\ \Rightarrow\ 0~,\\
  \mu(\tilde \alpha,p,\alpha)~:&~\tilde\alpha\otimes(p\otimes \alpha)-(\tilde\alpha\otimes p)\otimes\alpha \ \Rightarrow\ 0~.
 \end{aligned}
\end{equation}

Upon substituting these expressions together with those for $\{p_{01}\}$, $\{q_0\}$ and $\{n_{01}\}$, $\{m_1\}$ into  \eqref{eq:GaugeTrafo-nprime}, we find
\begin{equation}\label{eq:Exnprime}
\begin{aligned}
 n'_{01}\ &=\ \id_p+(\theta_0-\theta_1)\lambda_p+\id_{p\otimes\alpha-\dd_{\rm K}p}\theta_1\,\,+\\
  &\kern1.5cm +\big[\id_p\otimes\beta+\mu(p,\alpha,\alpha)\big]\circ\big[-\dd_{\rm K}\lambda_p-\lambda_p\otimes\id_\alpha-\mu(\tilde\alpha,p,\alpha)\big]\theta_0\theta_1~.
  \end{aligned}
\end{equation}
Here, we relied on the fact that each of the terms in \eqref{eq:GaugeTrafo-nprime} can be written as $\id_p+\theta_0 \pi_1+\theta_1 \pi_2$, where $\pi_{1,2}\in T_p[-1]N$, and for these, the linearised concatenation is well-defined.
 
Finally, we perform the modification transformation $o_0 : m'_0\Rightarrow \tilde m_0\otimes p$ with $\{o_0^{-1}=q(\theta_0,0)\}$ which we have introduced in \eqref{eq:ModTrafNorm},
\begin{equation}\label{eq:modtrafopf}
  o_0\circ n'_{01}\ =\ \hat n_{01}\circ(\id_{\tilde m_{01}}\otimes o_1)\quad\Longleftrightarrow\quad
   o^{-1}_0\circ \hat n_{01}\ =\ n'_{01}\circ(\id_{\tilde m_{01}}\otimes o^{-1}_1)~,
\end{equation}
to obtain $\hat n_{01}:\tilde m_{01}\otimes (\tilde m_1\otimes p)\Rightarrow\tilde m_0\otimes p$. Using \eqref{eq:GaugeTrafoResMod-2} and $\{o_0^{-1}=q(\theta_0,0)\}$, this can be rewritten as  
\begin{equation}\label{eq:modtrafopf-2}
  q(\theta_0,0)\circ (\tilde n_{01}\otimes \id_p)\circ \sfa^{-1}_{\tilde m_{01},\tilde m_1,p} \ =\ n'_{01}\circ[\id_{\tilde m_{01}}\otimes q(\theta_1,0)]~.
\end{equation}
To evaluate this expression we need to fix the expansion of the associator, which we do according to
\begin{equation}
  \sfa^{-1}_{\tilde m_{01},\tilde m_1,p}\ =:\ \id_{(\tilde m_{01}\otimes\tilde m_1)\otimes p}-\mu(\tilde\alpha,\tilde\alpha,p)\theta_0\theta_1~,
\end{equation}
where $\mu(\tilde\alpha,\tilde\alpha,p):\tilde\alpha\otimes(\tilde\alpha\otimes p)-(\tilde\alpha\otimes \tilde\alpha)\otimes p\Rightarrow 0$. Substituting this expression, \eqref{eq:expansion-pq}, and \eqref{eq:Exnprime} into \eqref{eq:modtrafopf-2}, we find after some algebraic manipulations that $\tilde n_{01}=\id_e+\id_{\tilde\alpha}\theta_0+\tilde\beta\theta_0\theta_1$ with
\begin{equation}
\begin{aligned}
 \tilde\beta\otimes \id_p\ &=\ \mu(\tilde\alpha,\tilde\alpha,p)+ \big[\id_p\otimes\beta+\mu(p,\alpha,\alpha)\big]\,\circ\\
    &\kern1cm \circ\big[-\dd_{\rm K}\lambda_p-\lambda_p\otimes\id_{\alpha}-\mu(\tilde\alpha,p,\alpha)\big]\circ\big[-\id_{\sfs(\dd_{\rm K}\lambda_p)}-\id_{\tilde\alpha}\otimes(\lambda_p+\id_{\dd_{\rm K}p})\big]~,
\end{aligned}
\end{equation}
verifying \eqref{eq:SemStrGaugeTrafo-C}. Note that $\sft(\tilde\beta)=0$ as required.  This concludes the proof. \hfill $\Box$

\vspace{10pt}
Finally, we would like to emphasise that given $\lambda_p\in T_p [-1]N$, we can always construct a $\lambda\in \frn[-1]$ and vice versa.

\vspace{10pt}
\begin{definition}\label{def:DefOfLambda} 
 Let $p\in M$ and $\lambda_p\in T_p [-1]N$ be given as in Theorem \ref{th:SemStrGaugeTrafo}. We define a 2-cell $\lambda\in\frn[-1]$ by setting
 \begin{equation}
   \lambda\ :=\ (\lambda_p\otimes\id_{\bar p})\circ\sfa^{-1}_{\tilde\alpha,p,\bar p}~,
 \end{equation}
 that is, $\lambda:\tilde\alpha\Rightarrow (p\otimes\alpha)\otimes\bar p-\dd_{\rm K}p\otimes\bar p$, where $\bar p\in M$ with $p\otimes\bar p=\id_e=\bar p\otimes p$ and $\sfa_{\tilde\alpha,p,\bar p}: (\tilde\alpha\otimes p)\otimes\tilde p\Rightarrow \tilde\alpha\otimes (p\otimes\tilde p)$. In addition, we define a 2-cell $\lambda_0\in \frv[-1]$ by setting
 \begin{equation}
  \lambda_0\ :=\ \lambda-\id_{(p\otimes\alpha)\otimes\bar p-\dd_{\rm K}p\otimes\bar p}~,
 \end{equation}
 that is,   $\lambda_0:\tilde\alpha-(p\otimes\alpha)\otimes\bar p+\dd_{\rm K}p\otimes\bar p\Rightarrow 0$ with an intuitive notation to be understood.
\end{definition}

\begin{prop}\label{prop:PropOfLambda}
Given $\lambda$ as in Definition \ref{def:DefOfLambda}, we have
\begin{equation}
   \lambda_p\ =\ \sfa_{p\otimes\alpha+\dd_{\rm K}p,\bar p, p}\circ\big[(\lambda\circ\sfa_{\tilde \alpha,p,\bar p})\otimes\id_{p}\big]\circ\sfa^{-1}_{\tilde\alpha\otimes p,\bar p,p}~.
\end{equation}
\end{prop}

\begin{proof}
 Due to the naturalness of the associator, it is straightforward to see that $\lambda_p$ can be expressed in terms of $\lambda$ in the above way.
\end{proof}

\subsection{Example: strict Lie 2-groups}

As a consistency check, let us now consider a class of examples. Since it is notoriously difficult to construct non-trivial examples of Lie 2-groups which are not strict, we have to consider the strict case. That is, we start from descent data for strict principal 2-bundles in the general Lie 2-group framework. For such bundles, we have $n_{012}=\ell_{02}^{-1}\circ (\ell_{01}\otimes\ell_{12})$ as was discussed in Remark \ref{rem:strict2bundis}. One can check that then
\begin{equation}
 \ell_{01}\ =\ n_{01}\otimes \id_{\overline{m}_1}\ =\ \id_{\id_e}+ \id_\alpha(\theta_0-\theta_1)+(\beta+\id_{\alpha\otimes\alpha})\theta_0\theta_1~,
\end{equation}
which yields the following.
\begin{lemma}
 For strict Lie 2-groups, the functor between the category of smooth manifolds $X$ and the category of $\CCG$-valued descent data on $\FR^{0|1}\times X\to X$ reads as
\begin{equation}
\begin{aligned}\label{eq:someeqn}
   m_{01}\ &=\ \id_e+\alpha(\theta_0-\theta_1)+\theta_0\theta_1\big[\alpha\otimes\alpha+\sfs(\beta)\big]~,\\
   n_{012}\ &=\ \id_{\id_e}+\id_\alpha(\theta_0-\theta_2)+\beta(\theta_0\theta_1+\theta_1\theta_2-\theta_0\theta_2)\,+\\
  &\kern1cm+\, \id_{\alpha\otimes\alpha+\sfs(\beta)}\theta_0\theta_2+(\id_\alpha\otimes\beta-\beta\otimes\id_\alpha)\theta_0\theta_1\theta_2~,
 \end{aligned}
\end{equation}
which implies 
\begin{equation}\label{eq:someDiffFormsMapLan}
  \dd_{\rm K}\alpha\ =\ -\alpha\otimes\alpha-\sfs(\beta)\eand  \dd_{\rm K}\beta\ =\ -\id_\alpha\otimes\beta+\beta\otimes\id_\alpha~.
\end{equation}
\end{lemma}

To compare with the literature, we need to translate these results into expressions using crossed modules of Lie groups.

\begin{prop}\label{eq:inverse-identification}
 In terms of crossed modules of Lie groups $(\sH\overset{\partial}{\to}\sG,\acton)$, the functor between the category of smooth manifolds $X$ and the category of $(\sH\overset{\partial}{\to}\sG,\acton)$-valued descent data on $\FR^{0|1}\times X\to X$ is given by \v Cech 1- and 2-cochains $\{g_{01}\}$ and $\{h_{012}\}$ with values in the Lie groups $\sG$ and $\sH$, respectively. These are parameterised by $a\in\frg[-1]$ and $b\in\frh[-2]$, where $\frg$ and $\frh$ are the Lie algebras of $\sG$ and $\sH$, according to
  \begin{subequations}\label{eq:CrossedModExpans}
  \begin{equation}
  g_{01}\ =\ \unit_\sG+a(\theta_0-\theta_1)+\big\{\tfrac12 [a,a]-\partial(b)\big\}\theta_0\theta_1~
  \end{equation}
  and
  \begin{equation}
  h_{012}\ =\ \unit_\sH+b(\theta_0\theta_1+\theta_1\theta_2-\theta_0\theta_2)+(a\acton b)\theta_0\theta_1\theta_2~.
  \end{equation}
  \end{subequations}
  The action of the differential $\dd_{\rm K}$ translates to
  \begin{equation}
  \dd_{\rm K} a\ =\ -\tfrac12[a,a]+\partial(b)\eand \dd_{\rm K} b\ =\ -a\acton b~.
  \end{equation}
\end{prop}
\begin{proof}
 Starting from \eqref{eq:someeqn} and \eqref{eq:someDiffFormsMapLan}, we follow Proposition \ref{prop:equiv-strict-Lie-2-CM} and define $\sG:=M$ and $\sH=\ker(\sft)\subseteq N$. The products on $\sG$ and $\sH$, the action $\acton$ and the map $\dpar$ are defined according to equation \eqref{eq:definition-CM-ops}. We then identify
 \begin{equation}
  g_{01}\ =\ m_{01}\eand h_{012}\ =\ n_{012}\otimes \id_{\overline m_{02}}~,
 \end{equation}
 which implies $\alpha=a\in\mathfrak{g}[-1]$ and $\beta=b$. Clearly, this identification is reversible and therefore an equivalence. The cocycle relations \eqref{eq:2CellExpCocyc-B} for $(\{m_{01}\},\{n_{012}\})$ are then equivalent to those for $(\{g_{01}\},\{h_{012}\})$, cf.\ \eqref{eq:strict-cocycle}, using the identifications under Proposition \ref{prop:equiv-strict-Lie-2-CM}. In the strict case, $\alpha$ and $\beta$ take values in a 2-term $L_\infty$-algebra with trivial associator, which forms a differential crossed module. From the actions of $\dd_{\rm K}$ given in \eqref{eq:someDiffFormsMapLan} as well as equations \eqref{eq:CE-identification-1} and \eqref{eq:CE-identification-2}, we read off that the tensor products $\alpha\otimes \alpha$ and $\id_\alpha\otimes\beta-\beta\otimes\id_\alpha$ turn into the commutator and the action of $\sG$ onto $\sH$.
\end{proof}

\noindent These are the expressions that were already obtained in Jur\v co \cite{Jurco20122389}. 

Furthermore, combining the results of Theorem \ref{th:SemStrGaugeTrafo} and Definition \ref{def:DefOfLambda} with the interchange law \eqref{eq:interchange_law}, we arrive after a few algebraic manipulations at 
\begin{equation}\label{eq:strictGT-1}
\begin{aligned}
  \lambda\,:\,\tilde\alpha\ &\Rightarrow\ p\otimes\alpha\otimes \bar p-\dd_{\rm K}p\otimes \bar p~,\\
  \tilde\beta\ &=\ [\id_p\otimes\beta\otimes\id_{\bar p}]\circ[-\dd_{\rm K}\lambda-\lambda\otimes\lambda]~.
\end{aligned}
\end{equation}
Translated into crossed modules of Lie groups, this takes the following form.

\begin{prop}
Let  $(\{g_{01}\},\{h_{012}\})$ be a descent datum that is parameterised by $a\in \frg[-1]$ and $b\in \frh[-2]$. Furthermore, let $(\{\tilde g_{01}\},\{\tilde h_{012}\})$ be an equivalent descent datum that is parameterised by $\tilde a\in \frg[-1]$ and $\tilde b\in \frh[-2]$. Then, $(a,b)$ and $(\tilde a, \tilde b)$ are related by the following equations:
\begin{subequations}\label{eq:strictGT-all}
\begin{eqnarray}
   \tilde a&=& pap^{-1}+p\,\dd_{\rm K}p^{-1}-\partial(\lambda^\frh)~,\label{eq:strictGT-a}\\
  \tilde b&=& p\acton b-\dd_{\rm K}\lambda^\frh-\tilde a\acton\lambda^\frh-\tfrac12[\lambda^\frh,\lambda^\frh]\label{eq:strictGT-b}
\end{eqnarray}
 \end{subequations}
 for $p\in G$ and $\lambda^\frh \in\frh[-1]$.
\end{prop}
\begin{proof}
 We again follow Proposition \ref{prop:equiv-strict-Lie-2-CM}, which justifies the appearance of $p$ in \eqref{eq:strictGT-all} after identifying
 \begin{equation}
  \tilde a\ =\ \tilde \alpha~,~~~\tilde b\ =\ \tilde \beta~,\eand  \lambda^\frh\ =\ \lambda-\id_{p\otimes\alpha\otimes \bar p-\dd_{\rm K}p\otimes \bar p}~.
 \end{equation}
 More specifically, \eqref{eq:strictGT-a} immediately follows from computing $\sfs(\lambda^\frh)=-\dpar(\lambda^\frh)$. Recall that  $\id_p\otimes\beta\otimes\id_{\bar p}$ translates to $p\acton b$. Using Proposition \ref{prop:induced-concatenation} together with the identity $\sfs(\beta)=-\dd_{\rm K}\alpha-\alpha\otimes \alpha$, we can derive \eqref{eq:strictGT-b} by a lengthy but straightforward computation from the second equation in \eqref{eq:strictGT-1}.
\end{proof}

\subsection{Comment on differentiation and categorical equivalence}

Recall from Proposition \ref{prop:classify-Lie-2-groups} that every weak 2-group is categorically equivalent to a special weak 2-group given in terms of a group $\sG$, an Abelian group $\sH$, a representation $\varrho$ of $\sG$ on $\sH$, and an element $[\sfa]\in H^3(\sG,\sH)$. The corresponding Proposition  \ref{prop:classify-Lie-2-algebras} for Lie 2-algebras states that semistrict Lie 2-algebras are categorically equivalent to special Lie 2-algebras given in terms of a Lie algebra $\frg$, a representation $\varrho$ of $\frg$ on a vector space $\frv$, and an element $[\sfJ]\in H^3(\frg,\frv)$.

It is now tempting to assume that the natural integration process factors through categorical equivalence and therefore special Lie 2-algebras can be integrated to special Lie 2-groups. However, Baez \& Lauda proved a no-go theorem \cite[Section 8.5]{Baez:0307200}, which shows that certain special Lie 2-algebras can be integrated to 2-groups, which, however, can be turned into topological 2-groups only for the strict case $\sfa =0$. In particular, consider the case of a special Lie 2-algebra with $\frv=\au(1)$. We have $H^3(\frg,\au(1))\cong \FR$. The latter contains a lattice $\cong \RZ$, which can be embedded into $H^3(\sG,\sU(1))$, yielding the integration to a 2-group. In the topological case, however, we have to use continuous cohomology, for which $H^3_{\rm cont.}(\sG,\sU(1))=0$.

The differentiation of Lie 2-groups we performed in this section is the inverse operation to this integration. As integration does not commute with categorical equivalence, neither will differentiation.

\section{Semistrict higher gauge theory}\label{sec:Semistrict_HGT}

We now put the results of the previous section together and develop a description of semistrict principal 2-bundles with connective structure. We first discuss the local case\footnote{For more details on the local case, see Sati, Schreiber \& Stasheff \cite{Sati:0801.3480}.}, which can be readily derived from the Maurer--Cartan equation of an $L_\infty$-algebra. We then give the global description in terms of non-Abelian Deligne cohomology sets.

As before, let $X$ be a smooth manifold with covering $\frU=\{U_a\}$ and let $U\subseteq X$ be an open subset of $X$. Furthermore, let $\Omega^p_X$ be the sheaf of smooth differential $p$-forms on $X$ and set $\Omega^\bullet_X=\bigoplus_{p\geq 0} \Omega^p_X$. 

\subsection{Local semistrict higher gauge theory}

Local semistrict higher gauge theory corresponds to the Maurer--Cartan equation \eqref{eq:MCeqs} for a degree-$1$ element of the $L_\infty$-algebra arising from the tensor product of $\Omega^\bullet_X$ and a gauge $L_\infty$-algebra $L$. The corresponding infinitesimal gauge transformations are the gauge transformations of the Maurer--Cartan equation
\eqref{eq:MCgaugetrafos}. To make this explicit, we wish to recall the following proposition.

\begin{prop}
 A tensor product of a differential graded algebra $\fra$ and an $L_\infty$-algebra $L$ comes with a natural $L_\infty$-structure. The grading of an element of $\fra\otimes L$ is the sum of its individual gradings. Moreover, for a tuple of elements $(a_1\otimes\ell_1,\ldots,a_i\otimes \ell_i)$ of $\fra\otimes L$, the higher products $\tilde{\mu}_i$ read as
 \begin{equation}
  \tilde\mu_i(a_1\otimes\ell_1,\ldots,a_i\otimes \ell_i)\ =\ 
  \left\{\begin{array}{ll}
  (\dd a_1) \otimes \ell_1+(-1)^{\deg(a_1)}a_1\otimes \mu_1(\ell_1)&\efor i=1~,\\
  \chi(a_1a_2\cdots a_i\otimes \mu_i(\ell_1,\ldots,\ell_i))&\efor i>1~.\end{array}\right.
 \end{equation}
 Here, the $\mu_i$ are the higher products in $L$, $\deg$ denotes the degrees in $\fra$, and $\chi=\pm1$ is the so-called Koszul sign arising from moving graded elements of $\fra$ past graded elements of $L$.
\end{prop}

\begin{proof}
 The higher homotopy Jacobi identities, displayed in the appendix in  \eqref{eq:homotopyJacobi}, for the higher products $\tilde \mu_i$ are readily checked.
\end{proof}

\begin{exam}\label{ex:tensor-product-sh}
As an example, let us work out the details for the case where $\fra$ is the de Rham complex on $X$ and $L$ is a 2-term $L_\infty$-algebra. Let $U\subseteq X$ be an open subset.
 The tensor product of $H^0(U,\Omega^\bullet_X)$ and the 2-term $L_\infty$-algebra $\frv\xrightarrow{\,\mu_1\,} \frw$ consists of the following graded subspaces
 \begin{equation}
  H^0(U,\Omega^\bullet_X)\otimes (\frv\xrightarrow{\,\mu_1\,} \frw)\ \cong\ \underbrace{H^0(U,\Omega^0_X\otimes \frv)}_{\mbox{degree -1}}~\oplus~\bigoplus_{p\geq 0}\underbrace{H^0(U,\Omega^p_X\otimes \frw\oplus\Omega^{p+1}_X\otimes \frv)}_{\mbox{degree $p$}}~.
 \end{equation}
  For $\phi\in H^0(U,\Omega^1_X\otimes \frw\oplus\Omega^{2}_X\otimes \frv)$, the homotopy Maurer--Cartan equation \eqref{eq:MCeqs} reads as 
  \begin{equation}\label{eq:hMC1}
  -\tilde{\mu}_1(\phi)-\tfrac{1}{2}\tilde{\mu}_2(\phi,\phi)+\tfrac{1}{3!}\tilde{\mu}_3(\phi,\phi,\phi)\ =\ 0~.
  \end{equation}
  This equation is invariant under the (infinitesimal) transformations
  \begin{equation}\label{eq:hMCt1}
  \delta\phi\ =\ \tilde{\mu}_1(\gamma)-\tilde{\mu}_2(\gamma,\phi)-\tfrac{1}{2}\tilde{\mu}_3(\gamma,\phi,\phi)
  \end{equation}
  for $\gamma\in H^0(U,\Omega^0_X\otimes \frw\oplus\Omega^{1}_X\otimes \frv)$.
\end{exam}

\begin{prop}\label{prop:hMC-eqn-and-gauge}
 The homotopy Maurer--Cartan equation \eqref{eq:hMC1} and the transformations \eqref{eq:hMCt1} are equivalent to the equations
  \begin{equation}\label{eq:field_equations}
  \begin{aligned}
  \CF&\ :=\ \dd A+\tfrac{1}{2}\mu_2(A,A)-\mu_1(B)\ =\ 0~,\\
  H&\ :=\ \dd B+\mu_2(A,B)-\tfrac{1}{3!}\mu_3(A,A,A)\ =\ 0~,
  \end{aligned}
  \end{equation}
 where $A\in H^0(U,\Omega^1_X\otimes \frw)$ and $B\in H^0(U,\Omega^2_X\otimes \frv)$ and
  \begin{equation}\label{eq:gauge_trafos}
  \begin{aligned}
    \delta A\ &=\ \dd \omega+\mu_2(A,\omega)-\mu_1(\Lambda)~,\\
    \delta B\ &=\ -\dd \Lambda-\mu_2(A,\Lambda)+\mu_2(B,\omega)+\tfrac{1}{2}\mu_3(\omega,A,A)~,
  \end{aligned}
  \end{equation}
 where $\omega\in H^0(U,\Omega^0_X\otimes \frw)$ and $\Lambda\in H^0(U,\Omega^1_X\otimes \frv)$.
\end{prop}
\begin{proof}
 This trivially follows by identifying $\phi=A-B$ and $\gamma=\omega+\Lambda$ in \eqref{eq:hMC1} and \eqref{eq:hMCt1}.
\end{proof}

Let us now generalise from gauge potential 1- and 2-forms $A$ and $B$ satisfying the Maurer--Cartan equation to general kinematic data for local semistrict higher gauge theory. It makes sense to relax the equation $H=0$: a trivial calculation shows that in this case, $H$ transforms under under gauge transformations \eqref{eq:gauge_trafos} covariantly according to $\delta H=\mu_2(H,\omega)$. There are a number of reasons, however, why we cannot relax $\CF=0$. Firstly, consistency of the underlying parallel transport requires $\CF$ to vanish, just as it did in the strict case. Secondly, the above covariant transformation law is broken for non-vanishing $\CF$, which makes it impossible to impose a self-duality condition on $H$. Such a condition, however, is expected to arise in the $\CN=(2,0)$ superconformal field theory in six dimensions. We therefore arrive at the following definition.
\begin{definition}
 The kinematic datum of \uline{local semistrict higher gauge theory} with underlying 2-term $L_\infty$-algebra $\frv\xrightarrow{\,\mu_1\,} \frw$ is given by potential 1- and 2-forms $A\in H^0(U,\Omega^1_X\otimes \frw)$ and $B\in H^0(U,\Omega^2_X\otimes \frv)$, for which the 2-form \uline{fake curvature} $\CF:=\dd A+\tfrac{1}{2}\mu_2(A,A)-\mu_1(B)$ vanishes. An equivalence relation between such kinematic data is generated by the infinitesimal gauge transformations described in equations \eqref{eq:gauge_trafos}.
\end{definition}

\begin{rem}
 For trivial $\mu_3$, the equations \eqref{eq:field_equations} reduce to the field equations for a flat connective structure of a principal 2-bundle with strict structure 2-group and equations \eqref{eq:gauge_trafos} describe infinitesimal gauge transformations. 
 
 Note also that there are equivalence relations between gauge transformations which have the same effect on $A$ and $B$. These are given by
  \begin{equation}
  \delta \omega\ =\ \mu_1(\sigma)\eand \delta \Lambda\ =\ \dd \sigma+\mu_2(A,\sigma)~,
  \end{equation}
  where $\sigma\in H^0(U,\Omega^0_X\otimes \frv)$.
\end{rem}

\begin{rem}
 Finally, we would like to stress that the kinematic data, the local flatness conditions and the infinitesimal gauge transformations for local semistrict higher gauge theory based on an $n$-term $L_\infty$-algebras $L$ are similarly derived by considering the tensor product of $\Omega^\bullet_X$ with $L$.
\end{rem}

\subsection{Finite gauge transformations}

Having derived curvature and infinitesimal gauge transformation for semistrict higher gauge theory, let us now turn to the finite gauge transformations. Here, we rely on the results of Section \ref{sec:Lie-functor}, and the lift to Lie $n$-algebra valued potential and curvature forms is readily obtained. 

In Proposition \ref{prop:pre-gauge}, we showed that the equation $\dd_{\rm K} a+\tfrac{1}{2}[a,a]=0$ was invariant under $a\mapsto \tilde a=pap^{-1}+p\dd_{\rm K}p^{-1}$. Since $\dd_{\rm K}$ and the de Rham differential $\dd$ have the same algebraic properties, we derived the well-known statement

\begin{prop} 
 If a local connection 1-form $A$ taking values in the Lie algebra of a Lie group $\sG$ is flat, its curvature $F:=\dd A+\tfrac12[A,A]=0$ is invariant under the transformation
 \begin{equation}\label{eq:gauge_trafos_Lie-1}
  A\ \mapsto\  \tilde A\ =\ pAp^{-1}+p\dd p^{-1}
 \end{equation}
 for any $p\in H^0(U,\sG)$. Such transformations are called \uline{gauge transformations}.
\end{prop}

\noindent Note also the following consequence.
\begin{cor}
 At the infinitesimal level, the transformations \eqref{eq:gauge_trafos_Lie-1} amount to
 \begin{equation}
  A\ \mapsto\ \tilde A\ =\ \dd \pi+[A,\pi]~,
 \end{equation}
 where $\pi\in H^0(U,\Omega^0_X\otimes\frg)$. They match the gauge transformations in Proposition \ref{prop:hMC-eqn-and-gauge} for the 2-term $L_\infty$-algebra $\{0\}\rightarrow \frg$.
\end{cor}

Analogously, we treat the kinematic data of local semistrict higher gauge theory. In Theorem \ref{th:SemStrGaugeTrafo}, we showed that the equations 
\begin{equation}
 \dd_{\rm K} \alpha\ =\ -\alpha\otimes\alpha- \sfs(\beta)\eand\dd_{\rm K} \beta\ =\ -\id_\alpha\otimes\beta+\beta\otimes \id_\alpha-\mu(\alpha,\alpha,\alpha)
\end{equation}
are invariant under (\ref{eq:SemStrGaugeTrafo}) and (\ref{eq:SemStrGaugeTrafo-C}). Again, since $\dd_{\rm K}$ and $\dd$ have the same algebraic properties, we have derived the following statement.
 
\begin{prop}\label{prop:semistrict-gauge-trafos}
If the curvatures $\CF$ and $H$ of local gauge potential 1- and 2-forms $A$ and $B$ as defined in Proposition \ref{prop:hMC-eqn-and-gauge} vanish, then they are invariant under the transformation 
\begin{subequations}\label{eq:FiniteSemStrGaugeTrafo}
\begin{eqnarray}
 \kern-20pt\Lambda_p\,:\,\tilde A\otimes p\! &\Rightarrow&\! p\otimes A-\dd p~, \label{eq:FiniteSemStrGaugeTrafo-B}\\
 \kern-20pt  \tilde B\otimes \id_p\! &=&\! \mu(\tilde A,\tilde A,p)+ \big[\id_p\otimes B+\mu(p,A,A)\big]\circ\notag\\
    &&\kern2.5cm \circ\, \big[-\dd\Lambda_p-\Lambda_p\otimes\id_{A}-\mu(\tilde A,p,A)\big]\circ\notag\\
      &&\kern2.5cm \circ\,\big[-\id_{\sfs(\dd\Lambda_p)}-\id_{\tilde A}\otimes(\Lambda_p+\id_{\dd p})\big]~,\label{eq:FiniteSemStrGaugeTrafo-C}
\end{eqnarray}
\end{subequations}
where $p\in H^0(U,M)$ and\footnote{Here, $T_pN$ denotes the sheaf over $U\subseteq X$ with stalks $T_{p(x)}N$ over $x\in U$.} $\Lambda_p\in H^0(U,\Omega^1_X\otimes T_pN)$. We shall refer to such transformations as \uline{gauge transformations}.
\end{prop}  
\noindent As a consistency check, we can linearise these gauge transformations, obtaining the transformations \eqref{eq:gauge_trafos}:
\begin{prop}
 At the infinitesimal level, the gauge transformations \eqref{eq:FiniteSemStrGaugeTrafo} become
 \begin{equation}\label{eq:gauge_trafos1}
  \begin{aligned}
    \delta A\ &=\ \dd w+\mu_2(A,w)-\mu_1(v)~,\\
    \delta B\ &=\ -\dd v-\mu_2(A,v)+\mu_2(B,w)+\tfrac{1}{2}\mu_3(w,A,A)~,
  \end{aligned}
  \end{equation}
 where $w\in H^0(U,\Omega^0_X\otimes\frw)$ and $v\in H^0(U,\Omega^1_X\otimes\frv)$.
Hence, they agree with the gauge transformations in Proposition \ref{prop:hMC-eqn-and-gauge} for the 2-term $L_\infty$-algebra $\frv\xrightarrow{\,\mu_1\,} \frw$ concentrated in degrees -1 and 0.
\end{prop}
\begin{proof}
 We linearise $p=\id_e+\delta p$ and $\Lambda=\id_A+\delta \Lambda$ such that equation \eqref{eq:FiniteSemStrGaugeTrafo-B} reads as
 \begin{equation}
  (\id_A+\delta\Lambda)\ :\ (A+\delta A)\otimes (\id_e+\delta p)\ \Rightarrow\ (\id_e+\delta p)\otimes A-\dd \delta p~.
 \end{equation}
 Identifying
 \begin{equation}
 w\ =\ -\delta p\eand 
  v\ =\ \delta \Lambda-\id_{\delta p\otimes A-\dd \delta p}\ :\ \delta A+A\otimes \delta p -\delta p\otimes A+\dd \delta p \ \Rightarrow\ 0~,
 \end{equation}
 we immediately obtain the first equation in \eqref{eq:gauge_trafos1}. The derivation of the second equation in \eqref{eq:gauge_trafos1} from linearising \eqref{eq:FiniteSemStrGaugeTrafo-C} is somewhat more involved. We start from
 \begin{equation}
  \begin{aligned}
    (B+\delta B)&\otimes (\id_A+\id_{\delta p})\ =\ \mu(A,A,\delta p)+[\id_e\otimes B+\id_{\delta p}\otimes B+\mu(\delta p,A,A)]\circ\\
    &\circ[-\dd\,\id_A-\dd \delta \Lambda -\id_A\otimes \id_A-\delta\Lambda\otimes \id_A+\mu(A,w,A)]\circ\\
    &\circ[-\id_{\dd A}-\id_{\sfs(\dd \delta \Lambda)}-\id_A\otimes \id_A-\id_A\otimes \id_{\dd \delta p}-\id_A\otimes \delta \Lambda-\id_{\delta A}\otimes \id_A]~.
  \end{aligned}
 \end{equation}
The remaining calculation is rather lengthy but straightforward, if one makes use of the (linearised) interchange law, Proposition \ref{prop:induced-concatenation} and the identity $\sfs(B)=-\dd A+A\otimes A$.
\end{proof}

\subsection{Connective structure}

Consider a semistrict principal 2-bundle $\Phi$ with a semistrict structure 2-group $\CCG=(M,N)$ over a smooth manifold $X$ with covering $\frU=\{U_a\}$. We use again the notation $\frw:=T_{\id_e}M$ and $\frv:=\ker(\sft)\subseteq T_{\id_{\id_e}}N$. The bundle $\Phi$ is characterised by $\CCG$-valued  transition functions $(\{m_{ab}\},\{n_{abc}\})$. Next, we would like to equip $\Phi$ with a connective structure. 

From the discussion of strict principal 2-bundles, it is clear that a connective structure will consist locally of a $\frw$-valued 1-form $A_a$, a $\frv$-valued 2-form $B_a$, and, on intersections $U_a\cap U_b$, a $T_{m_{ab}}[-1]N$-valued 1-form $\Lambda_{ab}$. On intersections of patches $U_a\cap U_b$, $(A_a,B_a)$ and $(A_b,B_b)$ are related by a gauge transformation on $U_a\cap U_b$, which is parameterised by $(m_{ab},\Lambda_{ab})$. The explicit formula is then clear from Proposition \ref{prop:semistrict-gauge-trafos} and reads as follows:
\begin{eqnarray}\label{eq:cocycle1}
 \kern-20pt\Lambda_{ab}\,:\,\tilde A_b\otimes m_{ab}\! &\Rightarrow&\! m_{ab}\otimes A-\dd m_{ab}~,\\ 
 \kern-20pt  B_b\otimes \id_{m_{ab}}\! &=&\! \mu(A_b,A_b,m_{ab})+ \big[\id_{m_{ab}}\otimes B_a+\mu(m_{ab},A_a,A_a)\big]\circ\notag\\
    &&\kern2.5cm \circ\, \big[-\dd\Lambda_{ab}-\Lambda_{ab}\otimes\id_{A_a}-\mu(\tilde A_b,m_{ab},A_a)\big]\circ\notag\\
      &&\kern2.5cm \circ\,\big[-\id_{\sfs(\dd\Lambda_{ab})}-\id_{\tilde A_b}\otimes(\Lambda_{ab}+\id_{\dd m_{ab}})\big]~,
\end{eqnarray}
provided the fake curvature $\CF_a:=\dd A_a+A_a\otimes A_a+\sfs(B_a)$ vanishes on all coordinate patches $U_a$.
  
Note that the condition that two transformations of the form \eqref{eq:cocycle1} combine to a third one on non-empty triple intersection of coordinate patches yields the cocycle condition for  $\{\Lambda_{ab}\}$. To derive this condition, let us consider 
\begin{equation}
 \begin{aligned}
  \Lambda_{ab}\,:\,A_b\otimes m_{ba} \ &\Rightarrow\   m_{ba}\otimes A_a-\dd m_{ba}~,\\
  \Lambda_{bc}\,:\,A_c\otimes m_{cb} \ &\Rightarrow \ m_{cb}\otimes A_b-\dd m_{cb}~,\\
  \Lambda_{ac}\,:\,A_c\otimes m_{ca} \ &\Rightarrow \  m_{ca}\otimes A_a-\dd m_{ca}~,\\
 \end{aligned}
\end{equation}
over a non-empty triple intersections $U_a\cap U_b\cap U_c$. Recall also that
\begin{equation}
  n_{abc}\,:\,m_{ab}\otimes m_{bc}\ \Rightarrow\  m_{ac}~.
\end{equation}
Chasing the commutative diagram relating the two possible ways of going from $(A_b\otimes m_{ba})\otimes m_{ac}$ to $m_{bc} \otimes A_c -\dd m_{bc}$, we obtain the following proposition.

\begin{prop}
 The 1-forms $\{\Lambda_{ab}\}$ are consistent over triple overlaps $U_a\cap U_b\cap U_c$, if the following holds:
 \begin{equation}
  \begin{aligned}\label{eq:CompOfGaugeTrafos}
   &\Lambda_{cb}\circ (\id_{A_{b}}\otimes n_{bac})\circ \sfa_{{A_b},{m_{ba}},{m_{ac}}}\ =\\
   &\hspace{0.5cm}\ =\ (n_{bac}\otimes\id_{A_{c}}-\dd n_{bac})\circ (\sfa^{-1}_{{m_{ba}},{m_{ac}},A_{c}}-\id_{\dd({m_{ba}}\otimes{m_{ac}})}) \circ\\
   &\hspace{1.5cm}\circ(\id_{{m_{ba}}}\otimes\Lambda_{ca}- \id_{\dd{m_{ba}}\otimes{m_{ac}}})\circ(\sfa_{{m_{ba}},{A_{c}}, m_{ac}}-\id_{\dd{m_{ba}}\otimes{m_{ac}}})\circ(\Lambda_{ab}\otimes\id_{m_{ac}})~.
  \end{aligned}
  \end{equation}
\end{prop} 

\noindent 
In the above equation, we have again used our intuitive notation: for instance, $n_{bac}\otimes\id_{A_c}-\dd n_{bac}$ has to be understood as
\begin{equation}
n_{bac}\otimes\id_{A_c}-\dd n_{bac}:(m_{ba}\otimes m_{ac})\otimes A_c -\dd(m_{ba}\otimes m_{ac}) \ \Rightarrow\ m_{bc}\otimes A_c-\dd m_{bc}~.  
\end{equation}

We now have all the ingredients for defining the notion of a connective structure.
\begin{definition}
A \uline{connective structure} on a semistrict principal 2-bundle $\Phi$ with semistrict structure 2-group $\CCG=(M,N)$ with associated 2-term $L_\infty$-algebra $\frv\xrightarrow{\,\mu_1\,} \frw$ consists of $(\{A_a\},\{B_a\},\{\Lambda_{ab}\})$, where $A_a\in H^0(U_a,\Omega^1_X\otimes \frw)$, $B_a\in H^0(U_a,\Omega^2_X\otimes \frv)$, and $\Lambda_{ab}\in H^0(U_a\cap U_b,\Omega^1_X\otimes T_{m_{ab}}N)$ such that the cocycle conditions \eqref{eq:cocycle1}  as well as \eqref{eq:CompOfGaugeTrafos} are satisfied, and, in addition, the 2-form \uline{fake curvature}
\begin{equation}\label{eq:fakecurv}
 \CF_{a}\ :=\ \dd A_{a}+A_{a}\otimes A_{a}+\sfs(B_{a})
\end{equation}
vanishes.
\end{definition}

\begin{rem}\label{rem:LamVsLam0}
Note that by virtue of Definition \ref{def:DefOfLambda} and Proposition \ref{prop:PropOfLambda}, we can always work with a $\Lambda^0_{ab}\in H^0(U_a\cap U_b,\Omega^1_X\otimes \frv)$ such that
\begin{equation}
        \Lambda_{ab}^0\,:\, A_b-(m_{ab}\otimes A_a)\otimes\overline{m}_{ab}+\dd m_{ab}\otimes \overline{m}_{ab}\ \Rightarrow\ 0
\end{equation}
instead of $\Lambda_{ab}\in H^0(U_a\cap U_b,\Omega^1_X\otimes T_{m_{ab}}N)$ with \eqref{eq:cocycle1}.  Both are related by
\begin{subequations}
\begin{equation}
 \Lambda_{ab}^0\ =\ (\Lambda_{ab}\otimes\id_{\overline{m}_{ab}})\circ\sfa^{-1}_{A_b,m_{ab},\overline{m}_{ab}}-\id_{(m_{ab}\otimes A_a)\otimes\overline{m}_{ab}-\dd m_{ab}\otimes\overline{m}_{ab}}~,
\end{equation}
or, equivalently,
\begin{equation}
\begin{aligned}
  \Lambda_{ab}\ &=\ \sfa_{m_{ab}\otimes A_a+\dd m_{ab},\overline{m}_{ab},m_{ab}}\,\circ\\
  &\kern1cm \circ
  \left\{\left[ \left(\Lambda^0_{ab}+\id_{(m_{ab}\otimes A_a)\otimes\overline{m}_{ab}-\dd m_{ab}\otimes\overline{m}_{ab}} \right)\circ\sfa_{A_b,m_{ab},\overline{m}_{ab}} \right]\otimes\id_{m_{ab}}\right\}\circ\\
  &\kern1cm \circ \sfa^{-1}_{A_b\otimes m_{ab},\overline{m}_{ab},m_{ab}}~.
  \end{aligned}
\end{equation}
\end{subequations}
Therefore, we can say that a connective structure  $(\{A_a\},\{B_a\},\{\Lambda_{ab}\})$ is alternatively given by a tuple  $(\{A_a\},\{B_a\},\{\Lambda^0_{ab}\})$ in which $\Lambda^0_{ab}$ is as above.
\end{rem}

Finally, we would like to describe the action of a coboundary on a connective structure $(\{A_a\},\{B_a\},\{\Lambda_{ab}\})$. For $(\{A_a\},\{B_a\})$ this is again clear from Proposition \ref{prop:semistrict-gauge-trafos}. For instance,
\begin{equation}
 \begin{aligned}
  \Lambda_{a}\,:\,\tilde A_a\otimes  m_{a} \ &\Rightarrow\   m_{a}\otimes A_a-\dd m_{a}~,\\
  n_{ab}\,:\,\tilde m_{ab}\otimes m_b \ &\Rightarrow\   m_{a}\otimes m_{ab}~.  
 \end{aligned}
\end{equation}
To derive the action on  $\{\Lambda_{ab}\}$, we compare the two expressions,
\begin{equation}\label{eq:CoboundaryAction}
 \begin{aligned}
  \Lambda_{ab}\,:\,A_b\otimes  m_{ba} \ &\Rightarrow\   m_{ba}\otimes A_a-\dd m_{ba}~,\\
  \tilde \Lambda_{ab}\,:\,\tilde A_b\otimes \tilde{m}_{ba} \ &\Rightarrow\   \tilde m_{ba}\otimes \tilde A_a-\dd \tilde m_{ba}~.
 \end{aligned}
\end{equation}
Again, chasing the corresponding commutative diagram relating the two possible ways of going from $(\tilde A_a\otimes m_a)\otimes m_{ab}$ to $(\tilde m_{ab}\otimes \tilde A_b)\otimes m_b-\dd \tilde m_{ab}\otimes m_b$ yields the following proposition.

\begin{prop}  The 1-forms $\{\Lambda_{ab}\}$ and $\{\tilde \Lambda_{ab}\}$ of two equivalent connective structures $(\{A_a\},\{B_a\},\{\Lambda_{ab}\})$ and $(\{\tilde A_a\},\{\tilde B_a\},\{\tilde\Lambda_{ab}\})$ on a semistrict principal 2-bundle $\Phi$ with semistrict structure 2-group are related by
 \begin{equation}\label{eq:GaugeOfGaugeTrafos}
  \begin{aligned}
   &(\tilde \Lambda_{ba}\otimes \id_{m_b})\circ \sfa^{-1}_{\tilde{A}_a,\tilde m_{ab},{m_b}}\circ (\id_{\tilde{A}_a}\otimes n_{ab})\circ \sfa_{\tilde{A}_a,m_a,m_{ab}}\ =\ \\
   &=\ (\sfa^{-1}_{\tilde m_{ab},\tilde A_b,m_b}-\id_{\dd\tilde{m}_{ab}\otimes m_b})\circ(\id_{\tilde{m}_{ab}}\otimes\Lambda^{-1}_{b}-\id_{\dd\tilde{m}_{ab}\otimes m_b})\circ(\sfa_{m_{ab},m_b,A_b}-\id_{\dd(\tilde{m}_{ab}\otimes m_b)})\,\circ\\ 
   &\hspace{0.5cm}\circ(n_{ab}^{-1}\otimes\id_{A_b}-\dd n_{ab}^{-1})\circ(\sfa^{-1}_{m_a,m_{ab},A_b}-\id_{\dd(m_a\otimes m_{ab})})\circ (\id_{m_a}\otimes\Lambda_{ba}-\id_{\dd m_a\otimes m_{ab}})\,\circ\\ 
  &\hspace{0.5cm}\circ(\sfa_{m_a,A_a,m_{ab}}-\id_{\dd m_a\otimes m_{ab}})\circ(\Lambda_a\otimes\id_{m_{ab}})~.
  \end{aligned}
  \end{equation}
\end{prop}
\noindent As before, we have used our intuitive notation here.

\subsection{Semistrict non-Abelian Deligne cohomology}\label{ssec:Deligne}

Deligne cohomology describes gauge configurations on a principal bundle with connection modulo gauge transformations, which act simultaneously on the connection and the transition functions of the bundle. Deligne cohomology for categorified bundles was described previously in some special cases. In particular,  the case of Abelian gerbes was discussed in \cite{0817647309}, the case of principal 2-bundles with strict structure 2-group was given in \cite{Schreiber:2008aa}, and the case of principal 3-bundles was presented in \cite{Saemann:2013pca} (see also \cite{Wang:2013dwa}). Here, we wish to describe the low-lying sets of the Deligne cohomology with values in a semistrict Lie 2-group. In the special case of the 2-group $\sB\sU(1)$, this reduces to ordinary, Abelian Deligne cohomology.

As before, we consider a smooth manifold $X$ with covering $\frU=\{U_a\}$. We shall write $C^{p,q}(\frU,\CS)$ for the $\Omega_X^q\otimes\CS$-valued \v Cech $p$-cochains relative to the covering $\frU$, where $\CS$ is a some sheaf on $X$. Now, let $\CCG=(M,N)$ be a semistrict Lie 2-group. We again make  use of  the notation $\frw:=T_{\id_e}M$ and $\frv:=\ker(\sft)\subseteq T_{\id_{\id_e}}N$ and denote the corresponding 2-term $L_\infty$-algebra by $\frv\xrightarrow{\,\mu_1\,} \frw$. 

\begin{definition}
Let $\CCG=(M,N)$ be a semistrict Lie 2-group with underlying 2-term $L_\infty$-algebra $\frv\xrightarrow{\,\mu_1\,} \frw$.  A $\CCG$-valued \uline{degree-$p$ Deligne cochain}   consists of elements
 \begin{equation}
 \begin{aligned}
(\{n_{a_0\cdots a_p}\},\ldots, \{n_{a_0}\})\ &\in\ C^{p,0}(\frU, N)\times C^{p-1,1}(\frU,\frv)\times \cdots \times C^{0,p}(\frU,\frv)~,\\
(\{m_{a_0\cdots a_{p-1}}\},\ldots, \{m_{a_0}\})\ &\in\ C^{p-1,0}(\frU, M)\times C^{p-2,1}(\frU,\frw)\times \cdots \times C^{0,p-1}(\frU,\frw)~.
 \end{aligned}
 \end{equation}
\end{definition}

The sum of the \v Cech and de Rham degrees of $(\{n_{a_0\cdots a_p}\},\ldots, \{n_{a_0}\})$ is $p$ while for 
$(\{m_{a_0\cdots a_{p-1}}\},\ldots, \{m_{a_0}\})$  it is $p-1$. Compared to the analogous discussions of Deligne cochains for strict 2-groups in Schreiber \& Waldorf \cite{Schreiber:2008aa}, we have dropped \v Cech cochains that are always cohomologous to trivial ones, cf.\ \cite{Saemann:2013pca} and Proposition \ref{prop:Normalisation}.

Using our results from the previous sections as well as Appendix \ref{app:B}, we can describe Deligne cohomology with semistrict 2-groups up to degree 2. In particular, we have provided ample motivation for giving the following definition.

\begin{definition}\label{def:Deligne-cocycles}
 A \uline{degree-$p$ Deligne cocycle} is a degree-$p$ Deligne cochain satisfying a cocycle relation. Here, we restrict ourselves to the case $p\leq 2$, and define the following:
\begin{enumerate}[(i)]\setlength{\itemsep}{-1mm}
  \item A \uline{degree-$0$ Deligne cocycle} is an element $\{n_a\}\in C^{0,0}(\frU,N)$ such that on non-empty   intersections $U_a\cap U_b$
  \begin{equation}
   n_a\ =\ n_b~.
  \end{equation}
  \item A \uline{degree-$1$ Deligne cocycle} consists of elements $\{n_{ab}\}\in C^{1,0}(\frU,N)$, $\{B_{a}\}\in C^{0,1}(\frU,\frv)$, and $\{m_{a}\}\in C^{0,0}(\frU,M)$ such that on relevant non-empty intersections of coordinate patches
  \begin{equation}
   n_{ab}\,:\,m_b\ \Rightarrow\ m_a~,~~~ 
   n_{ab}\circ n_{bc}\ =\ n_{ac}~,
  \end{equation}
and\footnote{Here, the operations $\circ$ are defined in a detailed discussion of these relations in Appendix \ref{app:B}.}
  \begin{equation}\label{eq:Deligne-cocycle-1}
   B_b\ =\ (n_{ab}^{-1}\circ B_a\circ n_{ab})\circ(n_{ab}^{-1}\circ(-\dd n_{ab}))~.
  \end{equation} 
  \item A \uline{degree-$2$ Deligne cocycle} consists of elements $\{n_{abc}\}\in C^{2,0}(\frU,N)$ and $\{m_{ab}\}\in C^{1,0}(\frU,M)$ such that on the relevant non-empty intersections of coordinate patches
    \begin{subequations}\label{eq:Deligne-cocycle-2}
  \begin{equation}
  \begin{aligned}
            n_{abc}\,:\, m_{ab}\otimes  m_{bc}\ \Rightarrow\ m_{ac}~,\kern2.5cm\\
     n_{acd}\circ (n_{abc}\otimes \id_{m_{cd}})\circ  \sfa^{-1}_{m_{ab},m_{bc},m_{cd}}\ =\ 
   n_{abd}\circ  (\id_{m_{ab}}\otimes  n_{bcd})~,
   \end{aligned}
  \end{equation}
  elements  $\{A_{a}\}\in C^{0,1}(\frU,\frw)$ and $\{B_{a}\}\in C^{0,2}(\frU,\frv)$  such that 
    \begin{equation}
   \dd A_{a}+A_{a}\otimes A_{a}+\sfs(B_{a})\ =\ 0~,
  \end{equation}
  and elements  $\{\Lambda_{ab}^0\}\in C^{1,1}(\frU,\frv)$ such that 
  \begin{equation}
       \Lambda_{ab}^0\,:\, A_b-(m_{ab}\otimes A_a)\otimes\overline{m}_{ab}+\dd m_{ab}\otimes \overline{m}_{ab}\ \Rightarrow\ 0~,
  \end{equation}
  or, equivalently, 
   \begin{equation}
       \Lambda_{ab}\,:\, A_b\otimes m_{ab}\ \Rightarrow\ m_{ab}\otimes A_a-\dd m_{ab}
  \end{equation}
with
    \begin{equation}
\begin{aligned}
  \Lambda_{ab}\ &:=\ \sfa_{m_{ab}\otimes A_a+\dd m_{ab},\overline{m}_{ab},m_{ab}}\,\circ\\
  &\kern1cm \circ
  \left\{\left[ \left(\Lambda^0_{ab}+\id_{(m_{ab}\otimes A_a)\otimes\overline{m}_{ab}-\dd m_{ab}\otimes\overline{m}_{ab}} \right)\circ\sfa_{A_b,m_{ab},\overline{m}_{ab}} \right]\otimes\id_{m_{ab}}\right\}\circ\\
  &\kern1cm \circ \sfa^{-1}_{A_b\otimes m_{ab},\overline{m}_{ab},m_{ab}}
  \end{aligned}
\end{equation}
such that 
   \begin{equation}
  \begin{aligned}
   &\Lambda_{cb}\circ (\id_{A_{b}}\otimes n_{bac})\circ \sfa_{{A_b},{m_{ba}},{m_{ac}}}\ =\\
   &\hspace{0.5cm}\ =\ (n_{bac}\otimes\id_{A_{c}}-\dd n_{bac})\circ \big[\sfa^{-1}_{{m_{ba}},{m_{ac}},A_{c}}-\id_{\dd({m_{ba}}\otimes{m_{ac}})}\big] \circ\\
   &\hspace{1.5cm}\circ(\id_{{m_{ba}}}\otimes\Lambda_{ca}- \id_{\dd{m_{ba}}\otimes{m_{ac}}})\circ(\sfa_{{m_{ba}},{A_{c}}, m_{ac}}-\id_{\dd{m_{ba}}\otimes{m_{ac}}})\circ(\Lambda_{ab}\otimes\id_{m_{ac}})~,
  \end{aligned}
  \end{equation}
  and
  \begin{equation}
   \begin{aligned}
 \kern-20pt  B_b\otimes \id_{m_{ab}}\ &=\ \mu(A_b,A_b,m_{ab})+ \big[\id_{m_{ab}}\otimes B_a+\mu(m_{ab},A_a,A_a)\big]\,\circ\\
    &\kern2.5cm \circ\, \big[-\dd\Lambda_{ab}-\Lambda_{ab}\otimes\id_{A_a}-\mu(A_b,m_{ab},A_a)\big]\,\circ\\
      &\kern2.5cm \circ\,\big[-\id_{\sfs(\dd\Lambda_{ab})}-\id_{A_b}\otimes(\Lambda_{ab}+\id_{\dd m_{ab}})\big]~.    
   \end{aligned}
  \end{equation}
  \end{subequations}
 \end{enumerate}
\end{definition}

\noindent
Furthermore, we need to state what we would like to understand by Deligne coboundary transformations.

\begin{definition}\label{def:Deligne-coboundaries}
 Two degree-$p$ Deligne cocycles are called \uline{cohomologous} or \uline{equivalent} if and only if there is a degree-$(p-1)$ Deligne cochain relating both. In more detail, we define the following:
\begin{enumerate}[(i)]\setlength{\itemsep}{-1mm}
  \item Two degree-$1$ Deligne cocycles $(\{n_{ab}\},\{B_a\},\{m_a\})$ and $(\{\tilde n_{ab}\},\{\tilde B_a\},\{\tilde m_a\})$ are called cohomologous if and only if there is a degree-$0$ Deligne cochain $\{n_a\}\in C^{0,0}(\frU,N)$ such that on the relevant non-empty intersections of coordinate patches
  \begin{equation}
    n_a\,:\,\tilde m_a\ \Rightarrow\ m_a \eand
    n_{ab}\ =\ n_a\circ \tilde n_{ab} \circ n_b^{-1}
    \end{equation}
    and
 \begin{equation}   
   \tilde B_a\ =\ (n_{a}^{-1}\circ B_a\circ n_{a})\circ(n_{a}^{-1}\circ(-\dd n_{a}))~.
  \end{equation}
  \item Two degree-$2$ Deligne cocycles $(\{m_{ab}\},\{n_{abc}\},\{A_a\},\{B_a\},\{\Lambda_{ab}^0\})$ and $(\{\tilde m_{ab}\},\{\tilde n_{abc}\}$, $\{\tilde A_a\},\{\tilde B_a\},\{\tilde\Lambda_{ab}^0\})$ are called cohomologous if and only if there is a degree-1  Deligne cochain $(\{n_{ab}\},\{\Lambda_{a}\},\{m_a\})$ such that on the relevant non-empty intersections of coordinate patches
  \begin{subequations}\label{eq:Deligne-coboundaries-2}
  \begin{equation}
  \begin{aligned}
      n_{ab}\,:\,\tilde m_{ab}\otimes m_b \ \Rightarrow\   m_{a}\otimes m_{ab}~,\kern2.5cm\\
      n_{ac}\circ  (\tilde n_{abc}\otimes \id_{m_c})\ =\ (\id_{m_a}\otimes n_{abc})\circ \sfa_{m_a, m_{ab}, m_{bc}}\circ (n_{ab}\otimes \id_{ m_{bc}})\,\circ  \\
  \circ \,\, \sfa^{-1}_{\tilde m_{ab},m_b,\tilde  m_{bc}}\circ (\id_{\tilde m_{ab}}\otimes  n_{bc})\circ \sfa_{\tilde m_{ab},\tilde m_{bc},m_c}~,
  \end{aligned}
  \end{equation}
  and 
  \begin{equation}
   \Lambda_{a}^0\,:\,\tilde A_a -  (m_{a}\otimes A_a)\otimes\overline{m}_{a}+\dd m_{a}\otimes\overline{m}_{a}\ \Rightarrow\ 0~,
  \end{equation}
   or, equivalently, 
   \begin{equation}
       \Lambda_{a}\,:\, \tilde A_a\otimes m_{a}\ \Rightarrow\ m_{a}\otimes A_a-\dd m_{a}
  \end{equation}
with
    \begin{equation}
\begin{aligned}
  \Lambda_{a}\ &:=\ \sfa_{m_{a}\otimes A_a+\dd m_{a},\overline{m}_{a},m_{a}}\,\circ\\
  &\kern1cm \circ
  \left\{\left[ \left(\Lambda^0_{a}+\id_{(m_{a}\otimes A_a)\otimes\overline{m}_{a}-\dd m_{a}\otimes\overline{m}_{a}} \right)\circ\sfa_{\tilde A_a,m_{a},\overline{m}_{a}} \right]\otimes\id_{m_{a}}\right\}\circ\\
  &\kern1cm \circ \sfa^{-1}_{\tilde A_a\otimes m_{a},\overline{m}_{a},m_{a}}
  \end{aligned}
\end{equation}
such that 
  \begin{equation}
   \begin{aligned}
 \kern-20pt \tilde B_a\otimes \id_{m_a}\ &=\ \mu(\tilde A_a,\tilde A_a,m_a)+ \big[\id_{m_a}\otimes B_a+\mu(m_a,A_a,A_a)\big]\,\circ\\
    &\kern2.5cm \circ\, \big[-\dd\Lambda_a-\Lambda_a\otimes\id_{A}-\mu(\tilde A_a,m_a,A_a)\big]\,\circ\\
      &\kern2.5cm \circ\,\big[-\id_{\sfs(\dd\Lambda_a)}-\id_{\tilde A_a}\otimes(\Lambda_a+\id_{\dd m_a})\big]~,    
   \end{aligned}
  \end{equation}
 \begin{equation}\label{eq:gaugetraforlambdaab}
   \begin{aligned}
   &\kern-.5cm(\tilde \Lambda_{ba}\otimes \id_{m_b})\circ \sfa^{-1}_{\tilde{A}_a,\tilde m_{ab},{m_b}}\circ (\id_{\tilde{A}_a}\otimes n_{ab})\circ \sfa_{\tilde{A}_a,m_a,m_{ab}}\ =\ \\
   &\kern-.5cm=\ (\sfa^{-1}_{\tilde m_{ab},\tilde A_b,m_b}-\id_{\dd\tilde{m}_{ab}\otimes m_b})\circ(\id_{\tilde{m}_{ab}}\otimes\Lambda^{-1}_{b}-\id_{\dd\tilde{m}_{ab}\otimes m_b})\circ(\sfa_{m_{ab},m_b,A_b}-\id_{\dd(\tilde{m}_{ab}\otimes m_b})\,\circ\\ 
   &\circ(n_{ab}^{-1}\otimes\id_{A_b}-\dd n_{ab}^{-1})\circ(\sfa^{-1}_{m_a,m_{ab},A_b}-\id_{\dd(m_a\otimes m_{ab})})\circ (\id_{m_a}\otimes\Lambda_{ba}-\id_{\dd m_a\otimes m_{ab}})\,\circ\\ 
  &\circ(\sfa_{m_a,A_a,m_{ab}}-\id_{\dd m_a\otimes m_{ab}})\circ(\Lambda_a\otimes\id_{m_{ab}})~.
  \end{aligned}
  \end{equation}
  \end{subequations}
 \end{enumerate}
\end{definition}

Note that there are further equivalences between Deligne coboundaries arising from modification transformations. These are not relevant for our discussion of Deligne cohomology and we therefore do not wish to go into any further detail.

Let us end this section by briefly commenting on the interpretation of elements of Deligne cohomology sets. The first case of degree-0 Deligne cocycles is readily understood. A degree-0 Deligne cocycles describes an $N$-valued function on $X$, which could be regarded as a principal $0$-bundle.

The case of Deligne 1-cocycles is slightly more involved. If $N$ is a group, then a degree-1 Deligne cocycle defines a principal (1-)bundle with connection one-form $B$ and a preferred section $m$. This data was called a crossed module bundle, from which crossed module bundle gerbes were constructed in \cite{Aschieri:2003mw}, see also \cite{Jurco:2005qj,Jurco:2009px}. Recall that an Abelian bundle $(p+1)$-gerbe over a manifold $X$ can be constructed from the notion of an Abelian bundle $p$-gerbe, by considering a surjective submersion $Y\rightarrow X$ together with Abelian bundle $p$-gerbes over $Y\times_X Y$. The analogous construction for crossed module bundle gerbes starts from a crossed module bundle. If $N$ is not a group, then a Deligne 1-cocycle describes a 2-group principal bundle, which is a special form of a groupoid principal bundle. Considering 2-group principal bundles over $Y\times_X Y$ yields then to 2-group bundle gerbes or the principal 2-bundles described by Deligne 2-cocycles.

A degree-2 Deligne cocycle describes a semistrict principal 2-bundle with connective structure. Again, gauge equivalence is captured by the cohomology. To study such Deligne 2-cocycles further, it is useful to introduce the curvature 3-form, apart from the 2-form fake curvature \eqref{eq:fakecurv} that vanishes; see also Proposition \ref{prop:hMC-eqn-and-gauge}.

\begin{definition} Let $(\{A_a\},\{B_a\},\{\Lambda_{ab}\})$ be a connective structure on a semistrict principal 2-bundle $\Phi$. The associated \uline{3-form curvature} is defined as follows:
\begin{equation}
\begin{aligned}
 H_{a}\ & :=\ \dd B_{a}+\id_{A_{a}}\otimes B_{a}-B_a\otimes\id_{A_a}+\mu(A_{a},A_{a},A_{a})~.
\end{aligned}
\end{equation}
\end{definition}

\section{Application: Penrose--Ward transform}\label{sec:PenroseWard}

As an application of the theory of principal 2-bundles which we have developed in the previous sections, we now show how to generalise the results of \cite{Saemann:2012uq}. Specifically,  \cite{Saemann:2012uq} established a Penrose--Ward transform that yields a bijection between holomorphic principal 2-bundles with strict structure 2-group over a twistor space and non-Abelian self-dual tensor fields on six-dimensional flat space-time. We can now replace the strict principal 2-bundles by semistrict ones in this construction.

In the following, we denote by $\CO_X$ the sheaf of holomorphic functions and by $\Omega^p_X$ the sheaf of {\it holomorphic} differential $p$-forms on a complex (super)manifold $X$.

\subsection{Supertwistor space}

The twistor space $P^6$ underlying chiral field theories on flat complexified six-dimensional space-time $\FC^6$ is the moduli space of $\alpha$-planes or self-dual 3-planes in $\FC^6$. This twistor space has been described in great detail before \cite{springerlink:10.1007/BF00132253,Hughston:1986hb,Hughston:1979TN,Hughston:1987aa,Hughston:1988nz,0198535651,Saemann:2011nb,Mason:2011nw}, and its supersymmetric extension $P^{6|2n}$ was discussed in \cite{Saemann:2012uq,Mason:2012va,Saemann:2013pca}. We therefore keep our following exposition brief.

The starting point is the chiral superspace $M^{6|8n}:=\FC^{6|8n}$ with $n=0,1,2$. This space can be equipped with the coordinates $(x^{AB},\eta^A_I)$, where $x^{AB}=-x^{BA}$ with $A,B,\ldots=1,\ldots,4$ are the usual Gra{\ss}mann-even  coordinates in spinor notation, $\eta^A_I$ are the Gra{\ss}mann-odd coordinates and $I,J,\ldots=1,\ldots,2n$ are the R-symmetry indices. We may raise and lower the spinor indices using the Levi-Civita symbol, that is, $x_{AB}=\frac12\varepsilon_{ABCD} x^{CD}$ $\Leftrightarrow$ $x^{AB}=\frac12\varepsilon^{ABCD} x_{CD}$. Note that in the real setting, the R-symmetry group of the superconformal group $\sOSp(2,6|2n)$ is
\begin{equation}
 \sSp(n)\ =\ \left\{\begin{array}{ll}
                 \sSp(1)\cong \sSU(2)&\mbox{for $n=1$}\\
                 \sSp(2)\cong\mathsf{USp}(4)\subset \sSp(4,\FC)&\mbox{for $n=2$}
                \end{array}\right.~.
\end{equation}
The group $\sSp(n)$ is defined as the elements of $\sSU(2n)$ leaving an antisymmetric $2n\times 2n$ matrix $\Omega$ invariant, which we can fix according to
\begin{equation}
 \Omega\ =\ \diag(\underbrace{\varepsilon,\ldots,\varepsilon}_{n\rm{-times}})\ewith\varepsilon\ :=\ \left(\begin{array}{cc} 0 & 1 \\ -1& 0 \end{array}\right).
\end{equation}
However, working in the complex setting, we shall employ appropriate complexifications of the above groups. 

We further introduce the superspace derivatives
\begin{equation}
  P_{AB}\ :=\ \der{x^{AB}}\eand D^I_A\ :=\ \der{\eta^A_I}-2\Omega^{IJ}\eta_J^B\der{x^{AB}}~,
\end{equation}
which obey
\begin{equation}
 \{D^I_A,D^J_B\}\ =\ -4\Omega^{IJ}P_{AB}~.
\end{equation}

Next, we let $\PP^3$ be the complex projective 3-space and define the correspondence space $F^{9|8n}:=\FC^{6|8n}\times\PP^3$. It can be equipped with coordinates $(x^{AB},\eta^A_I,\lambda_A)$ where $\lambda_A$ are homogeneous coordinates on $\PP^3$. On the correspondence space, we introduce the twistor distribution, denoted by $D$, which is an integrable distribution of rank $3|6n$ generated by the vector fields
\begin{equation}\label{eq:SupTwisDis}
  D\ :=\ \mbox{span}\{ V^A, V^{IAB}\}\ewith
  V^A\ :=\ \lambda_B\partial^{AB}\eand
  V^{IAB}\ :=\ \tfrac12\varepsilon^{ABCD}\lambda_C D^I_D~.
\end{equation}

The supertwistor space $P^{6|2n}$ is then defined to be the associated leaf space $P^{6|2n}:=F^{9|8n}/D$. We can now establish a twistor correspondence which is captured by the double fibration
\begin{equation}\label{eq:superDoubleFibration}
 \begin{picture}(50,40)
  \put(0.0,0.0){\makebox(0,0)[c]{$P^{6|2n}$}}
  \put(64.0,0.0){\makebox(0,0)[c]{$M^{6|8n}$}}
  \put(34.0,33.0){\makebox(0,0)[c]{$F^{9|8n}$}}
  \put(7.0,18.0){\makebox(0,0)[c]{$\pi_1$}}
  \put(55.0,18.0){\makebox(0,0)[c]{$\pi_2$}}
  \put(25.0,25.0){\vector(-1,-1){18}}
  \put(37.0,25.0){\vector(1,-1){18}}
 \end{picture}
\end{equation}
where $\pi_2$ is the trivial projection, while
\begin{equation}
\pi_1:(x^{AB},\eta^A_I,\lambda_A)\ \mapsto\ (z^A,\eta_I,\lambda_A)\ =\ ((x^{AB}+\Omega^{IJ}\eta^A_I\eta^B_J)\lambda_B,\eta_I^A\lambda_A,\lambda_A)
\end{equation}
contains the so-called incidence relation
\begin{equation}\label{eq:superincidence}
 z^A\ =\ (x^{AB}+\Omega^{IJ}\eta^A_I\eta^B_J)\lambda_B\eand
 \eta_I\ =\ \eta_I^A\lambda_A~.
\end{equation}

This incidence relation yields a geometric correspondence between points $x\in M^{6|8n}$ and complex projective 3-spaces $\hat x=\pi_1(\pi_2^{-1}(x))\hookrightarrow P^{6|2n}$ as well as points $p\in P^{6|2n}$ in twistor space and $3|6n$-superplanes $\pi_2(\pi_1^{-1}(p))\hookrightarrow M^{6|8n}$ which is a supersymmetric extension of a totally null 3-plane in $\FC^6$. It also follows that $P^{6|2n}$ the quadric hypersurface given by the zero locus 
\begin{equation}\label{eq:superquadric}
 z^A\lambda_A-\Omega^{IJ}\eta_I\eta_J\ =\ 0
\end{equation}
inside the total space of the holomorphic fibration $\FC^{4|2n}\otimes \CO_{\PP^3}(1)\rightarrow \PP^3$ with fibre coordinates $z^A$ and $\eta_I$ as well as base coordinates $\lambda_A$.

\begin{rem}
In our subsequent discussion, we shall always choose the standard Stein cover $\hat\frU=\{\hat U_a\}$ on the twistor space $P^{6|2n}\to\PP^3$ (generated by the standard Stein cover on $\PP^3$) and the induced cover $\frU':=\{U'_a:=\pi_1^{-1}(U_a)\}$ on the correspondence space $F^{9|8n}$, respectively.
\end{rem}

\subsection{Penrose--Ward transform}

To formulate the Penrose--Ward transform, we first need to introduce a few basic notions. In particular, we will need the sheaf of holomorphic relative differential $p$-forms, denoted by $\Omega^p_{\pi_1}$, on $F^{9|8n}$ along the fibration $\pi_1:F^{9|8n}\to P^{6|2n}$. It is defined by the short exact sequence
\begin{equation}\label{eq:RelOneForms2}
 0\ \longrightarrow\ \pi_1^*\Omega^1_{P^{6|2n}}\wedge\Omega^{p-1}_{F^{9|8n}}\ \longrightarrow\ \Omega^p_{F^{9|8n}}\ \longrightarrow\ \Omega^p_{\pi_1}\ \longrightarrow\ 0~.
\end{equation}
 In addition, if $\mbox{pr}_{\pi_1}\!:\Omega_{F^{9|8n}}^p\to \Omega^p_{\pi_1}$ denotes the quotient mapping, we can define the relative exterior derivative $\dd_{\pi_1}$ by setting
\begin{equation}
 \dd_{\pi_1}\ :=\ \mbox{pr}_{\pi_1}\circ\dd\,:\, \Omega^p_{\pi_1}\ \to\ \Omega^{p+1}_{\pi_1}~,
\end{equation}
where $\dd$ denotes the usual holomorphic exterior derivative on the correspondence space. In the local coordinates $(x^{AB},\eta^A_I\lambda_A)$ on $F^{9|8n}$, $\dd_{\pi_1}$ is presented in terms of the vector fields of the twistor distribution \eqref{eq:SupTwisDis}; see also \eqref{eq:ExplicitRelDer} below. The relative exterior derivative characterises the so-called relative holomorphic de Rham complex, which is the complex that is given in terms of an injective  resolution of the topological inverse $\pi_1^{-1}\CO_{P^{6|2n}}$ of the sheaf $\CO_{P^{6|2n}}$ on the correspondence space $F^{9|8n}$:
\begin{equation}\label{eq:RelDeRhamCom}
 0\ \longrightarrow\ \pi_1^{-1}\CO_{P^{6|2n}}\ \longrightarrow\ \CO_{F^{9|8n}}\ \xrightarrow{\dd_{\pi_1}}\  \Omega^1_{\pi_1}\ \xrightarrow{\dd_{\pi_1}}\  \Omega^2_{\pi_1}\ \xrightarrow{\dd_{\pi_1}}\  \cdots ~.
\end{equation}
Note that $\pi_1^{-1}\CO_{P^{6|2n}}$ consists of those holomorphic functions that are locally constant along the fibres of $\pi_1:F^{9|8n}\to P^{6|2n}$.

Next, let $\Phi'$  be a holomorphic semistrict principal 2-bundles on the correspondence space $F^{9|8n}$, with $\CCG=(M,N)$ as its semistrict structure 2-group. As before, we denote the 2-term $L_\infty$-algebra associated with $\CCG$ by $\frv\overset{\mu_1}{\longrightarrow}\frw$, where $\frw:=T_{\id_e}M$ and $\frv:=\ker(\sft)\subseteq T_{\id_{\id_e}}N$. The bundle $\Phi'$ is described by holomorphic $\CCG$-valued transition functions $(\{m'_{ab}\},\{n'_{abc}\})$ relative to the cover $\frU'$. 

As we shall see momentarily, the Penrose--Ward transform will be based on so-called relative degree-2 Deligne cohomology. For this reason, we wish to equip $\Phi'$ with a relative connective structure and study its behaviour under equivalence transformations.  Concretely, $\Phi'$ is then described by a degree-2 Deligne cocycle\footnote{To simplify notation, we shall suppress the superscript $0$ in the $\Lambda$-part of the cocycle here and in the following.} 
\begin{equation}
(\{m'_{ab}\},\{n'_{abc}\},\{A'_a\},\{B'_a\},\{\Lambda'_{ab}\})
\end{equation}
with $\{m'_{ab}\}\in C^{1,0}_{\pi_1}(\frU',M)$, $\{n'_{abc}\}\in C^{2,0}_{\pi_1}(\frU',N)$, $\{\Lambda'_{ab}\}\in C^{1,1}_{\pi_1}(\frU',\frv)$,  $\{A'_{a}\}\in C^{0,1}_{\pi_1}(\frU',\frw)$, and $\{B'_{a}\}\in C^{0,2}_{\pi_1}(\frU',\frv)$. Here, the subscript `$\pi_1$' indicates that these are relative differential forms.  For instance, the  $\Lambda'_{ab}$ and $A'_a$ take values in $\Omega^1_{\pi_1}\otimes \frv$ and $\Omega^1_{\pi_1}\otimes \frw$, respectively,  while the $B'_a$ take values in $\Omega^2_{\pi_1}\otimes \frv$. In addition, we call the relative connective structure flat whenever, apart from the vanishing of 2-form fake curvature,
\begin{equation}
  \CF'_{a}\ =\ \dd_{\pi_1} A'_{a}+\tfrac{1}{2}\mu_2(A'_{a},A'_{a})-\mu_1(B'_{a})\ =\ 0~,
\end{equation}
inherent to 2-degree Deligne cocycles,  also the 3-form curvature vanishes
\begin{equation}
 H'_{a}\ =\ \dd_{\pi_1} B'_{a}+\mu_2(A'_{a},B'_{a})-\tfrac{1}{3!}\mu_3(A'_{a},A'_{a},A'_{a})\ =\ 0~.
\end{equation}

The final ingredient we shall need is  a holomorphic semistrict principal 2-bundle $\hat \Phi$ on $P^{6|2n}$ with $\CCG=(M,N)$ as its semistrict structure 2-group.  The bundle $\hat \Phi$ is described by holomorphic $\CCG$-valued transition functions $(\{\hat m_{ab}\},\{\hat n_{abc}\})$ relative to the cover $\hat \frU$.  Following Manin \cite{Manin:1988ds}, $\hat \Phi$ will be called $M^{6|8n}$-trivial whenever it is holomorphically trivial on $\hat x=\pi_1(\pi_2^{-1}(x))\hookrightarrow P^{6|2n}$ for all $x\in M^{6|8n}$; see also Definition \ref{def:restriction-2}.

Then we have the following result.

\begin{prop}\label{prop:EquivTwistCorrespBundle}
Consider $\pi_1:F^{9|8n}\to P^{6|2n}$ of the double fibration \eqref{eq:superDoubleFibration}. There is a bijection between 
  \begin{enumerate}[(i)]\setlength{\itemsep}{-1mm}
\item equivalence classes of topologically trivial $M^{6|8n}$-trivial holomorphic semistrict principal 2-bundles on $P^{6|2n}$ and 
\item equivalence classes of holomorphically trivial semistrict principal 2-bundles on $F^{9|8n}$ equipped with a relative connective structure which is globally flat.
\end{enumerate}
\end{prop}

\begin{proof}
(i) $\to$ (ii) Let $\hat \Phi$ be an $M^{6|8n}$-trivial holomorphic semistrict principal 2-bundle on the twistor space $P^{6|2n}$ described by holomorphic transition functions $(\{\hat m_{ab}\},\{\hat n_{abc}\})$.  Furthermore, let $\Phi'=\pi_1^*\hat\Phi$ be its pullback to the correspondence space $F^{9|8n}$; see also Definition \ref{def:pullback-2}. It is described by holomorphic transition functions $(\{m'_{ab}\},\{n'_{abc}\})$ which are annihilated by the relative exterior derivative $\dd_{\pi_1}$. More precisely, it is described by the relative degree-2 Deligne cocycle
\begin{equation}\label{eq:DelCoCyTriv-1}
(\{m'_{ab}=\pi_1^*\hat m_{ab}\},\{n'_{abc}=\pi_1^*\hat n_{abc}\},\{\Lambda'_{ab}=0\},\{A'_a=0\},\{B'_a=0\})~.
\end{equation}

Since $\hat \Phi$ is $M^{6|8n}$-trivial, its pullback $\Phi'$ is holomorphically trivial on all of $F^{9|8n}$. Therefore, there exists a relative degree-2 Deligne cochain relating the degree-2 Deligne cocycle \eqref{eq:DelCoCyTriv-1} to the cocycle
\begin{equation}\label{eq:DelCoCyTriv-2}
(\{m''_{ab}=\id_{e_a}\},\{n''_{abc}=\id_{\id_{e_a}}\},\{\Lambda''_{ab}\neq 0\},\{A''_a\neq0\},\{B''_a\neq 0\})~.
\end{equation}
From \eqref{eq:Deligne-cocycle-2}, we realise that $\Lambda''_{ab}:A''_b-A''_a\Rightarrow 0$ and
\begin{equation}
 \Lambda''_{ac}\ =\ \Lambda''_{ab}+\Lambda''_{bc}~.
\end{equation}

\pagebreak

Hence, $\{\Lambda''_{ab}\}$ is a representative of an element in the Abelian \v Cech cohomology group $H^1(F^{9|8n},\Omega_{\pi_1}^1\otimes \frv)$. This cohomology group, however, vanishes as follows immediately from the arguments presented in  \cite{Saemann:2011nb,Saemann:2012uq} (see also \cite{Mason:2011nw}). Therefore, we have a splitting
\begin{equation}
  \Lambda''_{ab}\ =\ \Lambda''_a-\Lambda''_b\ewith \Lambda''_a\,:\, A'''_{a}-A''_{a}\ \Rightarrow\ 0~,
\end{equation}
where the $A'''_a$ define a  globally defined $\frw$-valued relative 1-form $A'''_{\pi_1}\in H^0(F^{9|8n},\Omega_{\pi_1}^1\otimes \frw)$, that is, $A'''_a=A'''_{\pi_1}|_{U'_a}$ and $A'''_a=A'''_b$ on $U'_a\cap U'_b$. 
Thus, using \eqref{eq:Deligne-coboundaries-2} with $\Lambda''_a$, we see that the degree-2 Deligne cocycle \eqref{eq:DelCoCyTriv-2} is cohomologous to  
\begin{equation}\label{eq:DelCoCyTriv-3}
(\{m'''_{ab}=\id_{e_a}\},\{n'''_{abc}=\id_{\id_{e_a}}\},\{\Lambda'''_{ab}=0\},\{A'''_a\neq 0\},\{B'''_a\neq 0\})~,
\end{equation}
where the $B'''_a$ yield a globally defined $\frv$-valued relative 2-form $B'''_{\pi_1}\in H^0(F^{9|8n},\Omega_{\pi_1}^2\otimes \frv)$, that is, $B'''_a=B'''_{\pi_1}|_{U'_a}$ and $B'''_a=B'''_b$ on $U'_a\cap U'_b$.

Altogether, we have obtained a holomorphically trivial semistrict principal 2-bundle $\Phi'$ on the correspondence space, equipped with a globally defined relative connective structure represented by $(A_{\pi_1},B_{\pi_2})$. As this relative connective structure is pure gauge, its curvatures necessarily vanish, and, therefore, the relative connective structure is globally flat.

(ii) $\to$ (i)  Conversely, starting from a holomorphically trivial semistrict principal 2-bundle $\Phi'$ on the correspondence space represented by a relative degree-2 Deligne cocycle of the form \eqref{eq:DelCoCyTriv-3} with a relative connective structure that is flat, we can use a generalised Poincar\'e lemma \cite{Demessie:2014aa} to find a relative degree-2 Deligne cochain to transform \eqref{eq:DelCoCyTriv-3} into a cocycle of the form \eqref{eq:DelCoCyTriv-2}. This cocycle descends down to twistor space to a relative degree-2 Deligne cocycle of the form \eqref{eq:DelCoCyTriv-1}.
\end{proof}
\noindent Note that there are equivalence transformations acting on the ingredients of this construction. For instance, constructing the degree-2 Deligne cochains explicitly that mediate between the different degree-2 Deligne cocycles amounts to solving Riemann--Hilbert problems whose solutions are not unique. We shall come back to this in Remark \ref{rem:gaugetrafospacetime}.

Next, we write the relative exterior derivative explicitly as
\begin{equation}\label{eq:ExplicitRelDer}
  \dd_{\pi_1}\ =\ e_A V^A+e_{IAB} V^{IAB}\ =\ e_{[A}\lambda_{B]}\partial^{AB}+e^{AB}_I\lambda_A D^I_B~,
\end{equation}
thereby introducing the relative 1-forms $e_A$ and $e_{IAB}=\tfrac{1}{2}\eps_{ABCD}e_I^{CD}$ which are defined dually to $V^A$ and $V^{IAB}$. Notice that since  $\lambda_A V^A=\lambda_A V^{IAB}=0$, these differential 1-forms are defined modulo terms proportional to $\lambda_A$; see also \cite{Saemann:2011nb,Saemann:2012uq} for more details.

\pagebreak

\begin{lemma}\label{lem:RelFormExpansions}
Let  $\alpha_{\pi_1}\in H^0(F^{9|8n},\Omega^1_{\pi_1})$, $\beta_{\pi_1}\in H^0(F^{9|8n},\Omega^2_{\pi_1})$, and $\gamma_{\pi_1}\in H^0(F^{9|8n},\Omega^2_{\pi_1})$. These relative differential forms are then expanded in $\lambda_A$ according  to 
\begin{equation}\label{eq:form_expansions}
\begin{aligned}
 \alpha_{\pi_1}\ &=\ e_{[A}\lambda_{B]}\,  \alpha^{AB}+e^{AB}_I\lambda_A\, \alpha^I_B~,\\
 \beta_{\pi_1}\ &=\ -\tfrac14e_A\wedge e_B\lambda_C\, \varepsilon^{ABCD} \beta_D{}^{E}\lambda_{E}+\tfrac12 e_A\lambda_B\wedge e^{EF}_I\lambda_E\,\varepsilon^{ABCD}\,\beta_{CD}{}^I_F~+\\
   &\kern1cm+\tfrac12 e^{CA}_I\lambda_C\wedge e^{DB}_J\lambda_D\, \beta^{IJ}_{AB}~,\\
  \gamma_{\pi_1}\ &=\  -\tfrac13 e_A\wedge e_B\wedge e_C\lambda_D\varepsilon^{ABCD}\,\gamma^{EF}\lambda_E\lambda_F~+\\
 &\kern1cm -\tfrac14 e_A\wedge e_B\lambda_C\, \varepsilon^{ABCD}\wedge e^{EF}_I\lambda_E\, (\gamma_{D}{}^G{}^I_F)_0\lambda_G~+\\
 &\kern1cm + \tfrac14 e_A\lambda_B\wedge e^{EF}_I\lambda_E\wedge e^{GH}_J\lambda_G\,\varepsilon^{ABCD}\,(\gamma_{CD}{}^{IJ}_{FH})_0~+\\
 &\kern1cm +\tfrac16 e^{DA}_I\lambda_D\wedge e^{EB}_J\lambda_E\wedge e^{FC}_K\lambda_F\, \gamma^{IJK}_{ABC}~,
 \end{aligned}
\end{equation}
where the coefficient functions depend only on the superspace coordinates $(x^{AB},\eta^A_I)\in M^{6|8n}$. The component $(\gamma_{A}{}^B{}^I_C)_0$ is the totally trace-less part of $\gamma_{A}{}^B{}^I_C$ while $(\gamma_{AB}{}^{IJ}_{CD})_0$ denotes the part of $\gamma_{AB}{}^{IJ}_{CD}$ that vanishes under contraction with $\varepsilon^{ABCD}$. 
\end{lemma}

\begin{proof}
This is a direct consequence of the explicit form of direct images of the sheaves $\Omega^1_{\pi_1}$ and $\Omega^2_{\pi_1}$ under the projection $\pi_2:F^{9|8n}\to M^{6|8n}$. See references \cite{Saemann:2011nb,Saemann:2012uq} for a detailed derivation.
\end{proof}

\begin{rem}\label{rem:CommentsForms}
Note that differential 1-, 2- and 3-forms $\alpha$, $\beta$, and $\gamma$ on chiral superspace $M^{6|8n}$ have components
\begin{equation}
  \big(\alpha_{AB}, \alpha^I_B\big)~,~~~\big(\beta_A{}^B,\beta_{AB}{}^I_C, \beta^{IJ}_{AB}\big)~,\eand 
  \big(\gamma_{AB},\gamma^{AB},\gamma_{A}{}^B{}^I_C,\gamma_{AB}{}^{IJ}_{CD},\gamma^{IJK}_{ABC}\big)~,
\end{equation}
where $\gamma_A{}^B{}^I_C$ is traceless over the AB indices. By virtue of Lemma \ref{lem:RelFormExpansions}, we realise that all of these components for the 1- and 2-forms and some of these components for the 3-form appear in the expansion of relative 1-, 2- and 3-forms $\alpha_{\pi_1}$, $\beta_{\pi_1}$, and $\gamma_{\pi_1}$ on the correspondence $F^{9|8n}$. Note further that the components $(\gamma_{AB},\gamma^{AB})$ represent the self-dual and anti-self dual parts of a Gra{\ss}mann-even differential 3-form $\gamma$ on $M^{6|0}$.

\end{rem}

These considerations then enable us to prove the following Penrose--Ward transform.

{\theorem\label{th:Penrose-Ward}
Consider the double fibration  \eqref{eq:superDoubleFibration}. There is a bijection between
  \begin{enumerate}[(i)]\setlength{\itemsep}{-1mm}
\item equivalence classes of topologically trivial $M^{6|8n}$-trivial holomorphic semistrict principal 2-bundles on $P^{6|2n}$ and
\item gauge equivalence classes of (complex holomorphic) solutions to the constraint equations 
\begin{subequations}\label{eq:Constraint-1}
\begin{equation}
  \CF_A{}^B\ =\ 0~,\quad \CF_{AB}{}^I_C\ =\ 0~,\eand \CF^{IJ}_{AB}\ =\ 0~,
\end{equation}
and
             \begin{equation}
\begin{aligned}
  H^{AB}\ &=\ 0~,\\
  H_{A}{}^B{}^I_C\ &=\ \delta^B_C\psi^I_A-\tfrac14\delta^B_A\psi^I_C~,\\
  H_{AB}{}^{IJ}_{CD}\ &=\ \varepsilon_{ABCD}\phi^{IJ}~,\\
  H^{IJK}_{ABC}\ &=\ 0
  \end{aligned}
\end{equation}
\end{subequations}
on chiral superspace $M^{6|8n}$. Here, the curvatures read explicitly as
\begin{subequations}\label{eq:Constraint-2}
\begin{equation}
\begin{aligned}
  \CF_A{}^B\ &=\ \partial^{BC} A_{CA}-\partial_{CA}A^{BC}+\mu_2(A^{BC},A_{CA})-\mu_1(B_A{}^B)~,\\
  \CF_{AB}{}^I_C\ &=\ ~\partial_{AB}A^I_C-D^I_CA_{AB}+\mu_2(A_{AB},A^I_C)-\mu_1(B_{AB}{}^I_C)~,\\
  \CF^{IJ}_{AB}\ &=\ D^I_AA^J_B+D^J_B A^I_A+\mu_2(A^I_A,A^J_B)+4\Omega^{IJ}A_{AB}-\mu_1(B^{IJ}_{AB})
  \end{aligned}
\end{equation}
and
\begin{equation}
\begin{aligned}
  H_{AB}\ &=\ \nabla_{C(A}B_{B)}{}^C+\mu_3(A_{C(A},A^{CD},A_{B)D})~,\\
  H^{AB}\ &=\ \nabla^{C(A}B_C{}^{B)}+\mu_3(A^{C(A},A_{CD},A^{B)D})~,\\
  H_{A}{}^B{}^I_C\ &=\ \nabla^I_CB_A{}^B-\nabla^{DB}B_{DA}{}^I_C+\nabla_{DA}B^{DB}{}^I_C-\mu_3(A^I_C,A^{BD},A_{DA})~,\\
  H_{AB}{}^{IJ}_{CD}\ &=\ \nabla_{AB}B^{IJ}_{CD}-\nabla^I_C B_{AB}{}^J_D-\nabla^J_D B_{AB}{}^I_C\,\,-\\
  &\kern1cm-2\Omega^{IJ}(\varepsilon_{ABF[C} B_{D]}{}^F-\varepsilon_{CDF[A} B_{B]}{}^F)-\mu_3(A_{AB},A^I_C,A^J_D)~,\\
  H^{IJK}_{ABC}\ &=\ \nabla^I_A B^{JK}_{BC}+\nabla^J_BB^{IK}_{AC}+\nabla^K_C B^{IJ}_{AB}\,+\\
     &\kern1cm+4\Omega^{IJ}B_{AB}{}^K_C+4\Omega^{IK}B_{AC}{}^J_B+4\Omega^{JK}B_{BC}{}^I_A-\mu_3(A^I_A,A^J_B,A^K_C)~.
  \end{aligned}
\end{equation}
\end{subequations}
\end{enumerate}
}

\noindent
Before proving the theorem, let us make a few comments. The fields $\psi^I_A$ are Gra{\ss}mann-odd spinor fields while the fields $\phi^{IJ}$ are Gra{\ss}mann-even scalar fields. The condition $H^{AB}=0$ implies that the Gra{\ss}mann-even part of the 3-form $H$ is self-dual, cf.\ Remark \ref{rem:CommentsForms}. Altogether, $(H_{AB},\psi^I_A,\phi^{IJ})$ constitutes an $\CN=(n,0)$ tensor multiplet for $n=0,1,2$. Note that only for $n=2$, the condition $\phi^{IJ}\Omega_{IJ}=0$ arises, so that we always find the correct number of scalar fields. See also Saemann \& Wolf \cite{Saemann:2012uq,Saemann:2013pca,Saemann:2011nb} for more details on this point.

\vspace{10pt}
\noindent
{\it Proof of theorem:}
(i) $\to$ (ii) By virtue of Proposition \ref{prop:EquivTwistCorrespBundle}, topologically trivial $M^{6|8n}$-trivial holomorphic semistrict principal 2-bundles on twistor space correspond to holomorphically trivial semistrict principal 2-bundles on $F^{9|8n}$ equipped with a relative connective structure which is globally flat and vice versa. Therefore, such a bundle on twistor space yields a globally defined relative connective structure $(A_{\pi_1},B_{\pi_1})\in H^0(F^{9|8n},\Omega^1_{\pi_1}\otimes \frw)\oplus H^0(F^{9|8n},\Omega^2_{\pi_1}\otimes \frv)$ on the correspondence space which is flat, that is,
\begin{equation}
\begin{aligned}
 \CF_{\pi_1}\ &=\ \dd_{\pi_1} A_{\pi_1}+\tfrac{1}{2}\mu_2(A_{\pi_1},A_{\pi_1})-\mu_1(B_{\pi_1})\ =\ 0~,\\
 H_{\pi_1}\ &=\ \dd_{\pi_1} B_{\pi_1}+\mu_2(A_{\pi_1},B_{\pi_1})-\tfrac{1}{3!}\mu_3(A_{\pi_1},A_{\pi_1},A_{\pi_1})\ =\ 0~.
\end{aligned}
\end{equation}
Upon using \eqref{eq:ExplicitRelDer} and the expansions given in Lemma \ref{lem:RelFormExpansions}, we arrive at the constraint equations \eqref{eq:Constraint-1} and \eqref{eq:Constraint-2} after a few algebraic manipulations. 

(ii) $\to$ (i) The converse is also readily derived. Given a solution to  \eqref{eq:Constraint-1} and \eqref{eq:Constraint-2}, by Lemma \ref{lem:RelFormExpansions} we can always construct a globally defined relative connective structure $(A_{\pi_1},B_{\pi_1})\in H^0(F^{9|8n},\Omega^1_{\pi_1}\otimes \frw)\oplus H^0(F^{9|8n},\Omega^2_{\pi_1}\otimes \frv)$ on the correspondence space which is flat. This defines a holomorphically trivial semistrict principal 2-bundles on $F^{9|8n}$ equipped with a flat relative connective structure. The construction of a topologically trivial $M^{6|8n}$-trivial holomorphic semistrict principal 2-bundles on twistor space then follows directly from Proposition \ref{prop:EquivTwistCorrespBundle}.
\hfill $\Box$

\begin{rem}\label{rem:gaugetrafospacetime}
Finally, we would like to mention that the gauge transformations of the connective structure $(A_{AB},A^I_A,B_A{}^B,B_{AB}{}^I_C,B^{IJ}_{AB})$ on $M^{6|8n}$ follow directly from the large class of equivalence relations between relative Deligne 2-cocycles of the form \eqref{eq:DelCoCyTriv-3} on $F^{9|8n}$. The Deligne 1-cochains parametrising the equivalence relations between relative Deligne 2-cocycles of the form \eqref{eq:DelCoCyTriv-3} are expressed in terms of $p\in H^0(F^{9|8n},M)$ and $\Lambda_{\pi_1}\in H^0(F^{9|8n},\Omega^1_{\pi_1}\otimes \frv)$. Their $\lambda_A$-expansions read as
\begin{equation}
p\ =\ p(x,\eta)\eand \Lambda_{\pi_1}\ =\  e_{[A}\lambda_{B]}\,  \Lambda^{AB}(x,\eta)+e^{AB}_I\lambda_A\, \Lambda^I_B(x,\eta)~.
\end{equation}
Such Deligne 1-cochains are therefore described by $p(x,\eta)$, $\Lambda_{AB}(x,\eta)$, and $\Lambda^I_A(x,\eta)$ which themselves form  a Deligne 1-cochain encoding an equivalence relation between Deligne 2-cocycles on the chiral superspace $M^{6|8n}$. The gauge transformations are then simply of the form given in Proposition \ref{prop:semistrict-gauge-trafos}.
\end{rem}

\appendices

\subsection{Strong homotopy Lie algebras}\label{app:A}

In this appendix, we recall the definitions of strong homotopy Lie algebras and their Chevalley--Eilenberg algebras as well as the homotopy Maurer--Cartan equation together with their infinitesimal gauge symmetries.

Recall that a permutation $\sigma$ of $i+j$ elements is called an $(i,j)$-unshuffle, if the first $i$ and the last $j$ images of $\sigma$ are ordered: $\sigma(1)<\cdots<\sigma(i)$ and $\sigma(i+1)<\cdots<\sigma(i+j)$. Moreover, the graded Koszul sign $\chi(\sigma;x_1,\ldots,x_n)$ of elements $x_i$ of a graded vector space is defined via the equation
\begin{equation}
 x_1\wedge \cdots \wedge x_n\ =\ \chi(\sigma;x_1,\ldots,x_n)\,x_{\sigma(1)}\wedge \cdots \wedge x_{\sigma(n)}
\end{equation}
in the free graded algebra $\wedge (x_1,\ldots,x_n)$, where $\wedge$ is considered graded antisymmetric. 

\begin{definition} \cite{Lada:1992wc,Lada:1994mn,0821843621} An \uline{$L_\infty$-algebra} or \uline{strong homotopy Lie algebra} is a $\RZ$-graded vector space $L=\oplus_{p\in\RZ} L_p$ endowed with $n$-ary multilinear totally antisymmetric products $\mu_n$, $n\in\NN^*$, of degree $2-n$, that satisfy the homotopy Jacobi identities
\begin{equation}\label{eq:homotopyJacobi}
 \sum_{i+j=n}\sum_\sigma\chi(\sigma;x_1,\ldots,x_n)(-1)^{i\cdot j}\mu_{i+1}(\mu_j(x_{\sigma(1)},\ldots,x_{\sigma(j)}),x_{\sigma(j+1)},\ldots,x_{\sigma(i+j)})\ =\ 0
\end{equation}
for all $n\in \NN^*$, where the sum over $\sigma$ is taken over all $(i,j)$-unshuffles.  
\end{definition}

An alternative sign convention is given in \cite{Stasheff:1997fe}, which is obtained from the above one by inverting the signs of all elements of $L$. The homotopy Jacobi identities \eqref{eq:homotopyJacobi} then read as 
\begin{equation}\label{eq:homotopyJacobi_1}
 \sum_{i+j=n}\sum_\sigma\chi(\sigma;x_1,\ldots,x_n)\mu_{j+1}(\mu_i(x_{\sigma(1)},\ldots,x_{\sigma(i)}),x_{\sigma(i+1)},\ldots,x_{\sigma(i+j)})\ =\ 0~.
\end{equation}

A simple example of an $L_\infty$-algebra is a differential graded Lie algebra, for which $\mu_1$ is the differential, $\mu_2$ is the Lie bracket and $\mu_i=0$ for $i\geq 3$. Another example of an $L_\infty$-algebra is given by the 2-term $L_\infty$-algebras of Definition \ref{def:2-term-SH-algebra}.

\begin{definition}
 A \uline{$\RZ$-graded coalgebra} is a $\RZ$-graded vector space $L=\oplus_{p\in\RZ} L_p$ endowed with a \uline{coproduct}  $\Delta: A\rightarrow A\otimes A$ of degree 0 such that $(\unit\otimes \Delta)\circ \Delta=(\Delta\otimes \unit)\circ \Delta$. A \uline{coderivation} of degree $k$ on a coalgebra $C$ is a linear map $D:C\rightarrow C$ of degree $k$ such that $\Delta\circ D=(\unit\otimes D+D\otimes \unit)\circ \Delta$. A \uline{differential graded coalgebra} is a graded coalgebra endowed with a coderivation $D$ of degree 1 such that $D\circ D=0$.
\end{definition}

Each $L_\infty$-algebra yields naturally a differential graded coalgebra. We start from an $L_\infty$-algebra $L$, and shift the degree of each element by $-1$, arriving at $L[-1]$. The symmetric tensor algebra $\odot^\bullet L[-1]$ of $L[-1]$ can be regarded as a graded coalgebra with coproduct
\begin{equation}
 \Delta(\ell_1\odot \cdots \odot \ell_n)\ :=\ \sum_{i=0}^n\sum_\sigma (\ell_{\sigma(1)}\otimes\cdots\otimes\ell_{\sigma(i)})\otimes (\ell_{\sigma(i+1)}\otimes\cdots\otimes\ell_{\sigma(n)})~,
\end{equation}
where the sum over $\sigma$ is taken over all $(i,n-i)$-unshuffles. Note that on $L[-1]$, the higher products $\mu_n$ all have degree $1$ and we can add them to a differential $D$, which acts as $\mu_i$ on $L[-1]^{\otimes i}$ and on higher tensor powers of $L[-1]$ as a coderivation. The property $D\circ D=0$ is then equivalent to the homotopy Jacobi identities \cite{Lada:1992wc}. 

On the other hand, given a commutative differential graded coalgebra, we can derive a corresponding $L_\infty$-algebra. Altogether, we arrive at the following proposition.
\begin{prop}
 An $L_\infty$-algebra is equivalent to a commutative differential graded coalgebra.
\end{prop}

Instead of working with coalgebras, it is usually more convenient to work directly with differential graded algebras. Assuming that the vector subspaces $L_p\subseteq L$ are finite dimensional, we can consider the dual complex $L[-1]^\vee$ to $L[-1]$.

\begin{definition}
 The \uline{Chevalley--Eilenberg algebra} of an $L_\infty$-algebra $L$ is the dual of the differential graded coalgebra $\odot^\bullet L[-1]$. In particular, $\sCE(L):=\odot^\bullet (L[-1]^\vee)$ and the differential $\dd_{\rm CE}:=D^\vee$ is the dual of the differential $D$ in $\odot^\bullet L[-1]$.
\end{definition}

\noindent
It is straightforward to verify the $\sCE(L)$ is indeed a differential graded algebra.

The Chevalley--Eilenberg algebra of a Lie algebra $\frg$ is a differential graded algebra that encodes the Lie bracket via the equation
\begin{equation}
 \dd_{\rm CE}\check\tau^k+\tfrac{1}{2}f_{ij}^k\check\tau^i\wedge \check\tau^j\ =\ 0~,
\end{equation}
where the $\check\tau^i$ form a basis of the dual $\frg^\vee$ of $\frg$ and $f_{ij}^k$ are the structure constants of $\frg$ with respect to the dual basis $(\tau_i)$ with $\check \tau^i(\tau_j)=\delta^i_j$. Evaluated at an element $a\in \frg[-1]$, we have
\begin{equation}
 \dd_{\rm CE}a+\tfrac{1}{2}[a,a]\ =\ 0~,
\end{equation}
the Maurer--Cartan equation of the differential graded algebra. This equation can be generalised to the case of $L_\infty$-algebras.

\begin{definition}
 An element $\phi$ of an $L_\infty$-algebra is called a \uline{homotopy Maurer--Cartan element} whenever it satisfies the \uline{homotopy Maurer--Cartan equation}
\begin{equation}\label{eq:MCeqs}
 \sum_{i=1}^\infty \frac{(-1)^{i(i+1)/2}}{i!}\mu_i(\phi,\ldots,\phi)\ =\ 0~.
\end{equation}
\end{definition}

\begin{theorem}
The homotopy Maurer--Cartan equation is invariant under the following infinitesimal symmetries parameterised by an element $\gamma\in L_0$:
\begin{equation}\label{eq:MCgaugetrafos}
 \phi\rightarrow \phi+\delta \phi\ewith \delta \phi\ =\ \sum_i \frac{(-1)^{i(i-1)/2}}{(i-1)!}\mu_i(\gamma,\phi,\ldots,\phi)~.
\end{equation}
\end{theorem}

\begin{proof}
The general proof of this theorem can be found, for instance, in \cite{Lazaroiu:2001nm}. Here, we give a shortened version for the case $\phi\in L_1$, which is the one we are interested in. We start by computing the homotopy Jacobi identities \eqref{eq:homotopyJacobi} for the tuple $(\gamma,\phi,\ldots,\phi)$, obtaining 
\begin{equation}
\begin{aligned}
  \sum_{i+j=n}\binom{n-1}{j-1}(-1)^{ij}\mu_{i+1}(\mu_j(\gamma,\phi,\ldots,\phi),\phi,\ldots,\phi)+~~&\\
  +\sum_{i+j=n,i\geq 1}\binom{n-1}{j}(-1)^{ij+n-1}\mu_{i+1}(\mu_j(\phi,\ldots,\phi)&,\phi,\ldots,\phi,\gamma)\ =\ 0
\end{aligned}
\end{equation}
or
\begin{equation}
\begin{aligned}
  \sum_{i+j=n}\frac{1}{(j-1)!i!}(-1)^{1+i(n-i)-\frac{n}{2}+\frac{n^2}{2}}\mu_{i+1}(\mu_j(\phi,\ldots,\phi,\gamma),\phi,\ldots,\phi)+~~&\\
  +\sum_{i+j=n,i\geq 1}\frac{(-1)^{1+i(n-i)+n-1-\frac{n}{2}+\frac{n^2}{2}}}{j!(i-1)!}\mu_{i+1}(\mu_j(\phi,\ldots,\phi),\phi,\ldots&,\phi,\gamma)\ =\ 0~.
\end{aligned}
\end{equation}
Next, we note the following identities for $i+j=n$:
\begin{equation}
 \begin{aligned}
 (-1)^{1+i(n-i)+n-1-\frac{n}{2}+\frac{n^2}{2}}\ &=\ (-1)^{i(n-i)+\frac{n}{2}+\frac{n^2}{2}}\ =\ (-1)^{i(i+1)/2+j(j+1)/2}~,\\
 (-1)^{((i+1)(i+2)+j(j-1))/2}\ &=\ (-1)^{1+2i+i^2-\frac{n}{2}-in+\frac{n^2}{2}}\ =\ (-1)^{1+i(n-i)-\frac{n}{2}+\frac{n^2}{2}}~.
 \end{aligned}
\end{equation}
Now we can compute the variation of \eqref{eq:MCeqs} under the transformations \eqref{eq:MCgaugetrafos}:
\begin{equation}\label{eq:delta_MC}
\begin{aligned}
 \delta\left(\sum_{i=1}^\infty \frac{(-1)^{i(i+1)/2}}{i!}\mu_i(\phi,\ldots,\phi)\right)\ &=\ \sum_{i=1}^\infty \frac{(-1)^{i(i+1)/2}}{(i-1)!}\mu_i(\delta\phi,\ldots,\phi)\\
 &\hspace{-4cm}=\ \sum_{i=1}^\infty\sum_{j=1}^\infty\frac{(-1)^{(i(i+1)+j(j-1))/2}}{(i-1)!(j-1)!}\mu_i(\mu_j(\gamma,\phi,\ldots,\phi),\phi,\ldots,\phi)\\
 &\hspace{-4cm}=\ \sum_{n=1}^\infty\sum_{i+j=n}^\infty\frac{(-1)^{((i+1)(i+2)+j(j-1))/2}}{i!(j-1)!}\mu_{i+1}(\mu_j(\gamma,\phi,\ldots,\phi),\phi,\ldots,\phi)\\
 &\hspace{-4cm}=\ -\sum_{n=1}^\infty\sum_{i+j=n,i\geq 1}\frac{(-1)^{i(i+1)/2+j(j+1)/2}}{j!(i-1)!}\mu_{i+1}(\mu_j(\phi,\ldots,\phi),\phi,\ldots,\phi,\gamma)\\
 &\hspace{-4cm}=\ -\sum_{i=1}^\infty\frac{(-1)^{i(i+1)/2}}{(i-1)!}\mu_{i+1}\left(\sum^\infty_{j=1}\frac{(-1)^{j(j+1)/2}	}{j!}\mu_j(\phi,\ldots,\phi),\phi,\ldots,\phi,\gamma\right)\\
 &\hspace{-4cm}=\ 0
\end{aligned}
\end{equation} 
as a consequence of the homotopy Maurer--Cartan equation \eqref{eq:MCeqs}.
\end{proof}

\subsection{Groupoid bundles}\label{app:B}

In this appendix, we present the parameterisation of a functor from the category of supermanifolds to the category of groupoid bundles with preferred section, completing the discussion of Deligne 1-cocycles with values in a semistrict Lie 2-group. Such cocycles arise from functors between the \v Cech groupoid and the Lie 2-group (regarded as a monoidal category). Our discussion follows closely the lines of that in Section \ref{sec:Lie-functor}, and we shall therefore be concise.

We start from $\CG=(M,N)$-valued descent data on surjective submersions $\FR^{0|1}\times X\rightarrow X$, which are represented my $M$-valued 1-cells $\{m_0:=m(\theta_0)\}$ and $N$-valued 2-cells $\{n_{01}:=n(\theta_0,\theta_1)\}$ such that
\begin{equation}
 n_{01}\ :\ m_1\ \Rightarrow m_0\eand n_{01}\circ n_{12}\ =\ n_{02}~.
\end{equation}
Note that $n_{01}=n_{10}^{-1}$ and we can normalise $m_0=m(\theta_0)=\id_e+\alpha \theta_0$ with $\alpha\in \frm[-1]$. These descent data are trivialised by degree-1 Deligne coboundaries $\{n_0:=n(\theta_0)\}$ with
\begin{equation}
 n_0\ : \ m_0\ \Rightarrow \ \id_e\eand n_{01}\ =\ n_0^{-1}\circ n_1~,
\end{equation}
where $n_{0}:=n(\theta_0)=n(0,\theta_0)$. Such a coboundary, and therefore the whole functor under consideration, is parameterised by a $\beta\in\frv=\ker(\sft)\subseteq T_{\id_{\id_e}}N$ according to
\begin{equation}
 n_0\ =\ \id_{\id_e}+\beta \theta_0~,
\end{equation}
and we conclude that $m_0=\id_e+\sfs(\beta)\theta_0$. Equivalence relations on such descent data are described by degree-1 Deligne coboundaries according to
\begin{equation}
 q_0\ : \ \tilde m_0\ \Rightarrow \ m_0\eand \tilde n_{01}\ =\ q_0^{-1}\circ n_{01}\circ q_1~,
\end{equation}
where 
\begin{equation}
 q_0\ =\ q-\dd_{\rm K} q\theta~,~~~q\in N\ewith \sfs(q)\ =\ \sft(q)\ =\ \id_e~.
\end{equation}
The trivialising coboundary between $(\{\tilde n_{01}\},\{\tilde m_0\})$ and the trivial cocycle $(\{\id_{\id_e}\},\{\id_e\})$ is then given by the composition
\begin{equation}
 n'_0\ :=\ n_0\circ q_0~.
\end{equation}
To compare this coboundary with $n_0$, we have to bring it to the form $\tilde n_0=\tilde n(0,\theta_0)$ by a modification transformation. Note that the coboundary relation
\begin{equation}
 \tilde n_{01}\ =\ q_0^{-1}\circ n_0^{-1}\circ n_1\circ q_1\ =\ (n'_0)^{-1}\circ n'_1
\end{equation}
is invariant under the modification transformation
\begin{equation}
 n'_0\ \rightarrow\ \tilde n_0\ =\ o\circ n'_0
\end{equation}
for some $o\in N$. The modification we need here is given by $o=q^{-1}$. Then
\begin{equation}\label{eq:modification-1-Deligne}
 \tilde n_{0}\ =\ q^{-1}\circ n_0\circ q_0\ =\ q^{-1}\circ(\id_{\id_e}+\beta \theta)\circ(q-\dd_{\rm K} q\theta)\ =\ \id_{\id_e}+\tilde \beta \theta~.
\end{equation}
To evaluate the concatenation, we introduce the following linearised forms:
\begin{subequations}\label{eq:new_circ}
\begin{equation}
 q\circ(\id_{\id_e}+\rho \theta)\ :=\ q\circ (\id_{\id_e}+\rho\theta)\eand (\id_{\id_e}+\rho \theta)\circ q\theta\ :=\ (\id_{\id_e}+\rho \theta)\circ q~,
\end{equation}
which implies
\begin{equation}
 (q_1+\rho_1\theta)\circ(q_2+\rho_2\theta)\ =\ q_1\circ q_2+(\rho_1\circ q_2)\circ(q_1\circ \rho_2)\theta~.
\end{equation}
\end{subequations}
for all $q,q_{1,2}\in N$ and $\rho,\rho_{1,2}\in \frn[-1]$. With this notation, equation \eqref{eq:modification-1-Deligne} simplifies to
\begin{equation}
 \tilde n_{0}\ =\ \id_{\id_e}+\tilde \beta \theta\ =\ \id_{\id_e}+(q^{-1}\circ \beta\circ q)\circ (q^{-1}\circ(-\dd_{\rm K} q))\theta~.
\end{equation}
We can now readily read off the cocycle conditions and coboundary relations for the $\{B_a\in C^{0,1}(\frU,\frv)\}$ contained in the degree-1 Deligne cochain with values in a semistrict Lie 2-group.



\end{document}